\documentclass[11pt,draftclsnofoot, onecolumn]{IEEEtran}

\usepackage{graphicx,cite,epsfig,amssymb,amsmath,subfigure,url,stfloats,latexsym}
\usepackage{array}
\usepackage{arydshln}
\usepackage{amsfonts}
\usepackage{pgfplots}
\usepackage{algorithm}
\usepackage[noend]{algpseudocode}
\usepackage{epstopdf}
\usepackage{amsfonts,amsthm}
\usepackage{multirow}
\usepackage{mathrsfs}
\usepackage{subfigure}
\usepackage{subcaption}
\usepackage{subfig}

\usepackage{algorithm}
\usepackage{algpseudocode}

\newtheorem{theorem}{Theorem}

\newtheorem{corollary}{Corollary}
\newtheorem{property}{Property}

\graphicspath{{./}{icons/}}
\tikzset{My Line Style/.style={samples=400}}
\graphicspath{{figures/}}

\usepackage[flushleft]{threeparttable}

\begin{document}

\title{\mbox{}\vspace{0.3cm}\\
\textsc{\huge Finite Field Multiple Access II: \\ from symbol-wise to codeword-wise} 
\vspace{0.5cm}}

\vspace{0.5cm}
\author{\normalsize
Qi-yue~Yu, {\it IEEE Senior Member},
Shi-wen~Lin, and Ting-wei~Yang
\thanks{Q.-Y.~Yu (email: yuqiyue@hit.edu.cn), S.-W.~Lin (email: 23S005089@stu.hit.edu.cn) and T.-W.~Yang (email: yangtingwei@stu.hit.edu.cn) are affiliated with Harbin Institute of Technology, Heilongjiang, China.}
}

\maketitle
\vspace{-0.5in}
\begin{abstract}
A finite-field multiple-access (FFMA) system separates users within a finite field by utilizing different element-pairs (EPs) as virtual resources. The Cartesian product of distinct EPs forms an EP code, which serves as the input to a finite-field multiplexing module (FF-MUX). 
This allows the FFMA technique to reorder the channel coding and multiplexing modules, enabling the superimposed signals to function as codewords that can be decoded by a channel code. This flexibility allows the FFMA system to efficiently support a large number of users with short packet traffic, addressing the finite blocklength (FBL) challenge in multiuser reliable transmission.
Designing EP codes is a central challenge in FFMA systems. In this paper, we construct EP codes based on a bit(s)-to-codeword transformation approach and define the corresponding EP code as a codeword-wise EP (CWEP) code. We then investigate the encoding process of EP codes, and propose unique sum-pattern mapping (USPM) structural property constraints to design uniquely decodable CWEP codes. 
Next, we present the \(\kappa\)-fold ternary orthogonal matrix \({\bf T}_{\rm o}(2^{\kappa}, 2^{\kappa})\) over GF\((3^m)\), where \(m = 2^{\kappa}\), and the ternary non-orthogonal matrix \({\bf T}_{\rm no}(M,m)\) over GF\((3^m)\), for constructing specific CWEP codes.
Based on the proposed CWEP codes, we introduce three FFMA modes: channel codeword multiple access (FF-CCMA), code division multiple access (FF-CDMA), and non-orthogonal multiple access (FF-NOMA). Simulation results demonstrate that all three modes effectively support massive user transmissions with well-behaved error performance. 
\end{abstract}


\newpage
\setcounter{page}{1}
\section{Introduction}
For next-generation wireless communications, one promising application is ultra-massive machine-type communications (um-MTC), which aims to support a large number of users (or devices) with short packet traffic while maintaining robust error performance \cite{6G, MIT_2017}.   
In this context, the \textit{finite-field multiple-access (FFMA)} technique is proposed, which separates multiple users in a finite field using \textit{element-pairs (EPs)}, also referred to as \textit{virtual resource blocks (VRBs)} \cite{FFMA, FFMA_ITW}. 
The Cartesian product of \( J \) distinct EPs forms an EP code, where the EP code acts as the input to a \textit{finite-field multiplexing module (FF-MUX)}, enabling the FFMA technique to interchange the order of channel coding and multiplexing. This approach effectively addresses the finite-block length (FBL) problem in multiuser reliable transmission.

One of the core challenges of the FFMA technique is to construct well-behaved EP codes. In \cite{FFMA}, EP codes are constructed using a straightforward \textit{bit-to-symbol transform} approach. For example, the additive inverse EP (AIEP) code, \( \Psi_{\rm s} \), is constructed over GF(\( p \)), and the orthogonal uniquely-decodable EP (UD-EP) code, \( \Psi_{\rm o,B} \), is constructed over GF(\( 2^m \)).
With the aid of the AIEP code \( \Psi_{\rm s} \), the multiplexing efficiency of an FFMA network can increase by a factor of \( \log_2 (p-1) \) \cite{FFMA}. Additionally, based on the orthogonal UD-EP code \( \Psi_{\rm o,B} \) over GF(\( 2^m \)), we can design \textit{time-division multiple access in finite field (FF-TDMA) system}.

In fact, the performance of an FFMA system is primarily determined by the design of its EP codes. In other words, different EP codes can lead to different multiple-access (MA) systems.
It is well-known that there are various \textit{complex-field multiple-access (CFMA)} techniques that distinguish users by allocating different physical resource blocks, such as TDMA, CDMA (code division multiple access) \cite{Walsh_1964, Walsh_1971, Walsh_1971_2, Walsh_1993, Walsh_1999, FAdachi_2005}, IDMA (interleave division multiple access) \cite{IDMA1, IDMA2, SIDMA}, NOMA (non-orthogonal multiple access) \cite{YChen_2018, QWang_2018, QY_ISJ_2019, SCMA_2013, UD_CDMA1_2012, UD_CDMA2_2012, UD_CDMA3_2014, UD_CDMA4_2016, UD_CDMA5_2018, UD_CDMA6_2019} and others. These classical CFMA techniques play crucial and distinct roles in supporting multiuser transmissions.
For instance, the orthogonal spreading sequences used in CDMA can provide spreading gain and separate users' sequences with a simple correlation detector \cite{Walsh_1999, FAdachi_2005}. The NOMA technique, on the other hand, enhances spectral efficiency (SE) by allowing multiple users to share the same physical resources \cite{YChen_2018}. In general, an NOMA system can separate users at the receiving end using spreading, power-domain, codebooks, and other methods.


It is appealing to extend these classical CFMA techniques into finite fields to support massive user transmissions with short packet traffic. However, this introduces several challenges. For instance, classical Walsh codes \cite{Walsh_1964, Walsh_1971} are typically constructed over the \textit{binary field}, which does not exhibit the \textit{unique sum-pattern mapping (USPM)} structural property as defined in \cite{FFMA}. Fortunately, an EP code can be constructed over various finite fields, not limited to the binary field or its extension fields, thus expanding the design space.

In this paper, we propose a general approach to constructing EP codes, namely the \textit{bit(s)-to-codeword transform}, which broadens the applicability of FFMA systems and supports various modes. We then define two specific types of \textit{codeword-wise} EP codes: the \textit{single codeword EP (S-CWEP)} and the \textit{additive inverse codeword EP (AI-CWEP)} codes. Next, we describe the encoder for a codeword-wise EP code. The input to this encoder can be either a single bit or multiple bits (referred to as a codeword), and the output is a codeword determined by the EP encoder. Specifically, we consider two input modes: the serial mode and the parallel mode of the codeword-wise EP encoder. Furthermore, we investigate the impact of a channel code and present the framework for an FFMA system integrated with channel coding.

To ensure the designed EP codes uniquely decodable, we present \textit{USPM structural property constraints} for both the S-CWEP codes over GF($2^m$) and AI-CWEP codes over GF($3^m$).
If the generator matrix of an S-CWEP code over GF($2^m$) satisfies the USPM constraint, we can obtain a \textit{uniquely decodable S-CWEP (UD-S-CWEP)} code, which can be used to support a \textit{channel codeword multiple-access in finite-field (FF-CCMA)} system.

Then, we present \textit{$\kappa$-fold ternary orthogonal matrix} ${\bf T}_{\rm o}(2^{\kappa}, 2^{\kappa})$ over GF($3^m$) where $m = 2^{\kappa}$, and \textit{ternary non-orthogonal matrix} ${\bf T}_{\rm no}(M,m)$ over GF($3^m$).
Based on the $\kappa$-fold ternary orthogonal matrix ${\bf T}_{\rm o}(2^{\kappa}, 2^{\kappa})$ and its additive inverse matrix ${\bf T}_{\rm o, ai}(2^{\kappa}, 2^{\kappa})$, we can construct UD-AI-CWEP codes $\Psi_{\rm ai,T}$ over GF($3^m$), where $m = 2^{\kappa}$. An FFMA based on the UD-AI-CWEP $\Psi_{\rm ai,T}$ over GF($3^m$) system is proposed for supporting massive users transmission.
Without considering the channel code ${\mathcal C}_{gc}$, the UD-AI-CWEP code $\Psi_{\rm ai,T}$ based FFMA system degenerates into a classical CDMA system, indicating the proposed FFMA system can form an \textit{error-correction orthogonal spreading code}. We call such an FFMA system \textit{CDMA in finite-field (FF-CDMA)},
which has \textit{double-orthogonality} in both finite-field and complex-field.

In addition, based on the ternary non-orthogonal matrix ${\bf T}_{\rm no}(M,m)$ over GF($3^m$) and its addtive inverse matrix ${\bf T}_{\rm no, ai}(M,m)$, where $M > m$, we construct non-orthogonal CWEP (NO-CWEP) codes $\Psi_{\rm no}$. The FFMA based on the NO-CWEP code $\Psi_{\rm no}$ is called an \textit{NOMA in finite-field (FF-NOMA)} system. Without the channel code ${\mathcal C}_{gc}$, the NO-CWEP code $\Psi_{\rm no}$ based FFMA system degenerates into an NOMA system. Thus, the FF-NOMA system can form an \textit{error-correction non-orthogonal spreading code}. The NO-EP code is also suitable for network FFMA systems, which can be used in a \textit{3-dimensional butterfly network}.

The remainder of this paper is organized as follows.
Section II introduces two types of codeword-wise EP codes over finite fields: the single codeword EP code (S-CWEP) and the additive inverse codeword EP code (AI-CWEP).  
Section III presents the encoding process for EP codes.  
Section IV discusses the USPM constraints used for designing codeword-wise EP codes.  
In Section V, we construct the codeword-wise EP codes.  
Section VI introduces the transmitter and receiver for the FFMA system over GF($3^m$) in a Gaussian multiple access channel (GMAC).  
Section VII outlines the decoding process of EP codes.  
A summary of the FFMA system is provided in Section VIII.  
Section IX presents network FFMA systems based on overload EP codes.  
Simulations of the proposed FFMA systems are provided in Section X, followed by the conclusion in Section XI.

In this paper, the symbol $\mathbb B=\{0, 1\}$, ${\mathbb T} = \{0, 1, 2\}$ and $\mathbb {C}$ express the binary-field, ternary-field and complex-field, respectively. 
The notation $(a)_q$ stands for modulo-$q$, and/or an element in GF($q$).
An EP and an EP code are expressed by $C_j$ and $\Psi$, respectively.

\vspace{-0.1in}
\section{EP Codes over Finite Fields}


In this section, we first provide the definition of codeword-wise EP codes over the extension field $\text{GF}(p^m)$ of the prime field $\text{GF}(p)$, where $q = p^m$, $p$ is a prime number, and $m$ is a positive integer with $m \geq 2$.  
As defined in \cite{FFMA}, the prime number $p$ is referred to as the \textit{prime factor (PF)}, and the integer $m$ is referred to as the \textit{extension factor (EF)}.  We then introduce two specific types of codeword-wise EP codes: the \textit{single codeword EP (S-CWEP)} and the \textit{additive inverse codeword EP (AI-CWEP)} codes.

\vspace{-0.1in}
\subsection{Definition of codeword-wise EP codes}

Let $\alpha$ be a primitive element in $\text{GF}(p^m)$, where $p$ is a prime number and $m$ is a positive integer with $m \geq 2$.
Then, the powers of $\alpha$, namely $\alpha^{-\infty} = 0, \alpha^0 = 1, \alpha, \alpha^2, \ldots, \alpha^{(p^m - 2)}$, give all the $p^m$ elements of GF($p^m$). 
Each element $\alpha^j$, with $j = -\infty, 0, \ldots, p^m - 2$, in GF($p^m$) can be expressed as a linear sum of $\alpha^{-\infty} = 0, \alpha^0 = 1, \alpha, \alpha^2, \ldots, \alpha^{(m - 1)}$ with coefficients from GF($p$) as follows:
\begin{equation} \label{e2.2}
\alpha^j = a_{j,0} + a_{j,1} \alpha + a_{j,2} \alpha^2 + \ldots + a_{j,m-1} \alpha^{(m-1)}.      
\end{equation}
From (\ref{e2.2}), it shows that the element $\alpha^j$ can be uniquely represented by the $m$-tuple $(a_{j,0}, a_{j,1}, \ldots, a_{j,m-1})$ over GF($p$). 
An element in GF($p^m$) can be expressed in three forms, namely \textit{power, polynomial and $m$-tuple forms} \cite{FFMA}.

For a binary source transmission system, the transmit bit is either $(0)_2$ or $(1)_2$.
Hence, let an EP denote by $C_j = ({\alpha}^{l_{j,0}}, {\alpha}^{l_{j,1}})$, 
where $0 \le l_{j,0}, l_{j,1} \le p^m-2$ and ${\alpha}^{l_{j,0}} \neq {\alpha}^{l_{j,1}}$.
The subscript ``$j$'' of ``$l_{j,0}$'' and ``$l_{j,1}$'' stands for the $j$-th EP, 
and the subscripts ``$0$'' and ``$1$'' of ``$l_{j,0}$'' and ``$l_{j,1}$'' represent the input bits are $(0)_2$ and $(1)_2$, respectively.
Let $M$ distinct EPs denote as $C_1 = ({\alpha}^{l_{1,0}}, {\alpha}^{l_{1,1}}), 
C_2 = ({\alpha}^{l_{2,0}}, {\alpha}^{l_{2,1}}), \ldots, 
C_j = ({\alpha}^{l_{j,0}}, {\alpha}^{l_{j,1}}), \ldots, 
C_M = ({\alpha}^{l_{M,0}}, {\alpha}^{l_{M,1}})$ for $1 \le j \le M$, then the Cartesian product
\begin{equation*}
{\Psi} \triangleq C_1 \times C_2 \times \ldots C_j \times \ldots \times C_M,
\end{equation*}
of the $M$ EPs forms an $M$-user EP code ${\Psi}$ over GF($p^m$) with $2^M$ codewords. Note that \( \Psi \) is also used to denote an EP set to minimize the number of variables and simplify the notation.

For the \( M \)-user EP code \( \Psi = \{C_1, C_2, \ldots, C_M\} \), each element \( \alpha^{l_{j,0}} \) (or \( \alpha^{l_{j,1}} \)) of \( C_j \) can be expressed in an \( m \)-tuple form, where \( 1 \le j \le M \). Thus, we can construct an \( M \times m \) matrix \( {\bf G}_{\rm M}^{\bf 0} \) by arranging the \( M \) elements \( \alpha^{l_{1,0}}, \alpha^{l_{2,0}}, \ldots, \alpha^{l_{M,0}} \) as its rows, i.e., from the first row to the \( M \)-th row of \( {\bf G}_{\rm M}^{\bf 0} \), as follows:
\begin{equation}
  {\bf G}_{\rm M}^{\bf 0} = 
  \begin{bmatrix}
      \alpha^{l_{1,0}}, &
      \alpha^{l_{2,0}}, &
      \cdots, &
      \alpha^{l_{M,0}}
  \end{bmatrix}^{\rm T},
\end{equation}
where the subscript ``M'' denotes multiplexing, and the superscript ``\( {\bf 0} \)'' indicates that all \( M \) users transmit the bit \( (0)_2 \). The \( M \times m \) matrix \( {\bf G}_{\rm M}^{\bf 0} \) is referred to as the \textit{full-zero generator matrix}.

Similarly, the elements \( \alpha^{l_{1,1}}, \alpha^{l_{2,1}}, \ldots, \alpha^{l_{M,1}} \) can form another \( M \times m \) matrix \( {\bf G}_{\rm M}^{\bf 1} \), given by
\begin{equation}
  {\bf G}_{\rm M}^{\bf 1} = 
  \begin{bmatrix}
      \alpha^{l_{1,1}}, &
      \alpha^{l_{2,1}}, &
      \cdots, &
      \alpha^{l_{M,1}}
  \end{bmatrix}^{\rm T},
\end{equation}
which is referred to as the \textit{full-one generator matrix}. 
The superscript ``\( {\bf 1} \)'' indicates that all \( M \) users transmit the bit \( (1)_2 \).

\vspace{-0.1in}
\subsection{Definition of Single CWEP Codes}
For a given finite-field GF($2^m$), if the two elements ${\alpha}^{l_{j,0}}$ and ${\alpha}^{l_{j,1}}$ of the EP $C_j^{} = ({\alpha}^{l_{j,0}}, {\alpha}^{l_{j,1}})$ satisfies
${\alpha}^{l_{j,0}} = {\bf 0}$ and ${\alpha}^{l_{j,1}} \neq {\bf 0}$, 
i.e., $C_j^{} = ({\bf 0}, {\alpha}^{l_{j,1}})$, where ${\bf 0}$ is an $m$-tuple,
then the EP $C_j^{}$ is called a \textit{single codeword EP (S-CWEP)}.
The Cartesian product
\begin{equation*}
{\Psi}_{\rm cw} \triangleq C_1^{} \times C_2^{} \times \ldots \times C_M^{},
\end{equation*}
of the $M$ S-CWEPs forms an $M$-user S-CWEP code ${\Psi}_{\rm cw}$ over GF($2^m$) with $2^M$ codewords. The subscript ``cw'' indicates ``codeword''.
The full-zero generator matrix of the S-CWEP code ${\Psi}_{\rm cw}$ is an $M \times m$ zero matrix,
i.e., ${\bf G}_{\rm M}^{\bf 0} = {\bf 0}$.

\textbf{Example 1:}
For an extension field GF($2^8$) of the prime field GF($2$), it can support $M \le 8$ users.
Suppose $M = 4$, we can construct a $4$-user S-CWEP code 
${\Psi}_{\rm cw} = \{C_1^{}, C_2^{}, C_3^{}, C_{4}^{}\}$, given as
\begin{equation} \label{e.Ex2}
  \begin{aligned}
 C_1^{} = ({\bf 0}, {\alpha}^{l_{1,1}}) = (0000 0000, 1111 1111) \\
 C_2^{} = ({\bf 0}, {\alpha}^{l_{2,1}}) = (0000 0000, 0000 1111) \\
 C_3^{} = ({\bf 0}, {\alpha}^{l_{3,1}}) = (0000 0000, 0011 0011) \\
 C_4^{} = ({\bf 0}, {\alpha}^{l_{4,1}}) = (0000 0000, 0101 0101) \\
  \end{aligned},
\end{equation}
whose Cartesian product can form a $4$-user S-CWEP code ${\Psi}_{\rm cw}$ with $2^{4}= 16$ codewords.
Thus, its full-one generator matrix ${\bf G}_{\rm M}^{\bf 1}$ can be given as,
\begin{equation} \label{e.RM_8}
  {\mathbf G}_{\rm M}^{\bf 1} = \left[
  \begin{array}{c}
    {\alpha^{l_{1,1}}}\\
    {\alpha^{l_{2,1}}}\\
    {\alpha^{l_{3,1}}}\\
    {\alpha^{l_{4,1}}}\\
  \end{array} 
  \right] =
  \left[
  \begin{array}{*{16}{cccccccccccccccc}}
    1 & 1  & 1 & 1 & 1 & 1 & 1 & 1  \\
    0 & 0  & 0 & 0 & 1 & 1 & 1 & 1  \\
    0 & 0  & 1 & 1 & 0 & 0 & 1 & 1  \\
    0 & 1  & 0 & 1 & 0 & 1 & 0 & 1  \\
  \end{array}
  \right],
\end{equation}
which is a $4 \times 8$ matrix, whose rank is $4$.
$\blacktriangle \blacktriangle$

In fact, the proposed orthogonal UD-EP code $\Psi_{\rm o,B}$, constructed over $\text{GF}(2^m)$, is a special case of the S-CWEP code, where its full-one generator matrix ${\mathbf G}_{\rm M}^{\bf 1}$ is an $m \times m$ identity matrix.
For example, suppose $\Psi_{\rm o,B}$ is constructed over GF($2^4$), then its full-one generator matrix ${\mathbf G}_{\rm M}^{\bf 1}$ can be given as
\begin{equation} 
  {\mathbf G}_{\rm M}^{\bf 1} = \left[
  \begin{array}{c}
    {\alpha^{l_{1,1}}}\\
    {\alpha^{l_{2,1}}}\\
    {\alpha^{l_{3,1}}}\\
    {\alpha^{l_{4,1}}}\\
  \end{array} 
  \right] =
  \left[
  \begin{array}{*{16}{cccccccccccccccc}}
    1 & 0  & 0 & 0 \\
    0 & 1  & 0 & 0 \\
    0 & 0  & 1 & 0 \\
    0 & 0  & 0 & 1 \\
  \end{array}
  \right],
\end{equation}
which is a $4 \times 4$ identity matrix.

\subsection{Definition of Additve Inverse CWEP Codes}

For a given finite-field GF($p^m$) where $p$ is a prime larger than $2$ and $m \ge 2$, if the two elements ${\alpha}^{l_{j,0}}$ and ${\alpha}^{l_{j,1}}$ of the EP $C_j^{\rm } = ({\alpha}^{l_{j,0}}, {\alpha}^{l_{j,1}})$ satisfies
\begin{equation} \label{e.G_AIEP_def}
  {\alpha}^{l_{j,0}} \oplus {\alpha}^{l_{j,1}} = {\bf p},
\end{equation}
where $\bf p$ is an $m$-tuple whose elements are all $p$, i.e., ${\bf p} = (p, p, \ldots, p)$.
We call the EP $C_j^{\rm }$ an \textit{additive inverse CWEP (AI-CWEP)}.
The Cartesian product
\begin{equation*}
{\Psi}_{\rm ai,cw} \triangleq C_1^{} \times C_2^{} \times \ldots \times C_M^{},
\end{equation*}
of the $M$ AI-CWEPs forms an $M$-user AI-CWEP code ${\Psi}_{\rm ai,cw}$ over GF($p^m$) with $2^M$ codewords.
The full-zero and full-one generator matrices ${\bf G}_{\rm M}^{\bf 0}$ and ${\bf G}_{\rm M}^{\bf 1}$ of the AI-CWEP code ${\Psi}_{\rm ai,cw}$ satisfy the relationship
\begin{equation}
  {\bf G}_{\rm M}^{\bf 0} \oplus {\bf G}_{\rm M}^{\bf 1} = {\bf P},
\end{equation}
where ${\bf P}$ is an $M \times m$ matrix whose elements are all $p$.
Note that the full-zero generator matrix of an AI-CWEP code is not a zero matrix.

\textbf{Example 2:}
For an extension field GF($3^4$) of the prime field GF($3$), 
we can construct a $4$-user AI-CWEP code 
${\Psi}_{\rm ai,cw} = \{C_1^{}, C_2^{}, C_3^{}, C_4^{}\}$, given as
\begin{equation} \label{e.Length4}
  \begin{aligned}
 C_1^{}  = (2 2 2 2, 1 1 1 1), \quad
 C_2^{} &= (1 2 1 2, 2 1 2 1),\\
 C_3^{}  = (1 1 2 2, 2 2 1 1), \quad
 C_4^{} &= (2 1 1 2, 1 2 2 1),\\
  \end{aligned}
\end{equation}
whose Cartesian product can form an AI-CWEP code ${\Psi}_{\rm ai,cw}$ with $2^4=16$ codewords.
The full-zero and full-one generator matrices ${\bf G}_{\rm M}^{\bf 0}$ and ${\bf G}_{\rm M}^{\bf 1}$ of the AI-CWEP code $\Psi_{\rm ai,cw}$ can be given as
\begin{equation} 
{\bf G}_{\rm M}^{\bf 0} = 
\left[
  \begin{matrix}
  2 & 2 & 2 & 2 \\
  1 & 2 & 1 & 2 \\
  1 & 1 & 2 & 2 \\
  2 & 1 & 1 & 2 \\
  \end{matrix}
  \right],
  \quad  
  {\bf G}_{\rm M}^{\bf 1} = 
  \left[
  \begin{matrix}
  1 & 1 & 1 & 1 \\
  2 & 1 & 2 & 1 \\
  2 & 2 & 1 & 1 \\
  1 & 2 & 2 & 1 \\
  \end{matrix}
  \right], 
\end{equation}
which are two $4 \times 4$ matrices.
$\blacktriangle \blacktriangle$

\section{Encoding of Codeword-wise EP codes}

In this section, we assume that the sizes of the full-zero generator matrix \( {\bf G}_{\rm M}^{\bf 0} \) and the full-one generator matrix \( {\bf G}_{\rm M}^{\bf 1} \) of the EP code \( \Psi \) are both \( M \times m \). 
We begin by introducing the encoding process of a codeword-wise EP code. The input to this encoding can be either a bit or muliple bits (or a codeword), and the output is a codeword determined by the EP encoder. Specifically, we consider two modes of input: the \textit{serial mode} and the \textit{parallel mode} of the codeword-wise EP encoder. In the serial mode, the input consists of one bit, while in the parallel mode, the input is a $ K $-bit codeword.
Next, we provide a summary of the EP encoder for binary source transmission. Subsequently, we examine the output codeword-wise EP code, which is further encoded by a global channel code $ \mathcal{C}_{gc} $.
Finally, we present the framework of the FFMA with channel-coding system.

\vspace{-0.1in}
\subsection{Serial Mode: from a Bit to a Codeword}


For the serial mode, assume there are \( M \) users. In this case, the EP encoder processes one bit from each user sequentially, handling the input bit sequences one at a time.

Let ${\mathbf b}_j = (b_{j,0}, b_{j,1},\ldots, b_{j,k}, \ldots, b_{j,K-1})$ be the bit-sequence of the $j$-th user, where $1 \le j \le M$ and $0 \le k < K$.
Assume each user is assigned an EP, e.g., the EP of $C_j = (\alpha^{l_{j,0}}, \alpha^{l_{j,1}})$ is assigned to the $j$-th user.
Based on the EP $C_j$, we can encode the bit ${b}_{j,k}$ of the $j$-th user by a \textit{binary field to finite-field GF($q$) transform function} denoted by ${\mathrm F}_{{\mathrm B}2q}$, and obtain the element-sequence ${\bf u}_j = (u_{j,0},u_{j,1},\ldots, u_{j,k},\ldots, u_{j,K-1})$ whose length is $1 \times mK$.
For the $k$-th element $u_{j,k}$ of ${\bf u}_j$, we have
\begin{equation} \label{F_b2q}
  u_{j,k} = {\mathrm F}_{{\mathrm B}2q}(b_{j,k}) \triangleq b_{j,k} \odot C_j = 
  \left\{
    \begin{aligned}
      \alpha^{l_{j,0}}, \quad b_{j,k} = 0  \\
      \alpha^{l_{j,1}}, \quad b_{j,k} = 1  \\
    \end{aligned},
  \right.
\end{equation}
where $b_{j,k} \odot C_j$ is defined as a \textit{switching function} \cite{FFMA}.
If the input bit is $b_{j,k} = 0$, the transformed element is $u_{j,k} = \alpha^{l_{j,0}}$; otherwise, $u_{j,k}$ is equal to $u_{j,k} = \alpha^{l_{j,1}}$.

The \textit{finite-field multiplexing module (FF-MUX)} of the codeword-wise EP code constructed over GF($p^m$) is set to be a $M \times 1$ vector ${\mathcal A}_{\rm M} = [1, 1, \ldots, 1]^{\rm T}$ whose components are all ones \cite{FFMA}.
Then, the multiplexing is operated finite-field addition to the $k$-th elements of $M$ users (also referred to as the EP codeword), i.e., $(u_{1,k}, u_{2,k}, \ldots, u_{M,k})$, and is given by:
\begin{equation} \label{e.w_k}
  \begin{aligned}
  {w}_k = \bigoplus_{j=1}^{M} u_{j,k} = u_{1,k} \oplus u_{2,k} \oplus \ldots \oplus u_{M,k},
  \end{aligned}
\end{equation}
where ${w}_k$ is a $1 \times m$ (or an $m$-tuple) \textit{finite-field sum-pattern (FFSP)} block. 
The subscript \(k\) denotes the \(k\)-th bit of each user's contribution.

The $K$ FFSP blocks, namely ${w}_0, {w}_1, \ldots, {w}_{K-1}$, together form a $1 \times mK$ FFSP sequence:
\begin{equation} \label{e.serial_w}
  {\bf w} = ({w}_0, {w}_1, \ldots, {w}_{K-1}).
\end{equation}
It is noted that \( w \), written in italics, represents an \( m \)-tuple (or an FFSP block), while \( {\bf w} \), written in upright (roman) font, denotes a sequence consisting of several \( m \)-tuples (or FFSP blocks).
In fact, the EP encoder in serial mode corresponds to the sparse-form structure, referred to as the \textit{sparse-form FFMA (SF-FFMA)} defined in \cite{FFMA}.

In summary, the above encoding process includes two phases: one phase is the ${\mathrm F}_{{\mathrm B}2q}$ operation which transforms \textit{a bit to an $m$-tuple} based on the constructed EP code $\Psi$, and the other phase is finite-field addition operation by the FF-MUX ${\mathcal A}_{\rm M}$. The two-phases encoding process is expressed as ${\mathrm F}_{{\mathrm B}2q}/{\mathcal A}_{\rm M}$-encoding, i.e., ${\bf b}_j \mapsto {\bf u}_j \mapsto {\bf w}$.


Now, we further analyze the FFSP block given in (\ref{e.w_k}).
Let ${\bf b}[k] = (b_{1,k}, b_{2,k}, \ldots, b_{M,k})$ represent a $1 \times M$ \textit{user block} of $M$ users corresponding to the $k$-th component, where $0 \le k < K$. The generator matrix ${\mathbf G}_{\mathrm M}^{{\bf b}[k]}$ associated with ${\bf b}[k]$ is defined as
\begin{equation}
  {\mathbf G}_{\mathrm M}^{{\bf b}[k]} = 
   \left[b_{1,k} \odot C_1, b_{2,k} \odot C_2, \ldots, b_{M,k} \odot C_M \right]^{\rm T},
 \end{equation}
which forms an $M \times m$ matrix. The set of all possible combinations of ${\mathbf G}_{\mathrm M}^{{\bf b}[k]}$ constitutes a \textit{generator matrix set} ${\mathcal G}_{\mathrm M}$, i.e.,
${\mathcal G}_{\rm M} = 
\{{\mathbf G}_{\mathrm M}^{{\bf 0}}, \ldots, {\mathbf G}_{\mathrm M}^{{\bf b}[k]}, \ldots, {\mathbf G}_{\mathrm M}^{{\bf 1}}\}$,
which includes $2^M$ distinct $M \times m$ matrices, encompassing both the full-one and full-zero generator matrices.

Thus, (\ref{e.w_k}) can be rewritten as
\begin{equation} \label{e.b_kandG_M}
  \begin{aligned}
  {w}_k 
      = {\bf b}[k] \cdot {\mathbf G}_{\mathrm M}^{{\bf b}[k]} 
      \overset{(a)}{=} {\bf b}[k] \cdot {\mathbf G}_{\mathrm M}^{\bf 1} + 
        \overline{{\bf b}[k]} \cdot {\mathbf G}_{\mathrm M}^{\bf 0}
      \overset{(b)}{=} {\bf b}[k] \cdot {\mathbf G}_{\mathrm M}^{\bf 1}.
  \end{aligned}
\end{equation}
To justify step (a) in (\ref{e.b_kandG_M}), let 
$\overline{{\bf b}[k]} = (\overline{b_{1,k}}, \overline{b_{2,k}}, \ldots, \overline{b_{M,k}})$ represent the \textit{inverse user block} of ${\bf b}[k]$, where $\overline{b_{j,k}}$ is the complement of $b_{j,k}$: if $b_{j,k} = 0$, then $\overline{b_{j,k}} = 1$, and if $b_{j,k} = 1$, then $\overline{b_{j,k}} = 0$. 
This allows us to split ${\bf b}[k] \cdot {\mathbf G}_{\mathrm M}^{{\bf b}[k]}$ into two parts: ${\bf b}[k] \cdot {\mathbf G}_{\mathrm M}^{\bf 1}$ and $\overline{{\bf b}[k]} \cdot {\mathbf G}_{\mathrm M}^{\bf 0}$.

For a specific case, when the full-zero generator matrix is a zero matrix (as in the case of an S-CWEP code), we obtain step (b) in (\ref{e.b_kandG_M}). In this case, the FFSP ${w}_k$ is simply the user block ${\bf b}[k]$ encoded by a channel code with generator matrix ${\mathbf G}_{\mathrm M}^{\bf 1}$.

Hence, by utilizing the generator matrix set ${\mathcal G}_{\mathrm M}$, we can directly encode the $K$ input user blocks ${\bf b}[0], {\bf b}[1], \ldots, {\bf b}[K-1]$ into $K$ FFSP blocks ${w}_0, {w}_1, \ldots, {w}_{K-1}$. These resulting FFSP blocks can be regarded as the codewords of a \textit{multiuser code (MC)}, denoted by ${\mathcal C}_{mc}$. This encoding process is referred to as one-phase ${\mathcal G}_{\rm M}$-encoding, i.e., ${\bf b}[k] \mapsto {w}_k$.

\subsection{Parallel Mode: from a Codeword to a Codeword}

Suppose the integer $ M $ is the product of two integers $ K $ and $ J_{mc} $, i.e., $ M = K \cdot J_{mc} $.  
In the parallel mode, the input for each user is a $ K $-bit codeword, and this mode can support a maximum of $ J_{mc} $ users. The subscript ``mc'' stands for ``multiuser code''.

In the parallel mode, each bit $b_{j,k}$ of the bit-sequence ${\mathbf b}_j = (b_{j,0}, b_{j,1}, \dots, b_{j,k}, \dots, b_{j,K-1})$ for the $j$-th user is assigned a unique element from the EP, where $1 \leq j \leq J_{mc}$ and $0 \leq k < K$.
Consider that each user has $ K $ bits, so each user is assigned $ K $ distinct EPs. Specifically, for the $ j $-th user, the assigned $ K $ EPs are given by:
$
C_{(j-1) \cdot K} = (\alpha^{l_{(j-1) \cdot K,0}}, \alpha^{l_{(j-1) \cdot K,1}}),
C_{(j-1) \cdot K+1} = (\alpha^{l_{(j-1) \cdot K+1,0}}, \alpha^{l_{(j-1) \cdot K+1,1}}),
\dots,
C_{(j-1) \cdot K+k} = (\alpha^{l_{(j-1) \cdot K+k,0}}, \alpha^{l_{(j-1) \cdot K+k,1}}),
\dots,
C_{(j-1) \cdot K+(K-1)} = (\alpha^{l_{(j-1) \cdot K+(K-1),0}}, \alpha^{l_{(j-1) \cdot K+(K-1),1}}).
$
For the $ k $-th bit of the $ j $-th user, the corresponding transformed element is defined as:
\begin{equation} \label{F_pll_b2q}
  u_{j,k} = {\mathrm F}_{{\mathrm B}2q}(b_{j,k}) \triangleq 
                           b_{j,k} \odot C_{(j-1) \cdot K+k} = 
  \left\{
    \begin{aligned}
      \alpha^{l_{(j-1) \cdot K+k,0}}, \quad b_{j,k} = 0  \\
      \alpha^{l_{(j-1) \cdot K+k,1}}, \quad b_{j,k} = 1  \\
    \end{aligned}.
  \right.
\end{equation}
Thus, if the input bit $ b_{j,k} = 0 $, the corresponding transformed element is $ u_{j,k} = \alpha^{l_{(j-1) \cdot K+k,0}} $; otherwise, if $ b_{j,k} = 1 $, then $ u_{j,k} = \alpha^{l_{(j-1) \cdot K+k,1}} $.
We then sum the \( K \) transformed elements together to obtain the output element \( {\bf c}_j \) for the \( j \)-th user, as given by
\begin{equation}
  {\bf c}_j = \bigoplus_{k=0}^{K-1} u_{j,k},
\end{equation}
where \( 1 \le j \le J_{mc} \), and \( {\bf c}_j \) is also an \( m \)-tuple. 
In fact, the output element \( {\bf c}_j \) is also the \textit{FFSP block} of the \( K \) elements. 
Similarly, the EP encoder in parallel mode corresponds to the diagonal-form structure, known as the \textit{diagonal-form FFMA (DF-FFMA)} defined in \cite{FFMA}.

Next, multiplexing is applied to the \( J_{mc} \) users, which is equivalent to the sum of the \( K \) elements of the \( J_{mc} \) users in parallel mode. Let the \textit{parallel user block} be denoted as \( \mathbf{b}_{\rm pll} = (\mathbf{b}_1, \mathbf{b}_2, \dots, \mathbf{b}_{J_{mc}}) \), a \( 1 \times (K \cdot J_{mc}) \) vector, which can also be represented as a \( 1 \times M \) vector.
Suppose \( \overline{\mathbf{b}_{\rm pll}} = (\overline{\mathbf{b}_1}, \overline{\mathbf{b}_2}, \dots, \overline{\mathbf{b}_{J_{mc}}}) \) represents the \textit{inverse parallel user block} of \( \mathbf{b}_{\rm pll} \). The subscript ``pll'' denotes ``parallel''. Then, the corresponding FFSP block of $J_{mc}$ users can be expressed as:
\begin{equation} \label{e.pll_w_k}
  \begin{aligned}
  {w} = \bigoplus_{j=1}^{J_{mc}} {\bf c}_{j}
          = \bigoplus_{j=1}^{J_{mc}} \bigoplus_{k=0}^{K-1}  u_{j,k} 
 \overset{(a)}{=} {\bf b}_{\rm pll} \cdot {\mathbf G}_{\mathrm M}^{\bf 1} + 
       \overline{{\bf b}_{\rm pll}} \cdot {\mathbf G}_{\mathrm M}^{\bf 0}
 \overset{(b)}{=} {\bf b}_{\rm pll} \cdot {\mathbf G}_{\mathrm M}^{\bf 1}.
  \end{aligned}
\end{equation}
The derivations of step (a) in \eqref{e.pll_w_k} is similar to that in the serial mode, and will not be repeated here. 
For the S-CWEP code, we can obtain step (b) in (\ref{e.pll_w_k}), where the FFSP block ${w}$ is the parallel user block ${\bf b}_{\rm pll}$ encoded by the generator matrix ${\mathbf G}_{\mathrm M}^{\bf 1}$.

If there are a total of \( M \) users, these users can be grouped into \( K = \frac{M}{J_{mc}} \) distinct groups, with each group containing \( J_{mc} \) users. Each user within a group contributes a \( K \)-bit, and the resulting parallel user block, denoted as \( {\bf b}_{\rm pll}^{(t)} \), is encoded by the generator matrix set ${\mathcal G}_{\mathrm M}$ of the EP code $\Psi_{\rm cw}$. The superscript ``\( (t) \)'' indicates the \( t \)-th data group, where \( 1 \le t \le K \). The resulting FFSP sequence \( {\bf w} \) consists of \( K \) FFSP blocks, i.e.,
\begin{equation} \label{e.parallel_w}
  {\bf w} = ({w}^{(1)}, {w}^{(2)}, \ldots, {w}^{(t)}, \ldots, {w}^{(K)}),
\end{equation}
where the total length of the sequence ${\bf w}$ is \( 1 \times mK \), which is identical to the length in the serial mode. 
The superscript ``\( (t) \)'' also denotes the \( t \)-th FFSP block, where \( 1 \le t \le K \).

\begin{figure}[t] 
  \centering
  \includegraphics[width=0.99\textwidth]{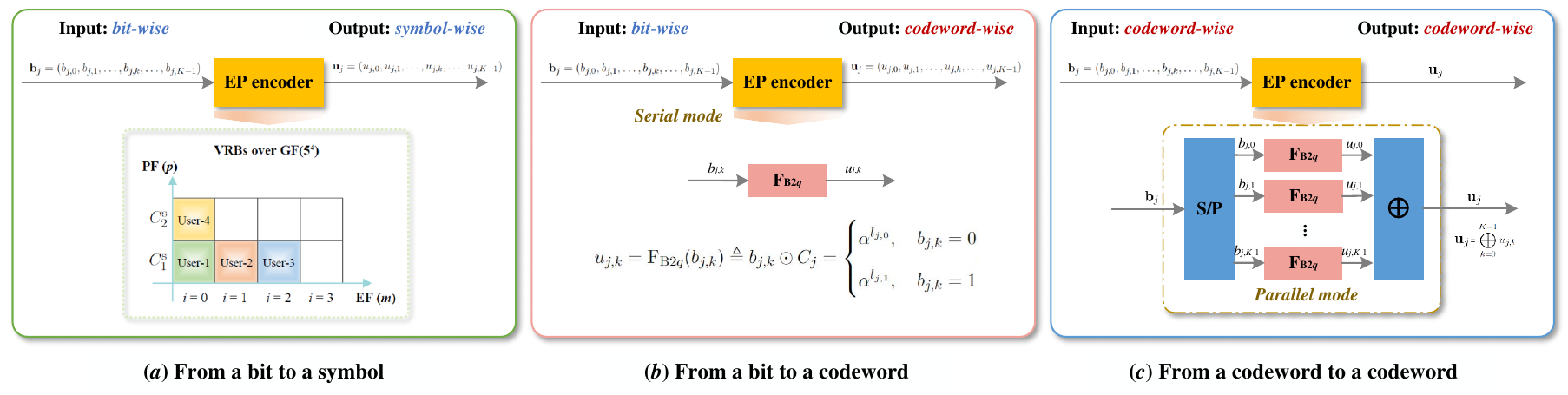}
 \caption{EP Encoder structure. (a) From a bit to a symbol. (b) From a bit to a codeword. (c) From a codeword to a codeword.} 
 \label{f_EP_Encoder}
 \vspace{-0.3in}
\end{figure}

\subsection{Overview of the EP Encoder}
Following \cite{FFMA} and this paper, we provide a summary of the EP encoder. Specifically, the input to an EP encoder can either be a bit or a codeword, while the output can either be a symbol or a codeword, as shown in Fig. \ref{f_EP_Encoder}.
When the output of the EP encoder is a symbol, we refer to the corresponding EP code as a \textit{symbol-wise} EP code. Conversely, when the output is a codeword, we call the corresponding EP code a \textit{codeword-wise EP} code.

In \cite{FFMA}, the input to the EP encoder is bit-wise, and the output is a symbol. The constructed EP codes are symbol-wise EP codes, which are based on a prime field GF($p$) and/or its extension field GF($p^m$). When \( p > 2 \), the number of EPs, also called virtual resource blocks (VRBs), is bounded by \( m \cdot \log_2(p-1) \) \cite{FFMA}. When \( p = 2 \), the number of EPs (or VRBs) is simply \( m \).
In this paper, the output of the EP encoder is a codeword, and the input can either be a bit or a codeword. This results in two modes for the EP encoder: serial mode and parallel mode, as discussed above.

In addition, distinct users and bits are assigned unique EPs in finite fields. At the receiver end, the identities of both users and bits are simultaneously recovered by decoding the FFSP. This leads to the following corollary (Corollary \ref{Coro_user_bit}) and theorem (Theorem \ref{Theorem_EP_encoder}).

\begin{corollary} \label{Coro_user_bit}
  In an FFMA system, since both users and bits are uniquely identified by EPs, the terms ``user'' and ``bit'' are interchangeable and play equivalent roles in finite field.
\end{corollary}

\begin{theorem} \label{Theorem_EP_encoder}
  (\textbf{Upper Bound on Bits})  
  Let \( \Psi \) be an EP code constructed over \( \text{GF}(p^m) \), where the generator matrix in the generator matrix set \( \mathcal{G}_{\rm M} \) has size \( M \times m \), with $M$ representing the number of EPs, \( p = 2, 3 \), and \( m \geq 2 \). Suppose there are \( J_{mc} \) users, and each user transmits \( K \) bits. Then, the relationship between the number of users \( J_{mc} \) and the number of bits per user \( K \) is given by:
  \begin{equation}
     J_{mc} \cdot K \le M,
  \end{equation} 
  where the inequality is constrained by the number of EPs.
\end{theorem}

\vspace{-0.1in}
\subsection{Channel Encoding on EP-coding}

In this subsection, we analyze the effect of channel codes on EP-coding, which can be viewed as the concatenation of an EP code \( \Psi \) and a channel code \( {\mathcal C}_{gc} \). The subscript ``gc'' of \( {\mathcal C}_{gc} \) stands for ``global channel code'', indicating that the channel code is applied throughout the entire transmission.

The effect of channel codes on EP-coding, as determined by the locations in the extension field, has been previously discussed in \cite{FFMA}. In \cite{FFMA}, the EPs (or VRBs) are determined based on their positions, and the elements \( {\bf u}_j \) produced by the EP encoder are passed to the channel code \( {\mathcal C}_{gc} \), resulting in the encoded codeword \( {\bf v}_j = {\bf u}_j \cdot {\bf G}_{gc} \).
Suppose the codeword \( {\bf v} \) is the FFSP block encoded by the channel code \( {\mathcal C}_{gc} \), i.e., 
${\bf v} = \left( \bigoplus_{j=1}^J {\bf u}_j \right) \cdot {\bf G}_{gc}$.
In \cite{FFMA}, it was shown that ${\bf v} = \sum_{j=1}^{J} {\bf v}_j$.
Thus, in this paper, we focus exclusively on the case where the EPs are located at the same positions and analyze the effect of the channel codes.

Suppose the EP code \( \Psi \) is constructed over \( \text{GF}(p^m) \), with the generator matrix from the generator matrix set \( \mathcal{G}_{\rm M} \) having sizes \( M \times m \), where \( M \) is an integer, \( p = 2 \) or \( 3 \), and \( m \geq 2 \). 
In this paper, we focus on the parallel mode of the EP encoder for analysis, as the serial mode can be considered a special case of the parallel mode. Therefore, let the EP code \( \Psi \) support \( J_{mc} \) users, each transmitting \( K \) bits, with the assumption that \( M = J_{mc} \cdot K \).
Next, after passing the bit-sequence \( \mathbf{b}_j \) of the \( j \)-th user to the encoder in parallel mode, we obtain the corresponding element \( \mathbf{c}_j \) for the \( j \)-th user, which is the sum of \( K \) elements $(u_{j,0}, u_{j,1}, ..., u_{j,K-1})$, i.e., 
${\bf c}_j = \oplus_{k=0}^{K-1} u_{j,k}$, where \( \mathbf{c}_j \) is an \( m \)-tuple.

To simplify the computational complexity, we assume that an $(N, mT)$ linear block channel code \( \mathcal{C}_{gc} \) is constructed over \( \text{GF}(p) \), where \( T \) and \( N \) are integers. Here, \( mT \) represents the length of the information section, which can be partitioned into \( T \) data blocks, each of length \( m \), and \( N \) is the length of the codeword of \( \mathcal{C}_{gc} \). Let \( R \) denote the length of the parity section, where \( R = N - mT \).

Let the generator matrix of the channel code \( \mathcal{C}_{gc} \) be an \( mT \times N \) matrix \( \mathbf{G}_{gc} \), i.e., \( \mathbf{G}_{gc} = [{\bf g}_1, {\bf g}_2, \ldots, \\ {\bf g}_t, ..., {\bf g}_T]^{\rm T} \), where \( \mathbf{G}_{gc} \) is a \( T \times 1 \) array, and each array \( \mathbf{g}_t \) is an \( m \times N \) matrix over \( \text{GF}(p) \) for \( 1 \leq t \leq T \). In other words, each array \( \mathbf{g}_t \) corresponds to the data block \( t \) of the information section.

In this paper, we assume that \( T \geq 1 \), meaning that the element-sequence of the \( j \)-th user, \( \mathbf{c}_j \), is always smaller than or equal to the length of the information section of the global channel code, \( mT \). We further assume that all \( J_{mc} \) users are located at the first block (\( t=1 \)) of the information section. 
Therefore, we append zeros to the element-sequence \( \mathbf{c}_j \) and define the transmitted information section of the \( j \)-th user as \( \mathbf{c}_{j, \rm D}^{(1)} = (\mathbf{c}_j, \mathbf{0}) \), where \( \mathbf{0} \) is a \( 1 \times (T-1)m \) zero vector. This ensures that the length of the transmitted information section becomes \( mT \), which can then be encoded using the global channel code \( \mathcal{C}_{gc} \).
The superscript ``$(1)$" indicates that the \( j \)-th user is located at the first data block of the information section of the channel code \( \mathcal{C}_{gc} \); and the subscript ``D'' repsrents the ``diagonal-form'' as defined in \cite{FFMA}.
Consequently, the encoded codeword \( \mathbf{v}_j^{(1)} \) of the \( j \)-th user is given by:
\begin{equation}
    \begin{aligned}
   {\bf v}_j^{(1)} = {\bf c}_{j,\rm D}^{(1)} \cdot {\bf G}_{gc}\\
   \end{aligned}.
\end{equation}

Thus, the sum of $J_{mc}$-user's codewords ${\bf v}_{1}^{(1)}, {\bf v}_{2}^{(1)}, \ldots, {\bf v}_{J_{mc}}^{(1)}$ are given as (\ref{e.ChannelCode_deduction}).
\begin{small}
\begin{equation} \label{e.ChannelCode_deduction}
    \begin{aligned}
   {\bf v}_{sum} =&\bigoplus_{j=1}^{J_{mc}} {\bf v}_{j}^{(1)}  
   = \bigoplus_{j=1}^{J_{mc}} \left[{\bf c}_{j, \rm D}^{(1)} \cdot {\bf G}_{gc} \right] 
   = \bigoplus_{j=1}^{J_{mc}} \left[({\bf c}_{j}, {\bf 0}, ..., {\bf 0}) \cdot 
                  [{\bf g}_{1}, {\bf g}_{2}, ..., {\bf g}_{T}]^{\rm T} \right]
   = \bigoplus_{j=1}^{J_{mc}} {\bf c}_{j} \cdot {\bf g}_{1} \\
   =& ({\bf c}_1 \oplus {\bf c}_2 \oplus ... \oplus {\bf c}_j \oplus ... \oplus {\bf c}_{J_{mc}} ) \cdot {\bf g}_{1} \\
   \overset{(a)}{=}& {w}^{(1)} \cdot {\bf g}_{1} 
   \overset{(b)}{=} {\bf w}_{\rm D}^{(1)} \cdot {\bf G}_{gc} = {\bf v}^{(1)}.
   \end{aligned}
\end{equation}
\end{small}
We can deduce (a) from equation (\ref{e.ChannelCode_deduction}) as follows: ${w}^{(1)}= \oplus_{j=1}^{J_{mc}} \mathbf{c}_{j} = \mathbf{c}_{1} \oplus \mathbf{c}_{2} \oplus \cdots \oplus \mathbf{c}_{J_{mc}}$,
and deduce (b) based on the assumption that \( \mathbf{w}_{\rm D}^{(1)} = ({w}^{(1)}, \mathbf{0}) \), where \( \mathbf{0} \) is a \( 1 \times m(T-1) \) zero vector and ${\bf w}_{\rm D}^{(1)}$ is a $1 \times mT$ vector.
In equation (\ref{e.ChannelCode_deduction}), \( \mathbf{v}^{(1)} \) represents the FFSP sequence \( \mathbf{w}_{\rm D}^{(1)} \) encoded by the generator \( \mathbf{G}_{gc} \).

From equation (\ref{e.ChannelCode_deduction}), we observe that the sum of the codewords \( \mathbf{v}_{1}^{(1)}, \mathbf{v}_{2}^{(1)}, \dots, \mathbf{v}_{J_{mc}}^{(1)} \) for the \( J_{mc} \)-user system is equal to the codeword \( \mathbf{v}^{(1)} \), which represents the FFSP sequence \( \mathbf{w}^{(1)} \) encoded by the generator \( \mathbf{G}_{gc} \), i.e.,
$\mathbf{v}^{(1)} = \bigoplus_{j=1}^{J_{mc}} \mathbf{v}_j^{(1)}$.

In summary, based on the results from \cite{FFMA} and this paper, we observe that whether the users are located at different positions or at the same location, the encoded codeword of FFSP is always equal to the sum of the individual users' codewords. This result can be generalized to a case where users simultaneously occupy both the same and different locations, a further deduction of which is not provided here.

Furthermore, due to the properties of finite fields, we can interchange the order of operations for finite-field addition \( {\mathcal A}_{\rm M} \) and the channel encoder \( \mathbf{G}_{gc} \), leading to the same encoded codeword. Consequently, the concatenation of the EP code \( \Psi \) and the channel code \( {\mathcal C}_{gc} \) can be implemented in two distinct ways:

\begin{enumerate}
  \item As a three-phase encoding process of \( {\rm F}_{{\rm B}2q}/{\mathcal A}_{\rm M}/{\mathbf G}_{gc} \), represented by the transformation \( \mathbf{b}_j \mapsto \mathbf{u}_j \mapsto \mathbf{w} \mapsto \mathbf{v} \).
  \item As a three-phase encoding process of \( {\rm F}_{{\rm B}2q}/{\mathbf G}_{gc}/{\mathcal A}_{\rm M} \), represented by the transformation \( \mathbf{b}_j \mapsto \mathbf{u}_j \mapsto \mathbf{v}_j \mapsto \mathbf{v} \).
\end{enumerate}

The first one, \( {\rm F}_{{\rm B}2q}/{\mathcal A}_{\rm M}/{\mathbf G}_{gc} \)-encoding, is suitable for decoding at the receiver. 
In contrast, the second one, \( {\rm F}_{{\rm B}2q}/{\mathbf G}_{gc}/{\mathcal A}_{\rm M} \)-encoding, is typically used for encoding at the transmitters of \( J \) users.

\vspace{-0.1in}
\subsection{Framework of an FFMA system with Channel Coding}
In this subsection, we present the framework of an FFMA system with channel coding. As described in \cite{FFMA}, the key idea behind an FFMA system is to exchange the order of channel code and MUX modules. Specifically, for an FFMA system, an EP code serves as the input to an FF-MUX, followed by a channel code. Thus, the overall framework of the FFMA system is primarily determined by the interaction between these two codes: the EP code, denoted as $\Psi$, and the channel code, represented by ${\mathcal C}_{gc}$.

Suppose the EP code $\Psi$ is constructed over $\text{GF}(p^m)$, where the generator matrix in the set ${\mathcal G}_{\rm M}$ has sizes $M \times m$, with $M$ being an integer, $p = 2, 3$, and $m \geq 2$. Each user transmits $K$ bits, with $1 < K < m$. The channel code $\mathcal{C}_{gc}$ is an $(N, mT)$ linear block code over $\text{GF}(p)$, where $T$ and $N$ are integers. We assume that $T > K$, so that the length of the information section of the channel code, $mT$, is longer than the length of the EP code in serial mode, $mK$.

As presented in \cite{FFMA}, the channel code in systematic form offers better error performance. Therefore, without further elaboration, this paper assumes that the generator matrix of the channel code is in systematic form, denoted as ${\bf G}_{gc}$, instead of ${\bf G}_{gc,sym}$ as defined in \cite{FFMA}. Furthermore, the systematic form of the channel code allows the FFMA system to operate in a diagonal-form structure.

First, the framework of an FFMA system can be divided into two sections: the information section and the parity section, which correspond to the sections of the channel code ${\mathcal C}_{gc}$. The information section consists of $T$ data blocks, each of length $m$, while the parity block length is $R = N - mT$. The $T$ data blocks can be treated as $T$ time slots, similar to the structure of slotted ALOHA. The indices of the data blocks range from $1, 2, \ldots, t, \ldots, T$.

Next, each data block can be assigned an EP code, which differs from the approach in \cite{FFMA}. In \cite{FFMA}, a data block is assigned a bit-sequence. Given $T$ data blocks, we can assign $T$ EP codes corresponding to these blocks.
Suppose the FFMA system with channel coding supports $J$ users. Each EP encoder can operate in either serial or parallel mode to serve the users.
\begin{enumerate}
  \item 
  In serial mode, we assume that $M$ users and each transmitting a $1$-bit share the same EP code (or data block). Each user transmits $K$ bits, and the $M$ users share the same $K$ data blocks. Let the indices of the $K$ data blocks be $t, t+1, \dots, t + (K-1)$. The corresponding arrays from ${\mathcal C}_{gc}$, i.e., ${\bf g}_t, {\bf g}_{t+1}, \dots, {\bf g}_{t+(K-1)}$, are assigned to these blocks. Thus, the information section can be viewed as a $\frac{T}{K} \times \frac{T}{K}$ diagonal array, with each array of size $M \times mK$. 
  \item 
  In parallel mode, we assume that $J_{mc}$ users and each transmitting $K$ bits share the same EP code (or data block). Given that $K < m$, the $J_{mc}$ users share a single data block, whose index is set to be $t$. The $t$-th array from ${\mathcal C}_{gc}$, defined as ${\bf g}_t$, is assigned to the $J_{mc}$ users. In this case, the information section can be represented as a $T \times T$ diagonal array, with each array of size $J_{mc} \times m$.
\end{enumerate}

It is important to determine the maximum number of users $J$ that can be served by the FFMA system. Referring to Corollary \ref{Coro_user_bit} and Theorem \ref{Theorem_EP_encoder}, we can derive the following result, stated in Theorem \ref{Theorem_EP_ChannelCode}.
\begin{theorem} \label{Theorem_EP_ChannelCode}
  (\textbf{Upper Bound on Users})
  Consider an FFMA system with channel coding, where the channel code $\mathcal{C}_{gc}$ is an $(N, mT)$ linear block code over $\text{GF}(p)$, with $T$ and $N$ being integers. The EP code $\Psi$ is constructed over $\text{GF}(p^m)$, where the generator matrix in the set ${\mathcal G}_{\rm M}$ has sizes $M \times m$, with $M$ representing the number of EPs, $p = 2, 3$ and $m \ge 2$. Suppose each user transmits $K$ bits. Then, the maximum number of users that can be served is given by:
  \begin{equation}
    J  \le \frac{M \cdot T}{K}.
  \end{equation}
\end{theorem}

\begin{figure}[t] 
  \centering
  \includegraphics[width=0.6\textwidth]{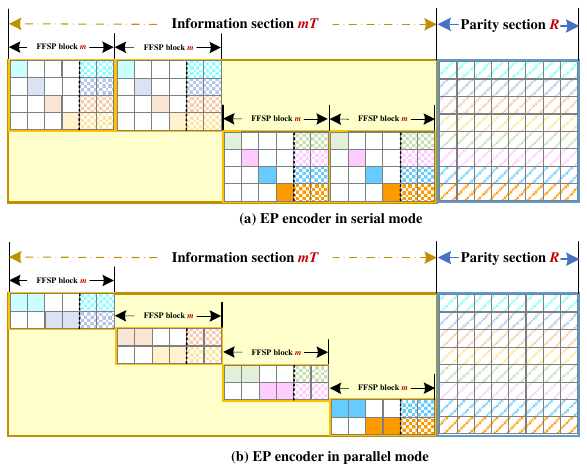}
 \caption{A diagram of Example 3 to show the framework of an FFMA system, where the total number of users is \( J = 8 \) and the number of bits per user is \( K = 2 \). 
 (a) EP encoder in serial mode;
 (b) EP encoder in parallel mode.
 } 
 \label{f_EP_ChannelCode}
 \vspace{-0.3in}
\end{figure}

\textbf{Example 3:}
We provide an example to illustrate the framework of an FFMA system with channel coding. Assume the size of the full-one generator matrix \( {\bf G}_{\rm M}^{\bf 1} \) of the EP encoder \( \Psi \) is \( 4 \times 6 \), and the size of the generator matrix \( {\bf G}_{gc} \) of the linear \( (32, 24) \) block channel code \( \mathcal{C}_{gc} \) is \( 24 \times 32 \). Consequently, the number of data blocks is \( T = 24 / 6 = 4 \). Suppose each user transmits \( K = 2 \) bits. The maximum number of supported users is then calculated as
\begin{equation*}
  J = \frac{M \cdot T}{K} = \frac{4 \times 4}{2} = 8.
\end{equation*}

Next, we organize the \( 8 \) users, each with \( 2 \) bits, into the framework.
If the EP encoder operates in serial mode, the framework of the proposed FFMA system is shown in Fig. \ref{f_EP_ChannelCode} (a). If the EP encoder operates in parallel mode, the framework is shown in Fig. \ref{f_EP_ChannelCode} (b). In both serial and parallel modes, a total of \( 16 \) bits can be accommodated.
$\blacktriangle \blacktriangle$

\section{Unique Sum-pattern Mapping Constraints}
There are various finite fields, which provide both flexibility and complexity in designing EP codes. In this paper, we focus on two types of codeword-wise EP codes: one is the S-CWEP code over \( \text{GF}(2^m) \), and the other is the AI-CWEP code over \( \text{GF}(3^m) \). We introduce the \textit{unique sum-pattern mapping (USPM) structural property} constraints for these two types of codeword-wise EP codes, one by one. Based on the USPM constraint, we can construct uniquely decodable S-CWEP (UD-S-CWEP) codes, which will be introduced in the following section.

\vspace{-0.15in}
\subsection{USPM constraint of S-CWEP code over GF($2^m$)}

As discussed earlier, for the EP encoder in serial mode, the FFSP of an S-CWEP code is simply the user block \( \mathbf{b}[k] \) encoded by a channel code with generator matrix \( \mathbf{G}_{\mathrm{M}}^{\mathbf{1}} \), i.e.,
\(
w_k = \mathbf{b}[k] \cdot \mathbf{G}_{\mathrm{M}}^{\mathbf{1}}.
\)
For the EP encoder in parallel mode, the FFSP of an S-CWEP code is given by
\(
w = \mathbf{b}_{\mathrm{pll}} \cdot \mathbf{G}_{\mathrm{M}}^{\mathbf{1}}.
\)
Therefore, investigating the properties of the full-one generator matrix \( \mathbf{G}_{\mathrm{M}}^{\mathbf{1}} \) is crucial for constructing UD-S-CWEP codes.
Next, we introduce the \textit{USPM structural property constraint} for the S-CWEP code over \( \text{GF}(2^m) \) in Theorem \ref{USPM_2m}.
It is important to note that, according to Corollary \ref{Coro_user_bit}, there is no distinction between \( \mathbf{b}[k] \) and \( \mathbf{b}_{\mathrm{pll}} \) in finite fields. Therefore, we proceed by using the serial mode to derive the USPM constraints.

\begin{theorem} \label{USPM_2m}
(\textbf{USPM constraint of S-CWEP codes over GF($2^m$)})
For an extension field GF($2^m$) where $m \ge 2$, a S-CWEP code $\Psi_{\rm cw} = \{C_1^{}, C_2^{}, \ldots, C_M^{}\}$ is utilized to support an $M$-user FFMA system, with $M \le m$. To ensure unique decoding of the FFSP block $w_k$, there must be a one-to-one mapping between the user block $\mathbf{b}[k]$ and the FFSP block $w_k$, i.e., $\mathbf{b}[k] \leftrightarrow w_k$. This indicates that the code possesses a unique sum-pattern mapping structural property. The $M \times m$ full-one generator matrix of the multiuser code, denoted as $\mathbf{G}_{\mathrm{M}}^{\mathbf{1}} = \left[\alpha^{l_{1,1}}, \alpha^{l_{2,1}}, \cdots, \alpha^{l_{M,1}}\right]^{\rm T}$, should have full row rank, which can be expressed as
  \begin{equation} \label{e.thereom1}
    {\rm Rank}\left({\bf G}_{\rm M}^{\bf 1}\right) = M,
  \end{equation}
indicating that $\alpha^{l_{1,1}}, \alpha^{l_{2,1}}, \cdots, \alpha^{l_{M,1}}$ are linearly independent. We refer to $\Psi_{\rm cw}$ as a uniquely decodable S-CWEP (UD-S-CWEP) code.
\end{theorem}

\begin{proof}
For a S-CWEP code, we have $w_k = {\bf b}[k] \cdot {\mathbf G}_{\mathrm M}^{\bf 1}$, 
\begin{equation*}
  w_k = b_{1,k} \cdot \alpha^{l_{1,1}} + b_{2,k} \cdot \alpha^{l_{2,1}} + \ldots +
             b_{M,k} \cdot \alpha^{l_{M,1}}.
\end{equation*}
Thus, $\alpha^{l_{1,1}}, \alpha^{l_{2,1}},\cdots,\alpha^{l_{M,1}}$ should be linear independent, and it is a one-to-one mapping between the user-block ${\bf b}[k]$ and the FFSP block $w_k$, with the USPM structural property.
In other words, the full-one generator matrix ${\bf G}_{\rm M}^{\bf 1}$ is full row rank.
\end{proof}

Recall Example 1, the rank of ${\bf G}_{\rm M}^{\bf 1}$ of the S-CWEP code $\Psi_{\rm cw}$ is equal to $4$.
Hence, the proposed S-CWEP code $\Psi_{\rm cw}$ is a $4$-user UD-S-CWEP code.
In fact, Eq. (\ref{e.RM_8}) is also the generator matrix of $1$-order RM($1, 3$) code, whose codeword length is $8$ and minimum distance is equal to $2^2 = 4$.

We can construct UD-S-CWEP codes based on classical channel codes, such as LDPC codes \cite{LinBook3}, Reed-Muller (RM) codes, and others. When a UD-S-CWEP code \( \Psi_{\rm cw} \) is constructed based on the channel code \( \mathcal{C}_{mc} \), the full-one generator matrix \( \mathbf{G}_{\rm M}^{\mathbf{1}} \) of \( \Psi_{\rm cw} \) is identical to the generator matrix of \( \mathcal{C}_{mc} \). We can also refer to \( \mathcal{C}_{mc} \) as a \textit{multiuser code (MC)}, as its role is to distinguish users and function as output of an FF-MUX.
Additionally, the \textit{loading factor} \( \eta \) of the UD-S-CWEP code \( \Psi_{\rm cw} \), is given by \( \eta = \frac{M}{m} \le 1 \).
Furthermore, if the full-one generator matrix of an UD-S-CWEP code is derived from a channel code, the resulting system is referred to as \textit{channel codeword multiple access (CCMA) in finite fields (FF-CCMA)}.

From Theorem \ref{USPM_2m}, we observe that the FFSP block \( w_k \) is, in fact, a codeword of the multiuser code \( \mathcal{C}_{mc} \). It is important to note that an EP code is employed at the transmitting end for each user, while the multiuser code is used at the receiving end to separate users.

\subsection{USPM constraint of AI-CWEP code over GF($3^m$)}

Before we introduce the USPM structural property constraint of the AI-CWEP codes over GF($3^m$), we first introduce some properties of the prime-filed GF($3$).

\begin{property}
For a prime-field GF($3$), there are three elements $(0)_3, (1)_3$ and $(2)_3$. 
There are some basic properties of the prime-field GF($3$).
  \begin{enumerate}
    \item
    The inverse of $(1)_3$ is $(2)_3$, and the inverse of $(2)_3$ is $(1)_3$, i.e.,
    $\overline{(1)}_3 = (2)_3$ and $\overline{(2)}_3 = (1)_3$.
    \item
    Multiplying $2$ to the elements $(1)_3$ and $(2)_3$ are equal to the inverse of the elements $(1)_3$ and $(2)_3$, respectively, i.e.,  
    $2 \cdot (1)_3 = \overline{(1)}_3$ and $2 \cdot (2)_3 = \overline{(2)}_3$.
    \item
    The power of $(1)_3$ is still $(1)_3$, i.e., $(1)_3^{n} = (1)_3$, 
    and the power of $(2)_3$ can be given as
      \begin{equation*}
        (2)_3^{n} = \left\{
          \begin{aligned}
            (1)_3, & \quad n = 2l, \\
            (2)_3, & \quad n = 2l+1, \\
          \end{aligned}
          \right.
      \end{equation*}
    where $l \in {\mathbb Z}$. If the power $n$ is an even number, $(2)_3^{n} = (1)_3$; 
    otherwise, $(2)_3^{n} = (2)_3$.
  \end{enumerate}
\end{property}

The USPM structural property constraint of the AI-CWEP code over GF($3^m$) is presented as following.

\begin{theorem} \label{USPM_3m}
(\textbf{USPM constraint of AI-CWEP code over GF($3^m$)})
For an extension field GF($3^m$) where $m \ge 2$, 
an AI-CWEP code $\Psi_{\rm ai,cw} =\{C_1^{}, C_2^{}, \ldots, C_M^{}\}$ is utilized to support an $M$-user FFMA system, with $M \le m$.
To ensure unique decoding of the FFSP block $w_k$, there must be a one-to-one mapping between the user block ${\bf b}[k]$ and the FFSP block $w_k$, i.e., ${\bf b}[k] \leftrightarrow w_k$.
This indicates that the code possesses a unique sum-pattern mapping structural property.
The $M \times m$ full-one generator matrix of the multiuser code, 
denoted as ${\bf G}_{\rm M}^{\bf 1} = \left[\alpha^{l_{1,1}}, \alpha^{l_{2,1}},\cdots,\alpha^{l_{M,1}}\right]^{\rm T}$, should have full row rank, which can be expressed as
  \begin{equation} \label{e.thereom1}
    {\rm Rank}\left({\bf G}_{\rm M}^{\bf 1}\right) = M,
  \end{equation}
indicating that $\alpha^{l_{1,1}}, \alpha^{l_{2,1}},\cdots,\alpha^{l_{M,1}}$ are linear independent. Consequently, the rows of any matrix in the generator matrix set $\mathcal{G}_{\rm M}$, 
e.g., $[\alpha^{l_{1,1}}, \alpha^{l_{2,1}}, \ldots, \alpha^{l_{j,0}}, \ldots, \alpha^{l_{M,1}}]^{\rm T}$, are also linearly independent. We refer to $\Psi_{\rm ai,cw}$ as a uniquely decodable AI-CWEP (UD-AI-CWEP) code.
\end{theorem}

\begin{proof}
Since ${\bf G}_{\rm M}^{\bf 1}$ is a full row rank matrix, it indicates that
\begin{equation} \label{e.theorem1_1}
  \alpha^{l_{1,1}} \oplus \alpha^{l_{2,1}} \oplus \cdots \oplus  \alpha^{l_{M,1}} \neq {\bf 0}.
\end{equation}

First, we take the element $\alpha^{l_{1,0}}$ of the AI-CWEP $C_1^{} = (\alpha^{l_{1,0}}, \alpha^{l_{1,1}})$ instead of the element $\alpha^{l_{1,1}}$ into the left side of the inequality in (\ref{e.theorem1_1}), and the other elements are kept the same. Then, we have
  \begin{equation} \label{e.theorem1_2}
    \begin{aligned}
   \alpha^{s_1} &= \alpha^{l_{1,0}} \oplus \left(\alpha^{l_{2,1}} \oplus \cdots \oplus  
                 \alpha^{l_{J,1}}\right) \Rightarrow\\
   \alpha^{s_1} \oplus \alpha^{l_{1,1}} &= \left(\alpha^{l_{1,0}} \oplus \alpha^{l_{1,1}}\right) \oplus
   \left(\alpha^{l_{2,1}} \oplus \cdots \oplus \alpha^{l_{J,1}}\right),\\
              &\overset{(a)}{=} \alpha^{l_{2,1}} \oplus \cdots \oplus \alpha^{l_{J,1}},
    \end{aligned}
  \end{equation}
where (a) is deduced based on the definition given by (\ref{e.G_AIEP_def}).
Because of the full rank constraint, it is able to know
\begin{equation} \label{e.theorem1_3}
  \begin{aligned}
  \alpha^{l_{2,1}} \oplus \cdots \oplus \alpha^{l_{M,1}} &\neq \alpha^{l_{1,1}},\\
  \end{aligned}
\end{equation}
which is taken into (\ref{e.theorem1_2}). Then, it is able to derive that
  \begin{equation} \label{e.theorem1_4}
    \begin{aligned}
   \alpha^{s_1} \oplus \alpha^{l_{1,1}} \neq \alpha^{l_{1,1}} \Rightarrow  \alpha^{s_1} \neq {\bf 0}.
    \end{aligned}
  \end{equation}
indicating the combination $\alpha^{l_{1,0}}, \alpha^{l_{2,1}},\ldots, \alpha^{l_{M,1}}$ are also linear independent.

Next, we take the element $\alpha^{l_{2,0}}$ of the AI-CWEP $C_2^{} = (\alpha^{l_{2,0}}, \alpha^{l_{2,1}})$ instead of the element $\alpha^{l_{2,1}}$ into the right side of (\ref{e.theorem1_2}), and the other elements are kept the same. It shows 
  \begin{equation} \label{e.theorem1_5}
    \begin{aligned}
   \alpha^{s_2} &=  \left(\alpha^{l_{1,0}} \oplus \alpha^{l_{2,0}}\right)
                  \oplus \left(\alpha^{l_{3,1}} \oplus \cdots \oplus  
                 \alpha^{l_{M,1}}\right) \Rightarrow\\
   & \alpha^{s_2} \oplus \alpha^{l_{1,1}} \oplus \alpha^{l_{2,1}} 
      = \left(\alpha^{l_{1,0}} \oplus \alpha^{l_{1,1}}\right) \oplus \left(\alpha^{l_{2,0}} \oplus \alpha^{l_{2,1}}\right) \oplus \\
   & \left(\alpha^{l_{3,1}} \oplus \cdots \oplus \alpha^{l_{M,1}}\right)
              \overset{(b)}{=} \alpha^{l_{3,1}} \oplus \cdots \oplus \alpha^{l_{M,1}},
    \end{aligned}
  \end{equation}
where (b) is deduced based on (\ref{e.G_AIEP_def}).
Because of the full rank constraint, it is able to know
\begin{equation} \label{e.theorem1_6}
  \begin{aligned}
  \alpha^{l_{3,1}} \oplus \cdots \oplus \alpha^{l_{M,1}} &\neq \alpha^{l_{1,1}} \oplus \alpha^{l_{2,1}}.\\
  \end{aligned}
\end{equation}
Take (\ref{e.theorem1_6}) into (\ref{e.theorem1_5}), it is known that
  \begin{equation} \label{e.theorem1_7}
    \begin{aligned}
   \alpha^{s_2} \oplus \alpha^{l_{1,1}} \oplus \alpha^{l_{2,1}} 
   \neq \alpha^{l_{1,1}} \oplus \alpha^{l_{2,1}} \Rightarrow  \alpha^{s_2} \neq {\bf 0}.
    \end{aligned}
  \end{equation}
indicating the combination $\alpha^{l_{1,0}}, \alpha^{l_{2,0}}, \alpha^{l_{3,1}}, \ldots, \alpha^{l_{M,1}}$ are also linear independent. Repeat the substitution operation, and we can obtain the similarly process and results.

When all the elements $\alpha^{l_{j,1}}$ have been substituted by $\alpha^{l_{j,0}}$ for $1 \le j \le M$,
we set
  \begin{equation} \label{e.theorem1_8}
    \begin{aligned}
   \alpha^{s_M} = \alpha^{l_{1,0}} \oplus \alpha^{l_{2,0}} \oplus \cdots \oplus  
                 \alpha^{l_{M,0}} 
                \overset{(c)}{=} 
   2\cdot \left(\alpha^{l_{1,1}} \oplus \cdots \oplus \alpha^{l_{M,1}}\right)
          \neq {\bf 0},
    \end{aligned}
  \end{equation} 
where (c) is deduced because of Property 1.

Thus, if ${\rm Rank}\left({\bf G}_{\rm M}^{\bf 1}\right) = M$, any matrix in the generator matrix set ${\mathcal G}_{\rm M}$ is a full row rank matrix.
\end{proof}

Recall Example 2, the rank of ${\bf G}_{\rm M}^{\bf 1}$ of the AI-CWEP code $\Psi_{\rm ai,cw}$ is equal to $4$. Thus, the proposed AI-CWEP code $\Psi_{\rm ai,cw}$ of Example 2 is a $4$-user UD-AI-CWEP code.

From Theorems \ref{USPM_2m} and \ref{USPM_3m}, we observe that if the full-one generator matrix \( \mathbf{G}_{\rm M}^{\mathbf{1}} \) of the CWEP code \( \Psi_{\rm cw} \) has full row rank, then \( \Psi_{\rm cw} \) is a UD-CWEP code. In this context, we will refer to the full-one generator matrix \( \mathbf{G}_{\rm M}^{\mathbf{1}} \) of the UD-CWEP code \( \Psi_{\rm cw} \) as the \textit{generator matrix} for brevity.
In fact, various generator matrices \( \mathbf{G}_{\rm M}^{\mathbf{1}} \) can be directly designed to obtain a range of UD-CWEP codes. A well-designed generator matrix \( \mathbf{G}_{\rm M}^{\mathbf{1}} \) can significantly enhance performance in several areas, such as error performance, low-complexity detection, and more.

\subsection{Classification of EP Codes}

The loading factor, denoted as \( \eta \), is a parameter used to measure the resource utilization efficiency of EP codes. For an AIEP code over a prime field \( \text{GF}(p) \), the loading factor \( \eta \) is generally determined by the sizes of the FF-MUX \( \mathcal{A}_{\rm M} \). In contrast, for an EP code over an extension field \( \text{GF}(p^m) \), where \( p = 2, 3 \) and \( m \geq 2 \), the loading factor \( \eta \) is typically determined by the size of the generator matrix \( \mathbf{G}_{\rm M}^{\mathbf{1}} \). Consequently, the loading factor can be less than, equal to, or greater than one. Based on the different loading factors, the EP code \( \Psi \) can be classified into three modes. For clarity, we consider the EP encoder operating in serial mode as an example, while the parallel mode operates similarly, and will not be repeated here.
\begin{enumerate}
  \item
  \textbf{Error Correction EP Codes.}  
  The loading factor of the generator matrix \( \mathbf{G}_{\rm M}^{\mathbf{1}} \) for the S-CWEP code \( \Psi_{\rm cw} \) is less than one, i.e., \( \eta < 1 \). In this case, the number of served users \( M \) is less than the extension field size \( m \), i.e., \( M < m \). The proposed S-CWEP code \( \Psi_{\rm cw} \) may offer error correction (or detection) capabilities if its generator matrix is designed based on a channel code.
  \item
  \textbf{Orthogonal EP codes. }
  The loading factor of the generator matrix ${\bf G}_{\rm M}^{\bf 1}$ for the orthogonal UD-EP code $\Psi_{\rm o,B}$ over GF($2^m$) is equal to one, representing a typical orthogonal mode. In this case, the number of served users $M$ is equal to the extension field $m$, i.e., $M = m$. For example, in Example 2, the loading factor of the generator matrix is equal to one, i.e., $\eta = 1$, thus allowing the proposed AI-CWEP code $\Psi_{\rm ai,cw}$ to support orthogonal transmission mode.
  \item
  \textbf{Overload EP codes. }
  The \textit{overload EP codes} have loading factors greater than one, i.e., $\eta > 1$. In this scenario, the number of served users $M$ exceeds the extension field $m$, i.e., $M > m$. Typically, there are two main types of overload EP codes: the AIEP code constructed over the prime field GF($p$), and the \textit{non-orthogonal EP (NO-EP) code}. The former, the AIEP code, has been investigated in \cite{FFMA}, while the latter, the NO-EP code, will be introduced in the following section. The loading factor of the AIEP code $\Psi_{\rm s}$ over GF($p$) is in the range $1 \leq \eta \leq \log_2 (p-1)$, representing overloaded information on a \textit{symbol-wise} basis; in contrast, the NO-EP codes represent overloaded information on a \textit{codeword-wise} basis.
\end{enumerate}

\section{Construction of Codeword-wise EP Codes}

In the following, we begin by constructing UD-S-CWEP codes over GF($2^m$) based on linear block channel codes. Next, we introduce ternary matrices, specifically \textit{ternary orthogonal matrices} and \textit{ternary non-orthogonal matrices}, as generator matrices to develop AI-CWEP codes over GF($3^m$).

\subsection{Construction of UD-S-CWEP Codes Using Linear Block Channel Codes}

To construct UD-S-CWEP codes over \( \text{GF}(2^m) \), any classical linear block code can serve as a generator matrix, including LDPC codes \cite{LinBook3}, RM codes, Reed-Solomon (RS) codes \cite{Shu2009}, polar codes \cite{Polar2009}, and others. In this paper, we focus specifically on using LDPC codes to construct UD-S-CWEP codes, while polar codes will be explored in future work.

Let the generator matrix of an LDPC code ${\mathcal C}_{mc}$ be an $M \times m$ binary matrix ${\bf G}_{mc} \in {\mathbb B}^{M \times m}$. Then, we can set the generator matrix ${\bf G}_{\rm M}^{\bf 1}$ of an UD-S-CWEP code $\Psi_{\rm cw}$ as
\begin{equation}
  {\bf G}_{\rm M}^{\bf 1}  = {\bf G}_{mc}.
\end{equation}
Since the generator matrix ${\bf G}_{mc}$ of the linear channel code ${\mathcal C}_{mc}$ is of full row rank, ${\bf G}_{\rm M}^{\bf 1}$ satisfies the USPM constraint. Therefore, based on the linear block channel code ${\mathcal C}_{mc}$, we can construct an $M$-user UD-S-CWEP $\Psi_{\rm cw}$ over GF($2^m$) that has $2^M$ codewords. Notably, the extension factor (EF) $m$ of the UD-S-CWEP code $\Psi_{\rm cw}$ is equal to the codeword length of the linear block channel code ${\mathcal C}_{mc}$. Numerous studies have explored the construction of well-behaved linear block channel codes, which can be employed as UD-S-CWEP codes \cite{LinBook3}.

\textbf{Example 4:}
Suppose a $(16, 12)$ linear block code is used to construct a UD-S-CWEP code $\Psi_{\rm cw}$, and the generator matrix ${\bf G}_{eg}$ of the linear block code is given by (\ref{e.G}).

\vspace{-0.1in}
\begin{small}
\begin{equation} \label{e.G}
\begin{array}{ll}
  {\bf G}_{eg} = 
  \left[
  \begin{array}{cccccccccccc:cccc}
   1 & 0 & 0 & 0 & 0 & 0 & 0 & 0 & 0 & 0 & 0 & 0 & 1 & 0 & 0 & 0 \\
   0 & 1 & 0 & 0 & 0 & 0 & 0 & 0 & 0 & 0 & 0 & 0 & 0 & 1 & 0 & 0 \\
   0 & 0 & 1 & 0 & 0 & 0 & 0 & 0 & 0 & 0 & 0 & 0 & 0 & 0 & 1 & 0 \\
   0 & 0 & 0 & 1 & 0 & 0 & 0 & 0 & 0 & 0 & 0 & 0 & 0 & 0 & 0 & 1 \\
   0 & 0 & 0 & 0 & 1 & 0 & 0 & 0 & 0 & 0 & 0 & 0 & 0 & 0 & 0 & 1 \\
   0 & 0 & 0 & 0 & 0 & 1 & 0 & 0 & 0 & 0 & 0 & 0 & 1 & 0 & 0 & 0 \\
   0 & 0 & 0 & 0 & 0 & 0 & 1 & 0 & 0 & 0 & 0 & 0 & 0 & 1 & 0 & 0 \\
   0 & 0 & 0 & 0 & 0 & 0 & 0 & 1 & 0 & 0 & 0 & 0 & 0 & 0 & 1 & 0 \\
   0 & 0 & 0 & 0 & 0 & 0 & 0 & 0 & 1 & 0 & 0 & 0 & 0 & 0 & 1 & 0 \\
   0 & 0 & 0 & 0 & 0 & 0 & 0 & 0 & 0 & 1 & 0 & 0 & 0 & 0 & 0 & 1 \\
   0 & 0 & 0 & 0 & 0 & 0 & 0 & 0 & 0 & 0 & 1 & 0 & 1 & 0 & 0 & 0 \\
   0 & 0 & 0 & 0 & 0 & 0 & 0 & 0 & 0 & 0 & 0 & 1 & 0 & 1 & 0 & 0 \\
  \end{array}
  \right].
\end{array}
\end{equation}
\end{small}

Next, let the generator matrix ${\bf G}_{\rm M}^{\bf 1}$ of the UD-S-CWEP code $\Psi_{\rm cw}$ be ${\bf G}_{\rm M}^{\bf 1} = {\bf G}_{eg}$.
Hence, based on the $12 \times 16$ generator matrix ${\bf G}_{\rm M}^{\bf 1}$, we can obtain a $12$-user UD-S-CWEP code $\Psi_{\rm cw} = \{C_1^{}, C_2^{}, \ldots, C_{12}^{} \}$ over GF($2^{16}$), given as
\begin{equation*}
  \begin{array}{ll}
  C_1^{} = (0000000000000000, 1000000000001000),
  C_2^{} = (0000000000000000, 0100000000000100), \\
  C_3^{} = (0000000000000000, 0010000000000010), 
  C_4^{} = (0000000000000000, 0001000000000001), \\
  C_5^{} = (0000000000000000, 0000100000000001), 
  C_6^{} = (0000000000000000, 0000010000001000), \\
  C_7^{} = (0000000000000000, 0000001000000100), 
  C_8^{} = (0000000000000000, 0000000100000010), \\
  C_9^{} = (0000000000000000, 0000000010000010), 
  C_{10}^{} = (0000000000000000, 0000000001000001), \\
  C_{11}^{} = (0000000000000000, 0000000000101000), 
  C_{12}^{} = (0000000000000000, 0000000000010100). \\
   \end{array}
\end{equation*}
Each element of $C_{j}^{}$ for $1 \le j \le 12$ is a $16$-tuple, determined by the codeword length of the linear block channel code.
The Cartesian product of $C_1^{} \times C_2^{} \times \ldots \times C_{12}^{}$ can form an UD-S-CWEP code $\Psi$ with $2^{12}=4096$ EP codewords.
$\blacktriangle \blacktriangle$

\subsection{Ternary Orthogonal Matrices for Constructing AI-CWEP Codes}

We begin by defining the \textit{ternary orthogonal matrix} and its \textit{additive inverse matrix} over $\text{GF}(3^m)$, and we present several key properties of these matrices. Building upon the concept of ternary orthogonal matrices, we then construct AI-CWEP codes over $\text{GF}(3^m)$.

Let ${\bf T}_{\rm o}$ be an $m \times m$ matrix that satisfies the following two conditions:
\begin{enumerate}
\item 
All elements of ${\bf T}_{\rm o}$ belong to GF($3$), specifically $(0)_3$, $(1)_3$, and $(2)_3$; and
\item 
${\bf T}_{\rm o} \cdot {\bf T}_{\rm o}^{\rm T} = {\bf I}_m$ or ${\bf T}_{\rm o} \cdot {\bf T}_{\rm o}^{\rm T} = 2 \cdot {\bf I}_m$, where ${\bf I}_m$ is an $m \times m$ identity matrix.
\end{enumerate}
If both conditions are satisfied, ${\bf T}_{\rm o}$ is referred to as a \textit{ternary orthogonal matrix}, with the subscript ``o'' denoting ``orthogonal''.

Suppose ${\bf T}_{\rm o}(2, 2)$ is a $2 \times 2$ matrix over GF($3$), given as
\begin{equation}
  {\bf T}_{\rm o}(2, 2)  = \left[
    \begin{matrix}
      1 & 1 \\
      2 & 1 \\
    \end{matrix}
  \right].
\end{equation}
Based on the Property 1 of GF($3$), the \textit{two-fold Kronecker product} of ${\bf T}_{\rm o}(2, 2)$ is defined as

\vspace{-0.1in}
\begin{small}
\begin{equation}
  {\bf T}_{\rm o}(2^2, 2^2)  \triangleq \left[
    \begin{array}{cc}
      1 & 1 \\
      2 & 1 \\
    \end{array}
  \right] \otimes
  \left[
    \begin{array}{cc}
      1 & 1 \\
      2 & 1 \\
    \end{array}
  \right] =
   \left[
    \begin{array}{cccc}
      1 & 1 & 1 & 1\\
      2 & 1 & 2 & 1\\
      2 & 2 & 1 & 1\\
      1 & 2 & 2 & 1\\
    \end{array}
  \right],
\end{equation}
\end{small}
where $\otimes$ indicates Kronecker product.

The \textit{three-fold Kronecker product} of ${\bf T}_{\rm o}(2, 2)$ is defined as
\begin{equation}
\begin{aligned}
  {\bf T}_{\rm o}(2^3, 2^3) &\triangleq \left[
    \begin{array}{cc}
      1 & 1 \\
      2 & 1 \\
    \end{array}
  \right] \otimes
   \left[
    \begin{array}{cccc}
      1 & 1 & 1 & 1\\
      2 & 1 & 2 & 1\\
      2 & 2 & 1 & 1\\
      1 & 2 & 2 & 1\\
    \end{array}
  \right] 
  = \left[
    \begin{array}{cccccccc}
      1 & 1 & 1 & 1 & 1 & 1 & 1 & 1\\
      2 & 1 & 2 & 1 & 2 & 1 & 2 & 1\\
      2 & 2 & 1 & 1 & 2 & 2 & 1 & 1\\
      1 & 2 & 2 & 1 & 1 & 2 & 2 & 1\\
      2 & 2 & 2 & 2 & 1 & 1 & 1 & 1\\
      1 & 2 & 1 & 2 & 2 & 1 & 2 & 1\\
      1 & 1 & 2 & 2 & 2 & 2 & 1 & 1\\
      2 & 1 & 1 & 2 & 1 & 2 & 2 & 1\\
    \end{array}
  \right].
\end{aligned}
\end{equation}
Similarly, we can define the $\kappa$-fold Kronecker product of ${\bf T}_{\rm o}(2, 2)$ as,
\begin{equation}
  {\bf T}_{\rm o}(2^\kappa, 2^\kappa) = {\bf T}_{\rm o}(2, 2) \otimes 
                                        {\bf T}_{\rm o}(2^{\kappa-1}, 2^{\kappa-1}),
\end{equation}
which is a $2^\kappa \times 2^\kappa$ matrix. Obviously, ${\bf T}_{\rm o}(2^\kappa, 2^\kappa)$ satisfies the USPM structural property constraint of AI-CWEP codes over GF($3^m$).

Next, we define the \textit{additive inverse matrix of the ternary orthogonal matrix} by 
\begin{equation}
{\bf T}_{\rm o, ai}(2^\kappa, 2^\kappa) = {\bf P} - {\bf T}_{\rm o}(2^\kappa, 2^\kappa)
\overset{(a)}{=} 2 \cdot {\bf T}_{\rm o}(2^\kappa, 2^\kappa),
\end{equation}
where ($a$) is deduced based on the Property 1 of GF($3$).
For example, the $2 \times 2$ additive inverse matrix of the ternary orthogonal matrix ${\bf T}_{\rm o, ai}(2, 2)$ is given as
\begin{equation}
  {\bf T}_{\rm o, ai}(2, 2)  = {\bf P} - {\bf T}_{\rm o}(2, 2) = \left[
    \begin{matrix}
      2 & 2 \\
      1 & 2 \\
    \end{matrix}
  \right],
\end{equation}
and the $4 \times 4$ additive inverse matrix of the ternary orthogonal matrix ${\bf T}_{\rm o, ai}(2^2, 2^2)$ is
\begin{equation}
  {\bf T}_{\rm o, ai}(2^2, 2^2) = {\bf P} - {\bf T}_{\rm o}(2^2, 2^2) =
   \left[
    \begin{array}{cccc}
      2 & 2 & 2 & 2\\
      1 & 2 & 1 & 2\\
      1 & 1 & 2 & 2\\
      2 & 1 & 1 & 2\\
    \end{array}
  \right].
\end{equation}

Based on the definition of the ternary orthogonal matrices, some properties are presented as following.

\begin{property}
For the $\kappa$-fold ternary orthogonal matrix ${\bf T}_{\rm o}(2^\kappa, 2^\kappa)$ (or ${\bf T}_{\rm o}$ for simply) and its additive inverse matrix ${\bf T}_{\rm o, ai}(2^\kappa, 2^\kappa)$ (or ${\bf T}_{\rm o, ai}$ for simply), there are some basic properties.
\begin{enumerate}
  \item
  The cross-correlation of any two rows of ${\bf T}_{\rm o}$ is $(0)_3$.
  \item
  The self-correlation of any row of ${\bf T}_{\rm o}$ is $(1)_3$ or $(2)_3$.
  \item
  If ${\bf T}_{\rm o} \cdot {\bf T}_{\rm o}^{\rm T} = {\bf I}_{2^\kappa}$,
  then ${\bf T}_{\rm o} \cdot {\bf T}_{\rm o, ai}^{\rm T} = 2 \cdot {\bf I}_{2^\kappa}$;
  otherwise, 
  if ${\bf T}_{\rm o} \cdot {\bf T}_{\rm o}^{\rm T} = 2 \cdot {\bf I}_{2^\kappa}$,
  then ${\bf T}_{\rm o} \cdot {\bf T}_{\rm o, ai}^{\rm T} = {\bf I}_{2^\kappa}$.
\end{enumerate}
\end{property}

Then, based on the $\kappa$-fold ternary orthogonal matrix ${\bf T}_{\rm o}(2^{\kappa}, 2^{\kappa})$, we can construct an UD-AI-CWEP code over GF($3^m$), denoted by ${\Psi}_{\rm ai,T}$, i.e., ${\Psi}_{\rm ai,T} = \{C_1^{\rm cd}, C_2^{\rm cd}, \ldots, C_j^{\rm cd}, \ldots, C_M^{\rm cd}\}$ with $C_j^{\rm cd} = (\alpha^{l_{j,0}}, \alpha^{l_{j,1}})$ for $1 \le j \le M$ and $M \le m=2^{\kappa}$.
Let the generator matrix ${\bf G}_{\rm M}^{\bf 1}$ of the UD-AI-CWEP code ${\Psi}_{\rm ai,T}$ be ${\bf T}_{\rm o}(2^{\kappa}, 2^{\kappa})$, i.e.,
\begin{equation*}
  {\bf G}_{\rm M}^{\bf 1}  = {\bf T}_{\rm o}(2^{\kappa}, 2^{\kappa}),
\end{equation*}
and the full-zero generator matrix ${\bf G}_{\rm M}^{\bf 0}$ of ${\Psi}_{\rm ai,T}$ be ${\bf T}_{\rm o, ai}(2^{\kappa}, 2^{\kappa})$, i.e.,
\begin{equation*}
  {\bf G}_{\rm M}^{\bf 0}  = {\bf T}_{\rm o, ai}(2^{\kappa}, 2^{\kappa}),
\end{equation*}
where $m = 2^{\kappa}$. 
For the $j$-th AI-CWEP $C_j^{\rm cd} = (\alpha^{l_{j,0}}, \alpha^{l_{j,1}})$,
let $\alpha^{l_{j,0}}$ and $\alpha^{l_{j,1}}$ be the $j$-th row of ${\bf G}_{\rm M}^{\bf 0}$ 
and the $j$-th row of ${\bf G}_{\rm M}^{\bf 1}$, respectively. 
Then, we obtain totally $2^{\kappa}$ AI-CWEPs, which together form the UD-AI-CWEP code ${\Psi}_{\rm ai,T} = \{C_1^{\rm cd},C_2^{\rm cd}, \ldots, C_j^{\rm cd}, \ldots, C_{2^\kappa}^{\rm cd} \}$. 
The loading factor of the UD-AI-CWEP code $\Psi_{\rm ai,T}$ is equal to $\eta = 1$.
In fact, Example 2 is constructed using the $2$-fold ternary orthogonal matrix ${\bf T}_{\rm o}(4, 4)$ and its addtive inverse matrix ${\bf T}_{\rm o, ai}(4, 4)$, where $\kappa=2$.

In the following section, we will show that an UD-AI-CWEP code ${\Psi}_{\rm ai,T}$ in complex-field become as a Walsh code, which can play an role as orthogonal spreading code. Thus, the subscript ``cd'' of $C_{j}^{\rm cd}$ stands for ``code division''.

In addition, the UD-AI-CWEP code ${\Psi}_{\rm ai,T}$ can also be concatenated with a channel encoder, forming an FFMA with the channel-coding system.

\textbf{Example 5:}
We analyze the channel encoding process of the UD-AI-CWEP code ${\Psi}_{\rm ai,cw}$ presented in Example 2. The EP encoder operates in serial mode.
The channel code used is a $(16, 12)$ linear block code, with the generator matrix ${\bf G}_{eg}$ over $\text{GF}(3)$ in systematic form, as given in (\ref{e.G}).
Given that the generator matrix ${\bf G}_{eg}$ has sizes $12 \times 16$, the information section can accommodate $T = 3$ data blocks, with each block having a length of $m = 4$. The length of the parity section is $R = 16 - 12 = 4$.

Let the total number of users be $J = 3$, with each user having $K = 3$ bits. The bit sequences of the three users are:
${\bf b}_1 = (1, 1, 0)$, ${\bf b}_2 = (1, 0, 1)$ and ${\bf b}_3 = (0, 0, 1)$, respectively. We assign the AI-CWEPs to the three users as following:
\begin{equation*}
  {b}_{1,k} \rightleftharpoons C_1^{}= (2222, 1111), \quad
  {b}_{2,k} \rightleftharpoons C_2^{}= (1212, 2121), \quad
  {b}_{3,k} \rightleftharpoons C_3^{}= (1122, 2211),
\end{equation*}
where $0 \le k <3$.
Then, the element-sequences of three users are given as 
\begin{equation*} \label{e.Ex4_1}
  \begin{array}{c}
{\bf b}_1 = (1, 1, 0)  \Rightarrow
{\bf u}_1 = (1111, 1111, 2222)\\
{\bf b}_2 = (1, 0, 1) \Rightarrow
{\bf u}_2 = (2121, 1212, 2121)\\
{\bf b}_3 = (0, 0, 1) \Rightarrow
{\bf u}_3 = (1122, 1122, 2211)\\ 
  \end{array}.
\end{equation*}

Each element-sequence ${\bf u}_j$ is encoded into a codeword 
${\bf v}_j = {\bf u}_j {\bf G}_{gc}$. 
The codewords for the three users are as follows:
\begin{equation*}
  \begin{array}{c}
{\bf v}_1 = {\bf u}_1 \cdot {\bf G}_{eg}  =
(1111, 1111, 2222, \textcolor{blue}{1111})\\
{\bf v}_2 = {\bf u}_2 \cdot {\bf G}_{eg}  =
(2121, 1212, 2121, \textcolor{blue}{0000})\\
{\bf v}_3 = {\bf u}_3 \cdot {\bf G}_{eg}  =
(1122, 1122, 2211, \textcolor{blue}{0102})
  \end{array},
\end{equation*}
where the last $4$ symbols of each codeword form a \textit{parity block}.
The sum of ${\bf v}_1$, ${\bf v}_2$ and ${\bf v}_3$ is equal to
\begin{equation*}
{\bf v}_{sum} = \bigoplus_{j=1}^3 {\bf v}_j = (1021, 0112, 0221,\textcolor{blue}{1210}).
\end{equation*}

The FFSP sequence ${\bf w}$ and its encoded codeword ${\bf v}$ can be calculated as
\begin{equation*}
  \begin{aligned}
{\bf w} &= \bigoplus_{j=1}^3 {\bf u}_j = (w_0, w_1, w_2)=(1021, 0112, 0221), \\
{\bf v} &= {\bf w} \cdot {\bf G}_{eg} = ({\bf w}, {\bf v}_{\rm red}) = (1021, 0112, 0221,\textcolor{blue}{1210}).
  \end{aligned}
\end{equation*}

It is straightforward to verify that \( {\bf v} = {\bf v}_{sum} \), confirming that the order of finite-field addition \( \mathcal{A}_{\rm M} \) and channel encoding \( {\bf G}_{gc} \) does not affect the final encoded codewords. Additionally, the channel coding remains applicable to the ternary UD-AI-CWEA code.
$\blacktriangle \blacktriangle$

\subsection{Ternary Non-orthogonal Matrices for Constructing AI-CWEP Codes}

Unlike UD-CWEP codes, a \textit{non-orthogonal CWEP (NO-CWEP)} code $\Psi_{\rm no}$ does not satisfy the USPM structural constraint, where the subscript ``no'' stands for ``non-orthogonal''. Nevertheless, it is a type of overloaded EP code, known for providing high spectral efficiency (SE). 
In this subsection, we construct a type of non-orthogonal AI-CWEP (NO-AI-CWEP) codes, which are derived from \textit{ternary non-orthogonal matrices} constructed over an extension field GF($3^m$) of the prime field GF(3), where $m \ge 2$.

Let an $M \times m$ ternary non-orthogonal matrix over GF($3^m$) be denoted as ${\bf T}_{\rm no}(M, m)$, given by
\begin{equation} \label{e.T_no_general}
{\bf T}_{\rm no}(M, m) = \left[
  \begin{array}{ccccc}
    a_{1,0} & a_{1,1} & \ldots & a_{1,m-1}\\
    a_{2,0} & a_{2,1} & \ldots & a_{2,m-1}\\
    \vdots  & \vdots  & \ddots & \vdots  \\
    a_{M,0} & a_{M,1} & \ldots & a_{M,m-1}\\
  \end{array}
\right],
\end{equation}
where $a_{j,i} \in {\mathbb T}$ for $1 \le j \le M$, $0 \le i < m$ and $M > m$.
Define the additive inverse matrix of ${\bf T}_{\rm no}(M, m)$ by ${\bf T}_{\rm no, ai}(M, m)$, given as
\begin{equation} 
{\bf T}_{\rm no, ai}(M, m) = {\bf P} - {\bf T}_{\rm no}(M, m), 
\end{equation}
which is also an $M \times m$ ternary matrix over GF($3^m$).

Next, we construct an NO-AI-CWEP code $\Psi_{\rm no}$ based on the ternary non-orthogonal matrices ${\bf T}_{\rm no}(M, m)$ and ${\bf T}_{\rm no, ai}(M, m)$.
Let the full-one generator matrix ${\bf G}_{\rm M}^{\bf 1}$ and full-zero generator matrix ${\bf G}_{\rm M}^{\bf 0}$ of the NO-AI-CWEP code $\Psi_{\rm no}$ be
\begin{equation*}
  \begin{aligned}
  {\bf G}_{\rm M}^{\bf 1} &= {\bf T}_{\rm no}(M,m),\\
  {\bf G}_{\rm M}^{\bf 0} &= {\bf T}_{\rm no, ai}(M,m).
  \end{aligned}
\end{equation*}
Then, we can obtain an $M$-user NO-AI-CWEP code ${\Psi}_{\rm no} = \{C_1^{}, C_2^{}, \ldots, C_j^{}, \ldots, C_M^{}\}$ over GF($3^m$), where $C_j^{} = (\alpha^{l_{j,0}}, \alpha^{l_{j,1}})$ for $1 \le j \le M$.
In this case, $\alpha^{l_{j,0}}$ and $\alpha^{l_{j,1}}$ are set to be the $j$-th row of ${\bf G}_{\rm M}^{\bf 0}$ and ${\bf G}_{\rm M}^{\bf 1}$, respectively.
In addition, the loading factor of the NO-AI-CWEP code $\Psi_{\rm no}$ is equal to $\eta = M/m$. With the assumption $M > m$, it is able to know $\eta > 1$.

\textbf{Example 6:}
Now, we present a type of ternary non-orthogonal matrix constructed over GF($3^2$), which can be used to form an NO-AI-CWEP code.
Referring to \cite{UD_CDMA1_2012, UD_CDMA2_2012}, let the $3 \times 2$ matrix ${\bf T}_{\rm no}(3, 2)$ over GF($3^2$) be defined as follows:
\begin{equation} \label{e.T_no}
{\bf T}_{\rm no}(3, 2) = \left[
  \begin{array}{cc}
    {\bf T}_{\rm o}(2, 2)\\
    \hdashline
    {\bf E}\\
  \end{array}
\right] =
\left[
  \begin{array}{cc}
    1 & 1\\
    2 & 1\\
    \hdashline
    0 & 1\\
  \end{array}
\right].
\end{equation}
The additive inverse matrix of ${\bf T}_{\rm no}(3, 2)$ is defined as ${\bf T}_{\rm no, ai}(3, 2)$, 
given by
\begin{equation} 
{\bf T}_{\rm no, ai}(3, 2) = {\bf P} - {\bf T}_{\rm no}(3, 2) =
\left[
  \begin{array}{cc}
    2 & 2\\
    1 & 2\\
    \hdashline
    0 & 2\\
  \end{array}
\right],
\end{equation}
where $(0)_3 = (3)_3$ for GF($3$).

Suppose ${\bf G}_{\rm M}^{\bf 1} = {\bf T}_{\rm no}(3,2)$ and ${\bf G}_{\rm M}^{\bf 0} = {\bf T}_{\rm no, ai}(3,2)$. We can construct a $3$-user NO-AI-CWEP code ${\Psi}_{\rm no} = \{C_1^{}, C_2^{}, C_3^{}\}$ over GF($3^2$), where $C_j^{} = (\alpha^{l_{j,0}}, \alpha^{l_{j,1}})$ for $1 \le j \le 3$.
Here, $\alpha^{l_{j,0}}$ and $\alpha^{l_{j,1}}$ correspond to the $j$-th rows of ${\bf G}_{\rm M}^{\bf 0}$ and ${\bf G}_{\rm M}^{\bf 1}$, respectively. Consequently, the AI-CWEPs of the $3$-user NO-AI-CWEP code ${\Psi}_{\rm no}$ are given by:
\begin{equation*}
C_1^{} = (22, 11), \qquad
C_2^{} = (12, 21), \qquad
C_3^{} = (02, 01).
\end{equation*} 
The loading factor of the proposed NO-AI-CWEP code is equal to $\eta = 1.5$, which is also referred to as a overload codeword-wise EP code.

Next, we enumerate all possible combinations of the $3$ users, including the user block ${\bf b}[k]$, output element $u_{j,k}$, and FFSP block $w_k$, as shown in Table I.
From Table I, we observe that ${\bf b}[k] = (000)_2$ and ${\bf b}[k] = (111)_2$ share the same FFSP block, $w_k = (00)_3$. This indicates that the proposed AI-CWEP code ${\Psi}_{\rm no}$ does not satisfy the USPM constraint, and therefore, it is not an UD-AI-CWEP code.

\begin{table} [t] \label{tab.1}
  \begin{center}
    \caption{Combinations of the NO-AI-CWEP code $\Psi_{\rm no}$.}
    \begin{tabular}{c| c c c | c | c}
    \hline\\[-5.0mm]\hline
    ${\bf b}[k]$ & $u_{1,k}$ & $u_{2,k}$ & $u_{3,k}$ & FFSP $w_k$ & CFSP ${\bf r}^{(k)}$\\
    \hline
   $(000)_2$                  & $(22)_3$ & $(12)_3$ & $(02)_3$ & \textcolor{blue}{$(00)_3$} & $0, -3$\\
   $(\textcolor{red}{1}00)_2$ & $(11)_3$ & $(12)_3$ & $(02)_3$ & $(22)_3$ & $+2, -1$\\
   $(0\textcolor{red}{1}0)_2$ & $(22)_3$ & $(21)_3$ & $(02)_3$ & $(12)_3$ & $-2, -1$\\
   $(00\textcolor{red}{1})_2$ & $(22)_3$ & $(12)_3$ & $(01)_3$ & $(02)_3$ & $0, -1$\\
  \hdashline
  $(\textcolor{red}{111})_2$ & $(11)_3$ & $(21)_3$ & $(01)_3$ & \textcolor{blue}{$(00)_3$} & $0, +3$\\
  $(0\textcolor{red}{11})_2$ & $(22)_3$ & $(21)_3$ & $(01)_3$ & $(11)_3$ & $-2, +1$\\
  $(\textcolor{red}{1}0\textcolor{red}{1})_2$ & $(11)_3$ & $(12)_3$ & $(01)_3$ & $(21)_3$ & $+2, +1$\\
  $(\textcolor{red}{11}0)_2$ & $(11)_3$ & $(21)_3$ & $(02)_3$ & $(01)_3$ & $0, +1$\\
    \hline\\[-5.0mm]\hline
    \end{tabular}
  \vspace{-0.35in}
  \end{center}
\end{table}


Although the proposed NO-AI-CWEP code ${\Psi}_{\rm no}$ may not be separable in a finite-field, it can be separated in a complex-field or through certain network structures, e.g., a 3-dimensional butterfly network, as will be introduced in the subsequent sections. $\blacktriangle \blacktriangle$

\section{An FFMA over GF($3^m$) system in a GMAC}

This section introduces the transceiver of an FFMA system over GF($3^m$) in a GMAC. The block diagram of the system model is shown in Fig. \ref{f_SystemModel}. The transmitter consists of an EP encoder, a channel code encoder ${\bf G}_{gc}$, a \textit{finite-field to complex-field transform function} ${\rm F}_{{\rm F2C}}$, and a \textit{power allocation (PA)} module.
Compared to the transmitter described in \cite{FFMA}, the transmitter in this paper introduces a new module: the PA module. By adding the PA module, we can achieve \textit{polarization adjustment (PA)}, resulting in a PA-FFMA system. In fact, the transmitter with the PA module represents a more general system model of the FFMA system.

The framework of the FFMA with the channel-coding system is presented in Sect. III. As mentioned earlier, this framework is divided into two sections: the information section and the parity section. The information section consists of $T$ data blocks, with the indices of the data blocks ranging from $1, 2, \ldots, t, \ldots, T$.
Next, the $t$-th data block (for $1 \leq t \leq T$) is assigned an AI-CWEP code $\Psi$, which operates in serial mode. 
The AI-CWEP code $\Psi$ is constructed over $\text{GF}(3^m)$, with its full-one and full-zero generator matrices $\mathbf{G}_{\rm M}^{\bf 1}$ and $\mathbf{G}_{\rm M}^{\bf 0}$ defined as follows:
\[
  \begin{array}{cc}
    \mathbf{G}_{\rm M}^{\bf 1} = \mathbf{T}(M,m), \quad
    \mathbf{G}_{\rm M}^{\bf 0} = \mathbf{T}_{\rm ai}(M,m) = \mathbf{P} - \mathbf{T}(M,m),
  \end{array}
\]
where $\mathbf{T}(M,m)$ is an $M \times m$ ternary matrix with elements from $\mathbb{T}$, and $\mathbf{P}$ is an $M \times m$ matrix where all elements are equal to $p = 3$.

To illustrate the transmission process, we use the $t$-th data block (for $1 \leq t \leq T$) as an example. Suppose that the \( t \)-th data block can support \( J \) users, each transmitting a single bit, where \( J \leq M \).

\begin{figure}[t] 
  \centering
  \includegraphics[width=0.97\textwidth]{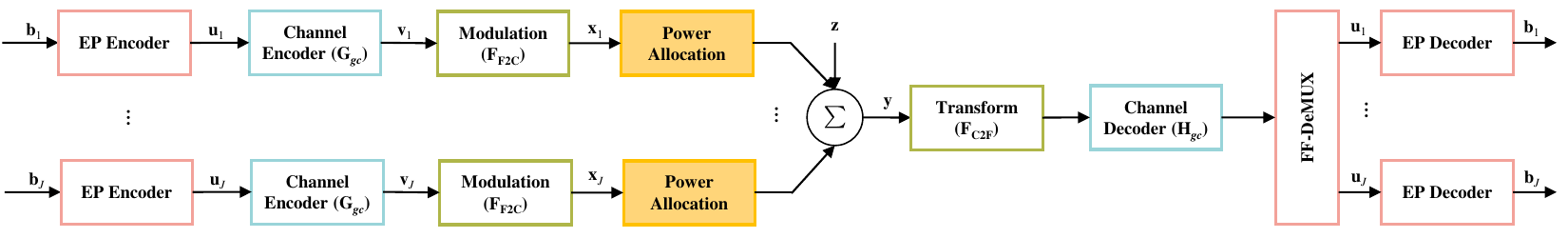}
 \caption{A block diagram of an FFMA system in a GMAC. The transmitter is consisted of an EP encoder, a channel code encoder ${\bf G}_{gc}$, a {finite-field to complex-field transform function} ${\rm F}_{{\rm F2C}}$, and {power allocation}.} 
 \label{f_SystemModel}
\end{figure}

\subsection{Transmitter of a PA-FFMA System}

Let the bit-sequence, consisting of a single bit, for the $j$-th user at the $t$-th block be denoted as ${\mathbf b}_j^{(t)}$, where $1 \leq j \leq J$. This bit-sequence ${\mathbf b}_j^{(t)}$ is then processed by the EP encoder, which operates in serial mode. After encoding, the bit-sequence ${\mathbf b}_j^{(t)}$ is transformed into the encoded element ${\mathbf u}_j^{(t)}$.
Since the encoded element ${\mathbf u}_j^{(t)}$ corresponds only to the $t$-th data block, we define the element sequence as:
\[
{\mathbf u}_{j,\rm D}^{(t)} = \left( {\mathbf 0}, \ldots, {\mathbf 0}, {\mathbf u}_j^{(t)}, {\mathbf 0}, \ldots, {\mathbf 0} \right),
\]
where ${\mathbf u}_j^{(t)}$ is placed at the $t$-th data block, and all other blocks are filled with zero vectors. Here, ${\mathbf 0}$ denotes a $1 \times m$ zero vector.

Subsequently, the sequence ${\mathbf u}_{j,\rm D}^{(t)}$ is encoded using an $(N, mT)$ linear block channel code ${\mathcal C}_{gc}$. The generator matrix $\mathbf{G}_{gc}$ of the code ${\mathcal C}_{gc}$ is structured in systematic form, consisting of two parts: the information section $\mathbf{I}_{\rm inf}$ and the parity section $\mathbf{F}_{\rm red}$. The subscripts ``inf'' and ``red'' refer to ``information'' and ``redundancy'', respectively.
The information section $\mathbf{I}_{\rm inf}$ is an $mT \times mT$ identity matrix, organized into a $T \times T$ array, where each element is either an $m \times m$ identity matrix $\mathbf{I}_{m}$ or an $m \times m$ zero matrix $\mathbf{0}$.
The parity section $\mathbf{F}_{\rm red}$ is structured as a $T \times 1$ array, defined as
$\mathbf{F}_{\rm red} = \left[ \mathbf{f}_{1}, \mathbf{f}_{2}, \ldots, \mathbf{f}_{t}, \ldots, \mathbf{f}_{T} \right]^{\rm T}$,
where each $\mathbf{f}_{t}$ is an $m \times R$ matrix for $1 \leq t \leq T$. Therefore, the size of $\mathbf{F}_{\rm red}$ is $mT \times R$.
Hence, the generator matrix in systematic form, $\mathbf{G}_{gc}$, can be expressed as:
\begin{equation} \label{e.gc_sym}
  \begin{aligned}
  {\bf G}_{gc} = [{\bf I}_{\rm inf} | {\bf F}_{\rm red}]  
  = \left[
  \begin{array}{ccccc:c}
  {\bf I}_{m} & {\bf 0}      & \ldots & {\bf 0}     & {\bf 0}     & {\bf f}_{1}  \\
  {\bf 0}     & {\bf I}_{m}  & \ldots & {\bf 0}     & {\bf 0}     & {\bf f}_{2}  \\
  \vdots      & \vdots       & \ddots & \vdots      & \vdots      & \vdots  \\
  {\bf 0}     & {\bf 0}      & \ldots & {\bf I}_{m} & {\bf 0}     & {\bf f}_{T-1}  \\
  {\bf 0}     & {\bf 0}      & \ldots & {\bf 0}     & {\bf I}_{m} & {\bf f}_{T}  \\
  \end{array}
  \right].
  \end{aligned}
\end{equation}

The encoded codeword ${\mathbf v}_{j}^{(t)}$ of the $j$-th user at the $t$-th data block can be written as
\begin{equation}
    \begin{aligned}
   {\mathbf v}_j^{(t)} = {\mathbf u}_{j,\rm D}^{(t)} \cdot {\mathbf G}_{gc}
        = ({\mathbf 0}, \ldots, {\mathbf 0}, {\mathbf u}_{j}^{(t)}, {\mathbf 0}, \ldots, {\mathbf 0}, 
                \textcolor{blue}{{\mathbf v}_{j, \rm red}^{(t)}}),
 \end{aligned}
\end{equation}
which consists of $T$ data blocks and a parity section, ${\mathbf v}_{j, \rm red}^{(t)} = {\mathbf u}_j^{(t)} \cdot {\mathbf F}_{\rm red}$. Here, the element ${\mathbf u}_{j}^{(t)}$ is placed in the $t$-th data block, with all other positions being zero.
For the sake of discussion, we also define $\mathbf{v}_{j,}^{(t)}$ in symbol-level form as $\mathbf{v}_{j}^{(t)} = (v_{j,0}^{(t)}, v_{j,1}^{(t)}, \ldots, v_{j,n}^{(t)}, \ldots, v_{j,N-1}^{(t)})$, which is a $1 \times N$ vector and $v_{j,n}^{(t)} \in {\mathbb T}$. 

Next, we modulate the encoded codeword ${\bf v}_{j}^{(t)}$ into a complex-field signal sequence ${\bf x}_{j}^{(t)}$, using the \textit{finite-field to complex-field transform function}, denoted as ${\rm F}_{\rm F2C}$. 
We can define ${\bf x}_{j}^{(t)}$ in symbol-level form as
${\bf x}_{j}^{(t)} = (x_{j,0}^{(t)}, x_{j,1}^{(t)}, \ldots, x_{j,n}^{(t)}, \ldots, x_{j,N-1}^{(t)})$, which is a $1 \times N$ vector.
Considering each symbol $v_{j,n}^{(t)} \in {\mathbb T}$, we utilize 3ASK (Amplitude Shift Keying) modulation as the transform function ${\rm F}_{\rm F2C}$, defined as follows:
\begin{equation} \label{e.3ASK}
x_{j,n}^{(t)} = {\rm F}_{\rm F2C}(v_{j,n}^{(t)}) =
\left\{
  \begin{matrix}
  +1,   & v_{j,n}^{(t)} = (1)_3 \\
   0,   & v_{j,n}^{(t)} = (0)_3 \\
  -1,   & v_{j,n}^{(t)} = (2)_3 \\
  \end{matrix},
\right.
\end{equation}
where $1 \le t \le T$, $1 \le j \le J$ and $0 \le n < N$. 
Similarly, the signal sequence ${\bf x}_{j}^{(t)}$ can be divided into $T$ signal data blocks and one signal parity section, i.e., 
${\bf x}_{j}^{(t)} = ({\bf 0}, \ldots, {\bf 0}, {\bf s}_{j}^{(t)}, {\bf 0}, \ldots, {\bf 0}, \textcolor{blue}{{{\bf x}_{j, \rm red}}})$, where ${\bf s}_{j}^{(t)}$ is a $1 \times m$ signal data block given by 
${\bf s}_{j}^{(t)} = {\rm F}_{\rm F2C}({\bf u}_{j}^{(t)})$.

Next, the modulated signal sequence \( \mathbf{x}_j^{(t)} \) undergoes power allocation. We define the \textit{polarization-adjusted vector (PAV)} in symbol-level form as \( \mu_{\text{pav}} = (\mu_0, \mu_1, \dots, \mu_n, \dots, \mu_{N-1}) \), which is a \( 1 \times N \) vector. In block-level form, the PAV is given by $\mu_{\text{pav}} = (\mathbf{0}, \dots, \mathbf{0}, \mu_{\text{inf}}^{(t)}, \mathbf{0}, \dots, \mathbf{0}, \mu_{\text{red}})$,
where \( \mu_{\text{inf}}^{(t)} \) is a \( 1 \times m \) non-negative vector representing the power allocated to the $t$-th data block, \( \mu_{\text{red}} \) is a \( 1 \times R \) non-negative vector representing the power allocated to the parity section, and \( \mathbf{0} \) denotes a \( 1 \times m \) zero vector. Let \( P_{\text{avg}} \) denote the average transmit power per symbol. To ensure that the total transmit power remains constant at \( N \cdot P_{\text{avg}} \), the following condition must hold:
\begin{equation}
  N = \|\mu_{\rm inf}^{(t)}\|_1 + \|\mu_{\rm red}\|_1,
\end{equation}
where $\| \cdot \|_1$ denotes the 1-norm.
Next, the transmit signal ${\bf x}_{j}^{(t)}$ is given by
\begin{equation}
  {\bf x}_{j, \rm PA}^{(t)} =  \sqrt{\mu_{\rm pav} P_{\rm avg}}  \circ 
                                {\bf x}_{j}^{(t)}
  = ({\bf 0}, \ldots, {\bf 0}, 
   \sqrt{\mu_{\rm inf}^{(t)} P_{\rm avg}} \circ {\bf s}_{j}^{(t)}, 
   {\bf 0}, \ldots, {\bf 0}, 
   \sqrt{\mu_{\rm red} P_{\rm avg}} \circ \textcolor{blue}{{\bf x}_{j, \rm red}}),
\end{equation}
where $\circ$ denotes the Hadamard product. The resulting signal ${\bf x}_{j, \rm PA}^{(t)}$ is then transmitted over a GMAC.

\subsection{FFSP Sequence}
To recover the bit sequences of the \( J \) users at reciver, it is essential to know the FFSP sequence \( \mathbf{w} \). As discussed earlier, the FFSP sequence consists of \( T \) FFSP blocks, denoted as
$\mathbf{w} = ({w}^{(1)}, {w}^{(2)}, \dots, \\ {w}^{(t)}, \dots, {w}^{(T)})$, where \( {w}^{(t)} \) represents the \( t \)-th FFSP block for \( 1 \leq t \leq T \). 
The codeword of the FFSP sequence \( \mathbf{w} \), encoded using the channel code \( \mathcal{C}_{gc} \), is expressed as
\[
\mathbf{v} = \mathbf{w} \cdot \mathbf{G}_{gc} 
= \bigoplus_{j=1}^{J} \mathbf{v}_{j}^{(t)}
= (\mathbf{w}, \textcolor{blue}{\mathbf{v}_{\text{red}}}),
\]
where \( \mathbf{v}_{\text{red}} \) denotes the parity section, which is given by \( \mathbf{v}_{\text{red}} = \mathbf{w} \cdot \mathbf{F}_{\text{red}} \). In symbol-level form, the codeword \( \mathbf{v} \) is written as
\(
\mathbf{v} = (v_0, v_1, \dots, v_n, \dots, v_{N-1}),
\)
which is a \( 1 \times N \) vector.

By obtaining the FFSP sequence \( \mathbf{w} \), the user block can be recovered, allowing us to retrieve the bit sequences of the \( J \) users at the reciver.

\vspace{-0.1in}
\subsection{Receiver of a PA-FFMA System}

At the receiving end, the received signal ${\bf y}$ is represented in block-level form as ${\bf y} = ({\bf y}_{\rm inf}, {\bf y}_{\rm red})$, where ${\bf y}_{\rm inf} = ({\bf y}^{(1)}, {\bf y}^{(2)}, \ldots, {\bf y}^{(t)}, \ldots, {\bf y}^{(T)})$, or in symbol-level form as ${\bf y} = (y_0, y_1, \ldots, y_n, \ldots, y_{N-1}) \in {\mathbb C}^{1 \times N}$. This represents the combined outputs of the $J$ users plus noise, given by
\vspace{-0.1in}
\begin{equation} 
{\bf y} = \sum_{j=1}^{J} {\bf x}_{j, \rm PA}^{(t)} + {\bf z} 
= \sqrt{\mu_{\rm inf}^{(t)} P_{\rm avg}} \circ \sum_{j=1}^{J} {\bf x}_{j}^{(t)} + {\bf z},
\end{equation}
where ${\bf z} \in \mathbb{C}^{1 \times N}$ represents an AWGN sequence with distribution ${\mathcal N}(0, N_0/2)$. 

We define ${\bf r} = \sum_{j=1}^{J} {\bf x}_{j}^{(t)}$, where ${\bf r}$ denotes the \textit{complex-field sum-pattern (CFSP)} of the modulated signal sequences. The CFSP signal sequence ${\bf r}$ can be represented in symbol-level form as
\(
{\bf r} = (r_0, r_1, \ldots, r_n, \ldots, r_{N-1}) \in \mathbb{C}^{1 \times N}.
\)
The $n$-th component $r_n$ of ${\bf r}$ is given by $r_{n} = \sum_{j=1}^{J} x_{j,n}^{(t)}$ for $0 \le n < N$.

Alternatively, we can partition ${\bf r}$ into block-level form as ${\bf r} = ({\bf r}_{\rm inf}, {\bf r}_{\rm red})$, where ${\bf r}_{\rm inf}$ consists of $T$ CFSP data blocks, such as ${\bf r}_{\rm inf} = ({\bf r}^{(1)}, {\bf r}^{(2)}, \ldots, {\bf r}^{(t)}, \ldots, {\bf r}^{(T)})$, 
and a parity CFSP section, denoted ${\bf r}_{\rm red}$. 

\begin{figure}[t] 
  \centering
  \includegraphics[width=0.56\textwidth]{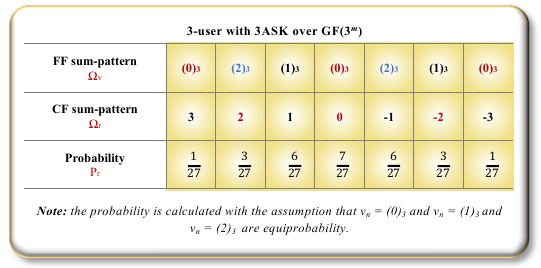}
 \caption{The relationship between CFSP signal $r_n$ and FFSP symbol $v_n$ of a $3$-user FFMA over GF($3^m$) system, where each user utilizes 3ASK.} 
 \label{f_TernaryMapTable}
 \vspace{-0.25in}
\end{figure}

Analogous to the BPSK signaling case, we also need to transform the received CFSP signal $r_n$ into its corresponding FFSP symbol $v_n$ using a \textit{complex-field to finite-field (C2F)} transform function $\rm F_{C2F}$, i.e., ${v}_n = {\rm F_{C2F}}({r}_n)$ for $0 \le n < N$. By examining the relationship between the CFSP signal $r_n$ and the FFSP symbol $v_n$, we can identify several key observations, as shown in Fig. \ref{f_TernaryMapTable}.
\begin{enumerate}
\item
The maximum and minimum values of $r_{n}$ are $J$ and $-J$, respectively. 
The set of $r_{n}$'s values in descending order is $\Omega_r = \{J, J-1, \ldots, -J\}$, in which the difference between two adjacent values is $1$. 
The total number of elements in $\Omega_r$ is equal to $|\Omega_r| = 2J+1$. 
\item
Since the values of $r_{n}$ are arranged in descending order, based on the transform function ${\rm F}_{\rm F2C}$, it is able to derive that the corresponding ternary set $\Omega_v$ of $\Omega_r$ is 
$\{(J)_3, (J-1)_3, (J-2)_3, \ldots\}$, in which $(0)_3$, $(2)_3$, and $(1)_3$ appear alternately,
i.e., $(0)_3 \mapsto (2)_3 \mapsto (1)_3 \mapsto (0)_3$.
When $r_n = 0$, the corresponding finite-field symbol $v_n$ is always $v_n = (0)_3$.
The number of $|\Omega_v|$ is also equal to $2J+1$, i.e., $|\Omega_v| = |\Omega_r|$.
\item
Assume $v_{n} = (0)_3, (1)_3$ and $(2)_3$ are equally likely, 
i.e., ${\rm Pr}(v_{n}=(0)_3) = {\rm Pr}(v_{n} = (1)_3) = {\rm Pr}(v_{n} = (2)_3) = 1/3$.
Let ${\iota}$ and ${\nu}$ denote the numbers of users which send ``$+1$'' and ``$-1$'', respectively.
Then, $J - \iota - \nu$ users send ``$0$''. 
In this case, the received CFSP signal $r_n$ is equal to $r_n = \iota - \nu$, whose probability is calculated as
\begin{equation*}
  {\mathcal P}_r(r_n = \iota - \nu) 
  = \frac{C_J^{\iota} \cdot C_{J-\iota}^{\nu}}{3^J},
\end{equation*}
where $0 \le \iota, \nu \le J$ and $\iota + \nu \le J$.
Hence, for a given $r_n$, we can calculate its probability ${\mathcal P}_r(r_n)$.
For $r_n = J, J-1, \ldots, -J$, we can obtain the corresponding probability set 
${\mathcal P}_r = \{{\mathcal P}_r(r_n=J), {\mathcal P}_r(r_n=J-1), \ldots, {\mathcal P}_r(r_n=-J)\}$, where $|{\mathcal P}_r| = |\Omega_r| = 2J + 1$.
\end{enumerate}

Based on the above facts, the transform function $\rm F_{C2F}$ maps each received CFSP signal $r_{n}$ to a unique FFSP symbol $v_{n}$, i.e., ${\rm F_{C2F}}: r_{n} \mapsto v_{n}$. 
With $r_{n} \in \{J, J-1, \ldots, -J\}$ and $v_{n} \in {\mathbb T}$, 
$\rm F_{C2F}$ is a \textit{many-to-one mapping} function.

Next, we calculate the posterior probabilities used for decoding ${\bf y}$. 
Since $y_{n}$ is determined by $r_{n}$ that is selected from the set $\Omega_r$, thus,
\begin{equation} \label{e_MAP}
  \begin{aligned}
  P(y_{n}) 
         &= \sum_{l=0}^{2J+1} {\mathcal P}_r(l) \cdot 
         \frac{1}{\sqrt{\pi N_0}} \exp \left\{
         - \frac{\left[y_{n} - \mu_{n} \Omega_r(l) \right]^2}{N_0} 
         \right\},  \\
  \end{aligned}
\end{equation}
where ${\mathcal P}_r(l)$ and $\Omega_r(l)$ stand for the $l$-th element in the sets ${\mathcal P}_r$ and $\Omega_r$, respectively. 
Thus, the posteriori probabilities of $v_{n} = \varsigma$ with $\varsigma \in {\mathbb T}$ are respectively given by

\vspace{-0.1in}
\begin{small}
\begin{equation}
  \begin{aligned}
%
   P(v_{n}=\varsigma |y_{n}) = \frac{\sum_{\Omega_v(l)= \varsigma} 
            {\mathcal P}_r(l)\cdot 
  \exp \left\{-\frac{\left[y_{n} - \mu_n \Omega_r(l) \right]^2}{N_0}\right\}}
  {{\sqrt{\pi N_0}} \cdot {P(y_{n})}}.\\
  \end{aligned}
\end{equation}
\end{small}

Next, \( P(v_n = \varsigma \mid y_n) \) is employed for decoding \( \mathbf{y} \). When LDPC codes are used as channel codes, the \( q \)-ary Sum Product Algorithm (QSPA) can be applied to decode the received signals and obtain the detected FFSP sequence
$\hat{\mathbf{w}} = (\hat{w}^{(1)}, \hat{w}^{(2)}, \ldots, \hat{w}^{(T)})$,
which consists of \( T \) FFSP blocks. Based on the obtained FFSP sequence \( \hat{\mathbf{w}} \) and the generator matrix \( \mathbf{G}_{\rm M}^{\mathbf{1}} \), the user block can then be recovered, which will further discussed in the following section.

In fact, we can easily extend the transmission process from the case where there are \( J \) users, each transmitting a single bit, to the case where there are \( J \) users, each transmitting \( K \) bits. We assume that the number of bits \( K \) for each user is less than or equal to the number of data blocks \( T \), i.e., \( K \leq T \). In this case, we can assign the \( k \)-th user-block \( {\mathbf b}[k] \) of \( J \) users to occupy the \( k \)-th data block of the information section, while the rest of the transmission process remains the same as described in the previous section. 
The corresponding equations can use either the superscript ``($k$)'' to denote the \( k \)-th data block or the subscript ``\( k \)'' to indicate the \( k \)-th user-block, both representing the same concept.

Note that, since the EP encoder operates in serial mode and the generator matrix \( \mathbf{G}_{gc} \) is in systematic form, the values and probabilities of the CFSP signals in the information section, specifically \( \Omega_r \) and \( \mathcal{P}_r \), can be calculated in an alternative manner.
As mentioned earlier, when the UD-AI-CWEP code \( \Psi_{\rm ai,T} \) is used in an FFMA system, the elements in \( \Psi_{\rm ai,T} \) consist solely of \( (1)_3 \) and \( (2)_3 \). In this case, the transform function \( {\rm F}_{\rm F2C} \) given by (\ref{e.3ASK}) simplifies to BPSK modulation.
\begin{equation} \label{e.BPSK}
x_{j,n}^{(t)} = {\rm F}_{\rm F2C}(v_{j,n}^{(t)}) =
\left\{
  \begin{matrix}
  +1,   & v_{j,n}^{(t)} = (1)_3 \\
  -1,   & v_{j,n}^{(t)} = (2)_3 \\
  \end{matrix}.
\right.
\end{equation}
Thus, we can directly apply the transform function ${\rm F}_{\rm C2F}$ for BPSK signaling, as discussed in \cite{FFMA}, to compute the values and probabilities of the CFSP signals for the information section.

\textbf{Example 7:}
Here, we provide an example to examine the transmission and reception processes of FFMA with channel coding. The encoded codewords \( {\mathbf v}_j \) are given in Example 5. The PAV \( \mu_{\rm pav} \) is set to be a \( 1 \times 16 \) full one vector, i.e., \( \mu_{\rm pav} = (1, 1, \dots, 1) \). Each codeword \( {\mathbf v}_j \) is then modulated into a signal sequence \( {\mathbf x}_j \) as follows:

\vspace{-0.2in}
\begin{small}
\begin{equation*}
  \begin{array}{lc}
{\bf v}_1 = 
(1111, 1111, 2222, \textcolor{blue}{1111}) \Rightarrow  
{\bf x}_1=\left((+1,+1,+1,+1), (+1,+1,+1,+1), (-1,-1,-1,-1), (+1,+1,+1,+1)\right);\\
{\bf v}_2 = 
(2121, 1212, 2121, \textcolor{blue}{0000}) \Rightarrow 
{\bf x}_2=\left((-1,+1,-1,+1), (+1,-1,+1,-1), (-1,+1,-1,+1), (0, 0, 0, 0)\right);\\
{\bf v}_3 = 
(1122, 1122, 2211, \textcolor{blue}{0102}) \Rightarrow 
{\bf x}_3=\left((+1,+1,-1,-1), (+1,+1,-1,-1), (-1,-1,+1,+1), (0, +1, 0, -1) \right).
  \end{array}
\end{equation*}
\end{small}

Then, the three modulated signal sequences ${\bf x}_1$,  ${\bf x}_2$, and  ${\bf x}_3$ are sent to the GMAC.
At the receiving end, assuming no effect of noise, the received CFSP signal sequence is

\vspace{-0.1in}
\begin{small}
\begin{equation*}
   \begin{array}{ll}
{\bf r} &= \sum_{j=1}^{3} {\bf x}_j = ({\bf r}_{\rm inf}, {\bf r}_{\rm red})
         = ({\bf r}_0, {\bf r}_1, {\bf r}_2, {\bf r}_{\rm red}) \\
        &= \left((+1, +3, -1, +1), (+3, +1, +1, -1), (-3, -1, -1, +1), (\textcolor{blue}{+1, +2, +1, 0})\right),
   \end{array}
\end{equation*}
\end{small}

\vspace{-0.05in}
which is then demodulated based on the designed transform function ${\rm F}_{\rm C2F}$. 
For $J = 3$, we have 
\begin{equation*}
\Omega_r = \{+3, +2, +1, 0, -1, -2, -3\} \quad \rm{and} \quad
\Omega_v = \{0, 2, 1, 0, 2, 1, 0\}. 
\end{equation*}
Hence, it is derived that 
${\rm F_{C2F}}(+3) = (0)_3$, 
${\rm F_{C2F}}(+2) = (2)_3$, 
${\rm F_{C2F}}(+1) = (1)_3$, 
${\rm F_{C2F}}(0) = (0)_3$, 
${\rm F_{C2F}}(-1) = (2)_3$, 
${\rm F_{C2F}}(-2) = (1)_3$, 
and ${\rm F_{C2F}}(-3) = (0)_3$.

Since no noise effect is assumed, this complex-field to finite-field transformation gives the following received sequence:
  \begin{equation*}
  \hat{\bf v} = {\rm F_{C2F}}({\bf r}) = 
  \left((1, 0, 2, 1), (0, 1, 1, 2), (0, 2, 2, 1), (\textcolor{blue}{1, 2, 1, 0})\right)_3.
  \end{equation*}
After removing the parity section, the decoded FFSP sequence \( \hat{\mathbf{w}} \) is given by
  \begin{equation*}
   \hat{\mathbf{w}} = (w_0, w_1, w_2) = (1021, 0112, 0221)_3,
  \end{equation*}
which consists of $3$ FFSP blocks, with each FFSP block being a 4-tuple. 
  
In the above, we assume that no noise affects the transmission. If the transmission is impacted by noise, the channel decoder must perform an error-correction process to estimate the transmitted sequence.  
$\blacktriangle \blacktriangle$

\section{Decoding of CWEP Codes}

This section investigates the decoding of the proposed CWEP codes, including the UD-S-CWEP codes constructed over GF($2^m$), UD-AI-CWEP codes constructed over GF($3^m$), and the NO-AI-CWEP code constructed over GF($3^2$).

\vspace{-0.1in}
\subsection{Decoding of UD-S-CWEP Codes}

As previously mentioned, we can utilize all classical linear block codes, such as LDPC codes, as the generator matrices ${\bf G}_{\rm M}^{\bf 1}$ for the UD-S-CWEP codes. In fact, the FFSP sequence can be viewed as an encoded codeword generated by passing the user block ${\bf b}_{\rm pll}$ through the LDPC encoder with generator matrix ${\bf G}_{\rm M}^{\bf 1}$.

Therefore, to decode a UD-S-CWEP code constructed over GF($2^m$), with its generator matrix designed based on an LDPC code, we can employ classical message passing algorithms (MPA), such as the sum-product algorithm (SPA), to de-multiplex the FFSP sequence, or alternatively, use the \textit{bifurcated minimum distance (BMD)} detection algorithm \cite{FFMA}. 
In general, classical message passing algorithms are more suitable for scenarios where the user is transmitting a relatively larger number of bits, while the BMD algorithm is better suited for cases involving a relatively smaller number of transmitted bits. 
These results will be further investigated in the simulation section.

\vspace{-0.1in}
\subsection{Decoding of UD-AI-CWEP Codes}

In Section V, we have constructed $\kappa$-fold ternary orthogonal matrix ${\bf T}_{\rm o}(2^{\kappa}, 2^{\kappa})$ over GF($3^m$). 
Now, we transform the finite-field $\kappa$-fold ternary orthogonal matrix ${\bf T}_{\rm o}(2^{\kappa}, 2^{\kappa})$ into complex-field form, 
by using the transform function ${\rm F}_{\rm F2C}$ given by (\ref{e.3ASK}).
It is easy to derive that, the complex-field form of ${\bf T}_{\rm o}(2^{\kappa}, 2^{\kappa})$, 
i.e., ${\bf C}_{\rm wal}(2^{\kappa}, 2^{\kappa}) = {\rm F}_{\rm F2C}\left({\bf T}_{\rm o}(2^{\kappa}, 2^{\kappa})\right)$, is an \textit{orthogonal Walsh code}, where the subscript ``wal'' stands for ``Walsh''. 
For $\kappa = 1, 2, 3$, it is shown that
\vspace{-0.1in}

\begin{small}
\begin{subequations}
  \begin{gather}
  {\rm F}_{\rm F2C}\left({\bf T}_{\rm o}(2, 2)\right)  = \left[
    \begin{matrix}
      +1 & +1 \\
      -1 & +1 \\
    \end{matrix}
  \right], \\
  {\rm F}_{\rm F2C}\left({\bf T}_{\rm o}(2^2, 2^2)\right)  = 
   \left[
    \begin{array}{cccc}
      +1 & +1 & +1 & +1\\
      -1 & +1 & -1 & +1\\
      -1 & -1 & +1 & +1\\
      +1 & -1 & -1 & +1\\
    \end{array}
  \right], \\
 {\rm F}_{\rm F2C}({\bf T}_{\rm o}(2^3, 2^3))  
 = \left[
    \begin{array}{cccccccc}
      +1 & +1 & +1 & +1 & +1 & +1 & +1 & +1\\
      -1 & +1 & -1 & +1 & -1 & +1 & -1 & +1\\
      -1 & -1 & +1 & +1 & -1 & -1 & +1 & +1\\
      +1 & -1 & -1 & +1 & +1 & -1 & -1 & +1\\
      -1 & -1 & -1 & -1 & +1 & +1 & +1 & +1\\
      +1 & -1 & +1 & -1 & -1 & +1 & -1 & +1\\
      +1 & +1 & -1 & -1 & -1 & -1 & +1 & +1\\
      -1 & +1 & +1 & -1 & +1 & -1 & -1 & +1\\
    \end{array}
  \right].
  \end{gather}
\end{subequations}
\end{small}

The orthogonal Walsh codes have been widely used in CDMA systems.
The \( j \)-th row of \( {\bf C}_{\rm wal} \), denoted as \( {\bf C}_{{\rm wal},j} \), represents the spreading code for the \( j \)-th user, where \( 1 \leq j \leq 2^{\kappa} \), with a spreading factor (SF) of \( 2^{\kappa} \).
Hence, without considering the channel code ${\mathcal C}_{gc}$, the FFMA over GF($3^m$) system utilizing the UD-AI-CWEP code $\Psi_{\rm ai,T}$ devolves into a classical CDMA system. 
In contrast, considering the channel code ${\mathcal C}_{gc}$, the FFMA over GF($3^m$) system based on the UD-AI-CWEP code $\Psi_{\rm ai,T}$ can form an \textit{error-correction orthogonal spreading code}.

There are two detection methods to recover the bit-sequences. One is by \textit{complex-field correlation operation}, and the other is by \textit{finite-field correlation operation}.
\begin{enumerate}
  \item
  Complex-field correlation operation. 
  As mentioned previously, for the $k$-th bit of the $j$-th user, denoted as $b_{j,k} \in \mathbb{B}$, it is assigned the AI-CWEP $C_j$. After passing through the EP encoder in serial mode, the output element ${\bf u}_{j}^{(k)}$, which is an $m$-tuple, is transformed into its complex-field representation ${\bf s}_{j}^{(k)}$, also an $m$-tuple, using the transform function ${\rm F}_{\rm F2C}$. 
  It is straightforward to derive that 
  ${\bf s}_{j}^{(k)} = {\rm F}_{\rm F2C}({\bf u}_{j}^{(k)}) = (2b_{j,k} - 1) \cdot {\bf C}_{{\rm wal},j}$.
  Therefore, the detected bit $\hat{b}_{j,k}$ is given by:
  \begin{equation}
  \hat{b}_{j,k} = \left\{
    \begin{array}{ll}
      (1)_2, & \text{if } {\bf y}^{(k)} \cdot {\bf C}_{{\rm wal},j} > \delta_{th}, \\
      (0)_2, & \text{if } {\bf y}^{(k)} \cdot {\bf C}_{{\rm wal},j} < -|\delta_{th}|, \\
    \end{array}
  \right.
  \end{equation}
  where $\delta_{th} \geq 0$ is the detection threshold. When the probabilities of $(0)_2$ and $(1)_2$ are equal, we can set $\delta_{th} = 0$.
  
  \item
  Finite-field correlation operation. 
  Based on the properties of the ternary orthogonal matrix ${\bf T}_{\rm o}$, we can also do finite-field correlation operation to the FFSP $w_k$ for recovering the transmit bits. If ${\bf T}_{\rm o} \cdot {\bf T}_{\rm o}^{\rm T} = {\bf I}_{2^\kappa}$, it is able to derive that
  \begin{equation}
  \hat{b}_{j,k}  = \left\{
    \begin{array}{ll}
      (1)_2, & {w}_k \cdot{\alpha}^{l_{j,1}} = (1)_3,\\
      (0)_2, & {w}_k \cdot{\alpha}^{l_{j,1}} = (2)_3.\\
    \end{array}
    \right.
  \end{equation}
  Otherwise, if ${\bf T}_{\rm o} \cdot {\bf T}_{\rm o}^{\rm T} = 2 \cdot {\bf I}_{2^\kappa}$, the recovered bits are given as 
  \begin{equation}
  \hat{b}_{j,k}  = \left\{
    \begin{array}{ll}
      (1)_2, & {w}_k \cdot{\alpha}^{l_{j,1}} = (2)_3,\\
      (0)_2, & {w}_k \cdot{\alpha}^{l_{j,1}} = (1)_3,\\
    \end{array}
    \right.
  \end{equation}
  where ${\alpha}^{l_{j,1}}$ is the $j$-th row of the orthogonal ternary matrix ${\bf T}_{\rm o}$.
\end{enumerate}

\textbf{Example 8:}
Recall Example 7, where we obtained the CFSP signal sequence ${\bf r}$ and its information section, i.e., ${\bf r}_{\rm inf} = ({\bf r}_0, {\bf r}_1, {\bf r}_2)$, as follows:
\begin{equation*}
{\bf r}_{\rm inf} = \left((+1, +3, -1, +1), (+3, +1, +1, -1), (-3, -1, -1, +1) \right).
\end{equation*}
If we exploit complex-field correlation operation, the detected bits are 
\begin{equation*} 
  \begin{array}{ll}
  \hat{\bf b}_1 = {\bf r}_{\rm inf} \cdot (+1, +1, +1, +1)
                = (+4, +4, -4) = (1, 1, 0)_2,\\
  \hat{\bf b}_2 = {\bf r}_{\rm inf} \cdot (-1, +1, -1, +1)
                = (+4, -4, +4) = (1, 0, 1)_2,\\
  \hat{\bf b}_3 = {\bf r}_{\rm inf} \cdot (-1, -1, +1, +1)
                = (-4, -4, +4) = (0, 0, 1)_2. \\
  \end{array}
\end{equation*}
If finite-field correlation operation is utilized to recover the bits, we should do correlation to the FFSP sequence, i.e., $\hat{\bf w} = ({w}_0, {w}_1, {w}_2) = (1021, 0112, 0221)$.
The finite-field correlation operation is given as 
\begin{equation*} 
  \begin{array}{ll}
  \hat{\bf b}_1 &= \hat{\bf w} \cdot (1111)_3 = (1021, 0112, 0221)_3 \cdot (1111)_3 
                = (1, 1, 2)_3 = (1, 1, 0)_2,\\
  \hat{\bf b}_2 &= \hat{\bf w} \cdot (2121)_3 = (1021, 0112, 0221)_3 \cdot (2121)_3 
                = (1, 2, 1)_3 = (1, 0, 1)_2,\\
  \hat{\bf b}_3 &= \hat{\bf w} \cdot (2211)_3 = (1021, 0112, 0221)_3 \cdot (2211)_3 
                = (2, 2, 1)_3 = (0, 0, 1)_2. \\
  \end{array}
\end{equation*}
Hence, both the complex-field and finite-field correlation operations can recover the transmit bits successfully.
$\blacktriangle \blacktriangle$

\subsection{Decoding of NO-CWEP Codes}

When the NO-CWEP code $\Psi_{\rm no}$ is used to support an FFMA system, the system is referred to as \textit{NOMA in finite-field (FF-NOMA)}. Without considering the channel code ${\mathcal C}_{gc}$, the FFMA system over GF($3^m$) based on the NO-CWEP code $\Psi_{\rm no}$ degenerates into a standard NOMA system, whose SE is determined by the loading factor of $\Psi_{\rm no}$, i.e., $\eta = M/m$ under the assumption that $M > m$. Thus, the FF-NOMA system can be viewed as an \textit{error-correction non-orthogonal spreading code}.

Although the FFSP blocks of the NO-CWEP code $\Psi_{\rm no}$ are non-orthogonal in the finite field, their corresponding CFSP signal blocks are distinct, as shown in Table I. From Table I, there is a one-to-one mapping between the transmitted user-block ${\bf b}[k]$ and the received CFSP signal block ${\bf r}^{(k)}$. Thus, by employing the \textit{maximum a posteriori (MAP)} criterion, we can recover the bit sequences in the complex field without ambiguity. Due to page limitations, we present only a special case of the NO-CWEP code and use the fundamental MAP detection algorithm.


\begin{figure*}[t]
  \centering
  \includegraphics[width=0.95\textwidth]{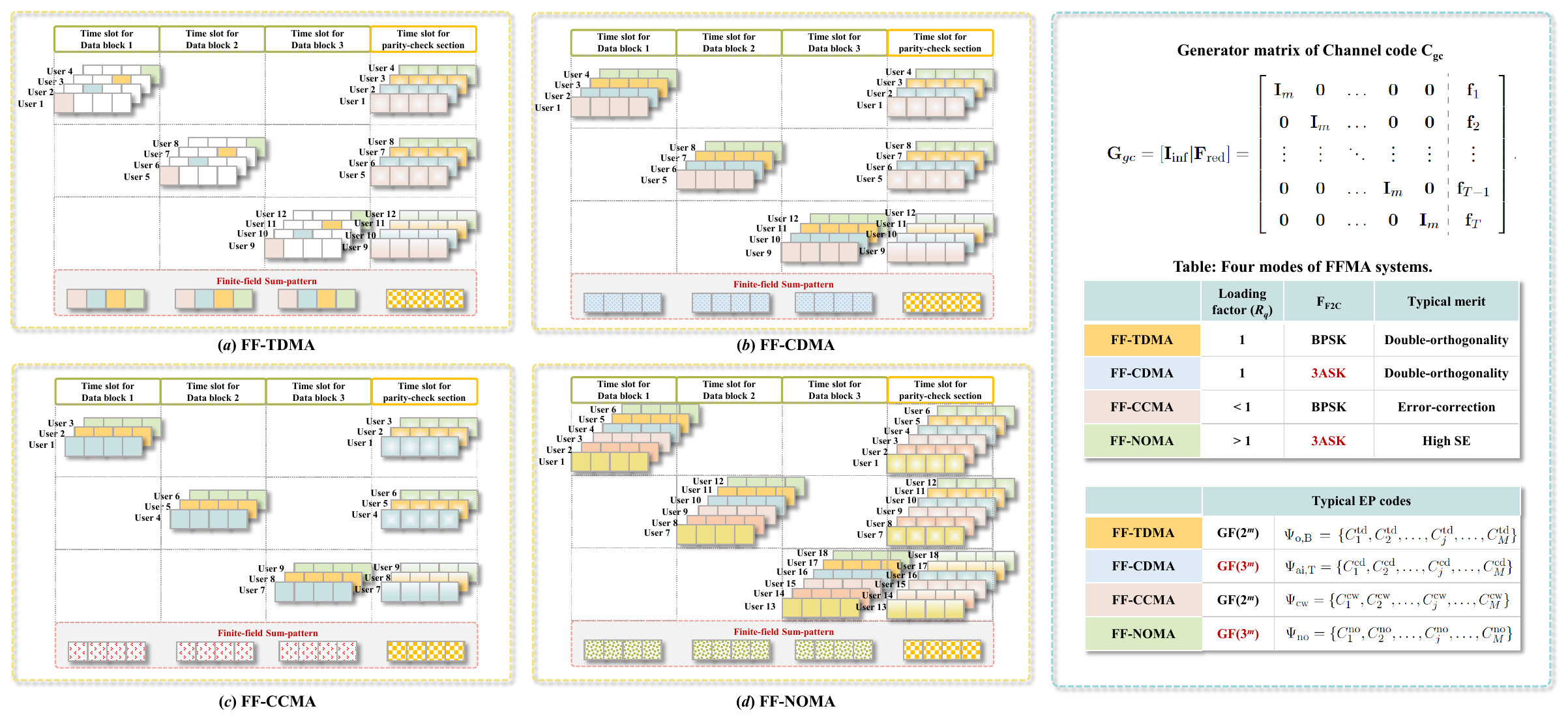}
  \caption{A diagram illustrating the four modes of FFMA systems: FF-TDMA, FF-CCMA, FF-CDMA, and FF-NOMA. Assume that a $(16, 12)$ linear block ${\mathcal C}_{gc}$ is used during transmission, with the parameters $m = 4$, $K =1$, $T = 3$, and $N = 16$.} 
  \label{f_summary}
  \vspace{-0.3in}
\end{figure*}

\section{A summary of FFMA systems}

In paper \cite{FFMA} and this paper, we totally introduce four modes of FFMA systems, they are FF-TDMA, FF-CCMA, FF-CDMA and FF-NOMA, as shown in Fig. \ref{f_summary}. Now, we summarize these four modes.

\subsection{FF-TDMA mode}

The loading factor of the generator matrix \( {\bf G}_{\rm M}^{\bf 1} \) in an FF-TDMA system is equal to one, i.e., \( \eta = 1 \). To support such an FF-TDMA system, the typical UD-EP code is constructed over \( \text{GF}(2^m) \), which is a type of symbol-wise EP code, denoted by \( {\Psi}_{\rm o,B} \), i.e., \( {\Psi}_{\rm o,B} = \{ C_1^{\rm td}, C_2^{\rm td}, \ldots, C_j^{\rm td}, \ldots, C_M^{\rm td} \} \), where \( C_j^{\rm td} = \alpha^{j-1} \cdot C_{\rm B} \) for \( 1 \le j \le M \) and \( M \le m \).

In addition, if the EP encoder operates in serial mode, the FF-TDMA system adopts a \textit{sparse-form (SF)} structure; whereas, if the EP encoder operates in parallel mode, the FF-TDMA system takes on a \textit{diagonal-form (DF)} structure. 

In \cite{FFMA}, we removed the default zeros in a DF-FFMA system to reduce the transmit power. Similarly, the default zeros in a SF-FFMA system can also be removed to achieve power reduction. By eliminating the default zeros in both SF-FFMA and DF-FFMA systems, it becomes possible to enable \textit{polarization-adjusted} configurations. This facilitates the optimization of channel capacity and the realization of PA-SF-FFMA and/or PA-DF-FFMA systems.
Unless otherwise specified, we adopt the PA configurations as the default in the following discussion. In other words, we use FFMA to refer to PA-FFMA for simplicity.

An FF-TDMA system exhibits \textit{double-orthogonality} in both finite-field and complex-field settings. This orthogonal property allows for flexible bit assignment to different users. Hence, the FF-TDMA mode serves as the foundation of an FFMA system.

Note that the FF-TDMA mode can be regarded as a special case of the FF-CCMA mode, since the orthogonal UD-EP code \( \Psi_{\rm o,B} \) is a special instance of the UD-CWEP code \( \Psi_{\rm cw} \). Therefore, we can summarize the above result in the following corollary:

\begin{corollary} \label{Coro_ff_TDMA}
  An FF-TDMA system can be viewed as a special case of the FF-CCMA system, which arises from the concatenation of a symbol-wise UD-EP code \( \Psi_{\rm o,B} \) and a channel code.
\end{corollary}

\vspace{-0.15in}
\subsection{FF-CCMA mode}
The loading factor of the generator matrix ${\bf G}_{\rm M}^{\bf 1}$ of an FF-CCMA system is smaller than one, i.e., $\eta < 1$. 
If the generator matrix ${\bf G}_{\rm M}^{\bf 1}$ satisfies the USPM structural property constraint given by Theorem \ref{USPM_2m}, we can construct an UD-S-CWEP code over GF($2^m$), denoted by ${\Psi}_{\rm cw}$, i.e., ${\Psi}_{\rm cw} = \{C_1^{\rm cw}, C_2^{\rm cw}, \ldots, C_j^{\rm cw}, \ldots, C_M^{\rm cw}\}$, with $C_j^{\rm cw} = ({\bf 0},\alpha^{l_{j,1}})$ for $1 \le j \le M$. The UD-S-CWEP codes are used to support FF-CCMA systems.

For the FF-CCMA mode, the EP encoder operates in parallel mode and may exhibit better error performance than in serial mode. This is because the data rate in serial mode is generally very low, i.e., $1/m$, which leads to a worse BER, as discussed in \cite{FFMA}. When the number of bits per user, $K$, is greater than one, the data rate in parallel mode becomes $K/m$, which is greater than $1/m$.
In general, when the number of information bits is comparable to the number of parity bits, classical message passing algorithms (MPA) can be used to decode the PA-FF-CCMA system. On the other hand, the BMD detection algorithm may be more suitable.

Additionally, in the FF-CCMA mode, the EP encoder \( \Psi_{\rm cw} \) is concatenated with the channel code \( {\mathcal C}_{gc} \) at the transmitter, which requires a joint decoding algorithm at the receiver for both the multiuser code \( {\mathcal C}_{mc} \) and the channel code \( {\mathcal C}_{gc} \).
We currently present two decoding algorithms: MPA and BMD algorithms. As a result, there are four possible joint decoding methods:

\begin{itemize}
  \item MPA for multiuser code and BMD for channel code, referred to as MPA-BMD for brevity;
  \item BMD for multiuser code and BMD for channel code, referred to as BMD-BMD for brevity;
  \item MPA for multiuser code and MPA for channel code, referred to as MPA-MPA for brevity; and
  \item BMD for multiuser code and MPA for channel code, referred to as BMD-MPA for brevity.
\end{itemize}

Although four decoding methods exist, the latter two, i.e., MPA-MPA and BMD-MPA, are less commonly used. Since $K$ typically represents a short payload in the um-MTC scenario, the channel coding often operates at a low data rate. Therefore, the MPA algorithm is not well-suited for decoding the channel code ${\mathcal C}_{gc}$ in such cases.

\vspace{-0.1in}
\subsection{FF-CDMA mode}
The loading factor of the generator matrix ${\bf G}_{\rm M}^{\bf 1}$ of an FF-CDMA system is equal to one, i.e., $\eta = 1$. To support such an FF-CDMA system, 
the typical UD-AI-CWEP code is constructed based on a ternary orthogonal matrix over GF($3^m$), denoted as ${\Psi}_{\rm ai,T}$, i.e., ${\Psi}_{\rm ai,T} = \{C_1^{\rm cd}, C_2^{\rm cd}, \ldots, C_j^{\rm cd}, \ldots, C_M^{\rm cd}\}$, with $C_j^{\rm cd} = (\alpha^{l_{j,0}}, \alpha^{l_{j,1}})$ for $1 \le j \le M$.
In our paper, $\alpha^{l_{j,0}}$ and $\alpha^{l_{j,1}}$ are respectively the $j$-th row of the additive inverse matrix of the $\kappa$-fold ternary orthogonal matrix ${\bf T}_{\rm o, ai}(2^{\kappa}, 2^{\kappa})$ and the $j$-th row of the $\kappa$-fold ternary orthogonal matrix ${\bf T}_{\rm o}(2^{\kappa}, 2^{\kappa})$, where $m = 2^{\kappa}$.

In this paper, the EP encoder in the FF-CDMA mode operates in serial mode, ensuring \textit{double orthogonality} in both the finite field and the complex field. This enables the use of \textit{complex-field correlation operations} and/or \textit{finite-field correlation operations} to decode the EP code.
Furthermore, the UD-AI-CWEP code $\Psi_{\rm ai,T}$ and the channel code ${\mathcal C}_{gc}$ together form a \textit{concatenated code}, which functions as an \textit{error-correction orthogonal spreading code}. Similar to the FF-CDMA mode, the MPA and BMD algorithms can be applied to decode the channel code.
In most cases, we first decode the multiuser code using correlation operations in the complex field, and then apply the MPA or BMD algorithms to decode the channel code.

In addition, when power allocation is applied to the FF-CDMA systems, referred to as PA-FF-CDMA systems, different UD-AI-CWEP codes may yield the same error performance. This is because the differences caused by complex-field correlation gain can become equal, as the power allocation method assigns more power to the case with less correlation gain. Furthermore, since the PA-FF-TDMA system also operates with equal power, the PA-FF-CDMA systems exhibit the same error performance as the PA-FF-TDMA systems. Based on this, we can derive the following corollary:
\begin{corollary} \label{Coro_ff_CCMA}
  A PA-FF-TDMA system can be realized using an FF-CDMA system, where the length of the UD-AI-CWEP (which is equivalent to the spreading factor) in the FF-CDMA system must match the assigned power level of the PA-FF-TDMA system. Under these conditions, the complex-field correlation at the receiver ensures that the error performance of the FF-CDMA system is equivalent to that of the PA-FF-TDMA system.
\end{corollary}

\vspace{-0.2in}
\subsection{FF-NOMA mode}
The loading factor of the generator matrix ${\bf G}_{\rm M}^{\bf 1}$ of an FF-NOMA system is larger than one, i.e., $\eta > 1$. To support an FF-NOMA system, 
this paper introduces a ternary non-orthogonal matrix ${\bf T}_{\rm no}(3,2)$ constructed over GF($3^2$), denoted by ${\Psi}_{\rm no}$, i.e., ${\Psi}_{\rm no} = \{C_1^{\rm no} = (22, 11), C_2^{\rm no} = (12, 21), C_3^{\rm no} = (02, 01)\}$.

Due to the non-orthogonal nature in the finite field, users must be separated in the complex field. Therefore, the EP encoder should operate in serial mode.  
Furthermore, there should be a one-to-one (or many-to-one) mapping between the CFSP signal block ${\bf r}^{(k)}$ and the user-block ${\bf b}[k]$, so that we can directly decode the multiuser code using the MAP criterion, or obtain the soft input (likelihood information) for further channel decoding. Similar to the FF-CDMA mode, either the MPA or BMD algorithms can be employed to decode the channel code.

In \cite{FFMA}, the FFMA system based on an AIEP code $\Psi_{\rm s}$ constructed over a prime field GF($p$) can also support FF-NOMA transmission, e.g., network FFMA.
An FF-NOMA system can improve the SE, with some loss in error performance.


\subsection{Number of Users Served under Different FFMA Modes}

From the above discussion, we can observe that the generator matrix set of an EP code determines the modes of FFMA, such as FF-TDMA, FF-CDMA, FF-CCMA, and FF-NOMA. In other words, the number of users served by the FFMA system is also governed by the loading factor of the generator matrix set, i.e., \( \eta = \frac{M}{m} \).  
Assuming that the information section of the channel code supports \( T \) data blocks (or \( T \) EP codes), according to Theorem \ref{Theorem_EP_ChannelCode}, we have
\[
  J \le \frac{MT}{K} = \frac{\eta m \cdot T}{K},
\]
where \( J \) represents the total number of users served by the FFMA system, with the EP encoder having a loading factor \( \eta \).

\textbf{Example 9:}
As shown in Fig. \ref{f_summary}, consider the use of an $(16, 12)$ linear block code for transmission, with parameters \( T = 3 \), \( m = 4 \), \( K = 1 \), and \( N = 16 \).  
For the FF-TDMA and/or FF-CDMA modes with a loading factor of \( \eta = 1 \), each data block can serve \( J_{mc} = \eta \cdot m = 4 \) users, meaning the system can support a maximum of \( J = J_{mc} \cdot T = 12 \) users.  
In the FF-CCMA mode, with a loading factor of \( \eta = \frac{3}{4} \), each data block can serve \( J_{mc} = \eta \cdot m = 3 \) users, so the system can support a maximum of \( J = J_{mc} \cdot T = 9 \) users.  
For the FF-NOMA mode, with a loading factor of \( \eta = 1.5 \), each data block can serve \( J_{mc} = \eta \cdot m = 6 \) users, and the system can support a maximum of \( J = J_{mc} \cdot T = 18 \) users.
$\blacktriangle \blacktriangle$



\section{Overload EP codes based network FFMA}

As aforementioned, the NO-EP code ${\Psi}_{\rm no}$ is not an UD-EP code, so it cannot be decoded based on the generator matrix ${\bf G}_{\rm M}^{\bf 1}$.
Nevertheless, when the NO-EP code ${\Psi}_{\rm no}$ is applied into network layer, 
we can recover the bit-sequences without ambiguity.

\begin{figure}[t]
  \centering
  \includegraphics[width=0.9\textwidth]{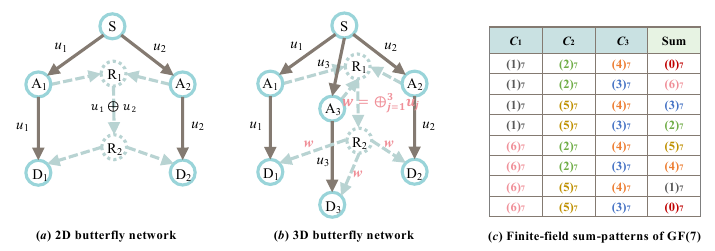}
  \caption{Butterfly networks. (a) 2-dimensional butterfly network; (b) 3-dimensional butterfly network; and (c) finite-field sum-patterns of $3$-user AI-EP code over GF($7$).} 
  \label{f.3D_network}
  \vspace{-0.3in}
\end{figure}

In this section, we investigate two overload EP codes based network FFMA systems.
One of the overload EP codes is the NO-AI-CWEP code $\Psi_{\rm no}$ constructed over GF($3^2$), which can support $3$-user with a loading factor of $\eta = 1.5$. 
The other overload EP code $\Psi_{\rm s}$ is constructed based on the prime field GF($7$), which is a symbol-wise AI-EP code.
The AI-EPs over GF($7$) are $C_1^{\rm s}= (1,6), C_2^{\rm s} = (2,5)$,
and $C_3^{\rm s} = (3, 4)$. 
The Cartesian product, i.e., $\Psi_{\rm s} = C_1^{\rm s} \times C_2^{\rm s} \times C_3^{\rm s}$, gives a 3-user AI-EP code. 
Since $\log_2\left(7-1\right)<3$, $\Psi_{\rm s}$ is not an UD-AIEP code, but is an overload EP code with LF of $R_q = 3$.

As known, network coding is one of the most important networking techniques, which can increase the throughput and reduce delay significantly \cite{R3}. The transmit multiple bit sequences can be encoded at the relay nodes, and the destination nodes can recover all the transmit bit sequences via algebraic decoding. 
Fig. \ref{f.3D_network} (a) shows the famous ``butterfly network'' which transmits 2-bit message from the source node to the two destination nodes.

To compare with the classical network coding, we present a 3-dimensional (3D) butterfly network in which the source node sends 3-bit message to 3-destinations as shown in Fig. \ref{f.3D_network} (b).

Let \((u_1, u_2, u_3)\) represent a codeword of the overload EP code \(\Psi\), i.e., \((u_1, u_2, u_3) \in \Psi\).  
We now describe the process of network FFMA.  
The source transmits a 3-bit message to the network.  
The \(j\)-th transmit node, denoted by \(A_j\), first maps the 1-bit message using the function \({\rm F}_{{\rm B}2q}\) to obtain \(u_j\), where \(u_j \in C_j\) and \(1 \leq j \leq 3\).  
We will now discuss the process for two overload EP codes, one at a time.  
Suppose the 3-bit message is \((0, 0, 0)_2\).

\subsection{Based on the NO-CWEP code $\Psi_{\rm no}$}
When the NO-AI-CWEP code ${\Psi}_{\rm no}$ constructed over GF($3^2$) is used,
we can obtain $u_1= (22)_3, u_2 = (12)_3$, and $u_3 = (02)_3$. 
When $u_j$ for $1\le j \le3$ arrives at the relay node $R_1$, it can form the FFSP block $w = \bigoplus_{j=1}^3u_j = (00)_3$, which is then transmitted to the relay node $R_2$. 
The relay node $R_2$ passes the FFSP block $w$ to the 3-destination nodes. 
At the destination node $D_j$, based on the FFSP block $w$ and transmit symbol $u_j$, it can recover the transmit 3-bit from the source node. 

Take the $1^{\rm st}$ destination node $D_1$ as an example, 
$D_1$ receives $u_1 = (22)_3$ and $w = (00)_3$. 
According to Table I, it is found that both the combinations of
$u_1=(22)_3, u_2 = (12)_3, u_3= (02)_3$ 
and $u_1=(11)_3, u_2 = (21)_3, u_3= (01)_3$ 
can make the FFSP block $w$ equal to $(00)_3$. 
Regarding as $u_1=(22)_3$, we find $u_1=(22)_3, u_2 = (12)_3, u_3= (02)_3$, and the transmit 3-bit message is $(0,0,0)_2$.
Thus, we can recover the transmit bits without ambiguity.

\subsection{Based on the AI-EP code $\Psi_{\rm s}$}
When the AI-EP code ${\Psi}_{\rm s}$ constructed over GF($7$) is used,
we can obtain $u_1=\left(1\right)_7, u_2=\left(2\right)_7$, and $u_3=\left(4\right)_7$. 
When $u_j$ for $1\le j \le3$ arrives at the relay node $R_1$, it can form the FFSP symbol $w=\bigoplus_{j=1}^3u_j = (1+2+4)_7 = (0)_7$, which is then transmitted to the relay node $R_2$. 
The relay node $R_2$ passes the FFSP symbol $w$ to the 3-destination nodes. 
At the destination node $D_j$, based on the FFSP symbol $w$ and transmit symbol $u_j$, it can recover the transmit 3-bit from the source node. 

Take the $1^{\rm st}$ destination node $D_1$ as an example, 
$D_1$ receives $u_1=\left(1\right)_7$ and $w = (0)_7$. 
By the decoding table as shown in Fig. \ref{f.3D_network} (c), 
it is found that both $\left(u_1,u_2,u_3\right)=\left(1,2,4\right)_7$ and $\left(u_1,u_2,u_3\right)=\left(6,5,3\right)_7$ can make the FFSP symbol $w$ equal to $(0)_7$. 
Regarding as $u_1=\left(1\right)_7$, it is derived that $\left(u_1,u_2,u_3\right)=\left(1,2,4\right)_7$, thus the transmit 3-bit message is $\left(0,0,0\right)_2$.

In summary, based on the proposed 3D butterfly network, we can decode the overload EP codes, i.e., the NO-CWEP code $\Psi_{\rm no}$ and the AI-EP code $\Psi_{\rm s}$, without ambiguity.

\vspace{-0.1in}
\section{Simulation results}
This section simulates the error performance of the proposed FFMA systems in a GMAC. We begin by introducing the configuration of the polarization-adjusted vector, followed by an analysis of the error performance of the FFMA systems, comparing them to classical CFMA systems such as NOMA, IDMA, and polar codes with spreading.

In the simulation, we construct three binary LDPC codes, denoted as ${\mathcal C}_{gc,b1}$, ${\mathcal C}_{gc,b2}$, ${\mathcal C}_{gc,b3}$, and one ternary LDPC code ${\mathcal C}_{gc,t}$. 
The code ${\mathcal C}_{gc,b1}$ is a binary $(400, 300)$ short LDPC code with a rate of $R_{gc,b1} = 0.75$, while ${\mathcal C}_{gc,b2}$ is a binary $(10000, 8400)$ long LDPC code with a rate of $R_{gc,b2} = 0.84$. 
The code ${\mathcal C}_{gc,b3}$ is a binary $(1000, 800)$ LDPC code with a rate of $R_{gc,b3} = 0.8$, while ${\mathcal C}_{gc,t}$ is a ternary $(2000, 1601)$ LDPC code with a rate of $R_{gc,t} \approx 0.8$. 
Note that ${\mathcal C}_{gc,b3}$ and ${\mathcal C}_{gc,t}$ have nearly the same rate, but the codeword length of ${\mathcal C}_{gc,t}$ is twice that of ${\mathcal C}_{gc,b3}$.

For the simulation, we default to the PA-FFMA configuration, where each FFMA system is assigned a corresponding power level. 

\vspace{-0.1in}
\subsection{Polarization Adjusted Vector}

To examine the polarization-adjusted vector (PAV), we consider two scenarios: the first involves FF-CCMA without channel coding, which is equivalent to the FF-TDMA mode, and the second involves FF-CCMA with channel coding. 
In addition, we adopt the concept of \textit{regular PAV} \cite{FFMA}, where all information bits are assigned a uniform power, and all parity bits are assigned another uniform power.

\subsubsection{Regular PA-FF-TDMA mode}
For the regular PA-FF-CCMA mode without channel coding (or the equivalent regular PA-FF-TDMA mode), power needs to be allocated to the information bits and the parity bits of the multiuser code ${\mathcal C}_{mc}$. We assume that the multiuser code ${\mathcal C}_{mc}$ is an $(m, M)$ linear block code with a loading factor $\eta = M/m$, and the length of the parity bits of ${\mathcal C}_{mc}$ is set to $Q = m - M$.

Let $K$ denote the number of information bits per user. The power allocated to each information symbol is represented by $\mu_{1} P_{avg}$, while the power allocated to each parity symbol is represented by $\mu_{2} P_{avg}$. Consequently, we redefine the PAV in its regular form (or regular PAV) as follows:
\begin{equation}
  \begin{array}{cc}
  {\mu}_{\rm reg}^{\rm td} = (\mu_{1}, \mu_{2}), \\
  \text{s.t.}, \quad C1: K \cdot \mu_{1} + Q \cdot \mu_{2} = m,
  \end{array}
\end{equation}
where the subscript ``reg'' and the superscript ``td'' in ${\mu}_{\rm reg}^{\rm td}$ denote ``regular'' and ``FF-TDMA mode'' (or FF-CCMA without channel-coding mode), respectively.
The condition $C1$ ensures that the total power remains constant at $m P_{avg}$. Note that the \textit{polarization-adjusted scaling factor (PAS)}, defined as ${\mu}_{\rm pas} = \mu_{1}/\mu_{2}$ in \cite{FFMA}, is also used to quantify the polarization feature.

In this paper, we assume that all unused power, i.e., $(M-K) P_{avg}$, is equally distributed among the information bits. We refer to this method as the \textit{Maximum Information Power (MIP)}, which was also introduced in \cite{FFMA}.
Therefore, the regular PAV is given by
\(
  {\mu}_{\rm reg}^{\rm td} = (\mu_{1}, \mu_{2}) 
                           = \left(\frac{M}{K}, 1\right),
\)
which implies that ${\mu}_{\rm pas} = \mu_{1}/\mu_{2} = \frac{M}{K}$.

Consider that the EP code (or multiuser code) is constructed based on ${\mathcal C}_{gc,b1}$, which is a $(400, 300)$ binary LDPC code. In this case, the unused power is given by $(300 - K) P_{avg}$. Therefore, the regular PAV is ${\mu}_{\rm reg}^{\rm td} = \left(\frac{300}{K}, 1\right)$.
For example, when $K = 10$, the regular PAV is ${\mu}_{\rm reg}^{\rm td} = (30, 1)$. Similarly, when $K = 100$, the corresponding regular PAV is ${\mu}_{\rm reg}^{\rm td} = (3, 1)$.

\subsubsection{Regular PA-FF-CCMA mode}

In the regular PA-FF-CCMA mode with channel coding, power must be allocated to both the information section and the parity section of the channel code ${\mathcal C}_{gc}$, where the information section is determined by the multiuser code ${\mathcal C}_{mc}$ (or the EP code $\Psi_{\rm cw}$). The parameters of the multiuser code ${\mathcal C}_{mc}$ are the same as those defined previously. Additionally, we assume that the number of bits per user, $K$, is equal to or smaller than $M$, and that the EP encoder operates in parallel mode.
Suppose the global channel code ${\mathcal C}_{gc}$ is an $(N, K_{gc})$ linear block code with a rate of $R_{gc} = \frac{K_{gc}}{N}$. Let the length of the parity bits of ${\mathcal C}_{gc}$ be $R = N - K_{gc}$.

In this context, let $\mu_{1}$ and $\mu_{2}$ represent the polarization-adjusted factors for the information symbol and the parity symbol of the multiuser code, respectively. Additionally, let $\mu_{c}$ denote the polarization-adjusted factor for each parity symbol of the channel code. 
Thus, the regular PAV is expressed as:
\begin{equation}
  \begin{array}{cc}
  {\mu}_{\rm reg}^{\rm cc} = (\mu_{1}, \mu_{2}, \mu_{c}), \\
  \text{s.t.}, \quad C2: K \cdot \mu_{1} + Q \cdot \mu_{2}
                        + R \cdot \mu_{c} = N,
  \end{array}
\end{equation}
where the superscript ``cc'' in ${\mu}_{\rm reg}^{\rm cc}$ denotes the ``FF-CCMA with channel coding mode''.
The condition $C2$ ensures that the total power remains constant at $N P_{avg}$.

In this paper, we consider two power allocation methods. The first method is the aforementioned MIP power allocation scheme, where we equally distribute all unused power, i.e., $(K_{gc}-K-Q)P_{avg}$, among the information bits.
Therefore, the regular PAV is given by:
\[
  {\mu}_{\rm reg, pll}^{\rm cc} = (\mu_{1}, \mu_{2}, \mu_{c}) 
                           = \left(\frac{K_{gc}-Q}{K}, 1, 1\right),
\]
which implies that the regular PAV for the multiuser code ${\mathcal C}_{mc}$ is ${\mu}_{\rm reg}^{\rm td} = \left(\frac{K_{gc}-Q}{K}, 1\right)$.

Now, consider the channel code constructed based on the ${\mathcal C}_{gc,b2}$, which is a $(10000, 8400)$ binary LDPC code. In this case, the unused power is equal to $(8400-100-K)P_{avg}$. 
For instance, when $K = 10$, the regular PAV is ${\mu}_{\rm reg, pll}^{\rm cc} = (830, 1, 1)$, whereas when $K = 100$, the regular PAV is ${\mu}_{\rm reg, pll}^{\rm cc} = (83, 1, 1)$.

In the second case, the unused power is allocated to both the information bits and the parity bits of the multiuser code ${\mathcal C}_{mc}$. We refer to this method as the \textit{Maximum Block Information Power (MBIP)} scheme. The MBIP method consists of two phases.

\textbf{Phase One:} The information section power, equal to $K_{gc} \cdot P_{avg}$, is evenly allocated to the occupied data blocks of one user. Since the EP encoder works in parallel mode, each user occupies one data block. Therefore, the data block is assigned power equal to $\frac{K_{gc}}{m} \cdot P_{avg}$.

\textbf{Phase Two:} The power allocated to the information section within a block is equally distributed among the information bits, following the MIP method. Since the EP encoder operates in parallel mode, each information bit is assigned a power of $\frac{K_{gc}}{m} \cdot \frac{M}{K} \cdot P_{avg}$.

In summary, the regular PAV for the FF-CCMA with an EP encoder operating in parallel mode is expressed as:
\[
  {\mu}_{\rm reg, pll}^{\rm cc} = (\mu_{1}, \mu_{2}, \mu_{c}) 
            = \left(\frac{M K_{gc}}{Km}, \frac{K_{gc}}{m}, 1\right).
\]

Now, consider the case where the EP code is based on ${\mathcal C}_{gc,b1}$, and the channel code is based on ${\mathcal C}_{gc,b2}$. In this case, when $K = 10$, the regular PAV becomes ${\mu}_{\rm reg, pll}^{\rm cc} = (630, 21, 1)$, and when $K = 100$, it becomes ${\mu}_{\rm reg, pll}^{\rm cc} = (63, 21, 1)$.

It should be noted that when the EP encoder operates in serial mode, the regular PAV must be updated accordingly.
For the MIP method, the regular PAV for the FF-CCMA with an EP encoder operating in serial mode is updated as:
\(
{\mu}_{\rm reg, sel}^{\rm cc} = (\mu_{1}, \mu_{2}, \mu_{c}) 
= \left( \frac{(K_{gc} - KQ)}{K}, 1, 1 \right).
\)
For the MBIP method, the regular PAV is updated as:
\(
{\mu}_{\rm reg, sel}^{\rm cc} = (\mu_{1}, \mu_{2}, \mu_{c}) 
= \left( \frac{M K_{gc}}{Km}, \frac{K_{gc}}{Km}, 1 \right),
\)
where the subscript ``sel'' stands for the serial mode. The derivation for the serial mode is not elaborated here, as the analysis process is analogous to that of the parallel mode.

\subsection{Error Performance of FF-CCMA Systems}
We first evaluate the error performance of two decoding algorithms for the FF-CCMA system: the MSA algorithm and the BMD algorithm. The MSA decoding is performed with $50$ iterations, while the BMD algorithm, introduced in \cite{FFMA}, will be further discussed in detail in our other papers.
Next, we analyze the error performance of the encoder in both serial and parallel modes, considering different power allocation methods.
During the simulation, the EP code $\Psi_{\rm cw}$ for the FF-CCMA system is constructed using the short LDPC code ${\mathcal C}_{gc,b1}$. The channel coding is considered in two scenarios: one without any channel coding and the other where the channel code employs the long LDPC code ${\mathcal C}_{gc,b2}$.

\begin{figure*}[t!]
  \centering
  \subfigure[FF-CCMA without channel coding.]{\includegraphics[width=0.46\textwidth]{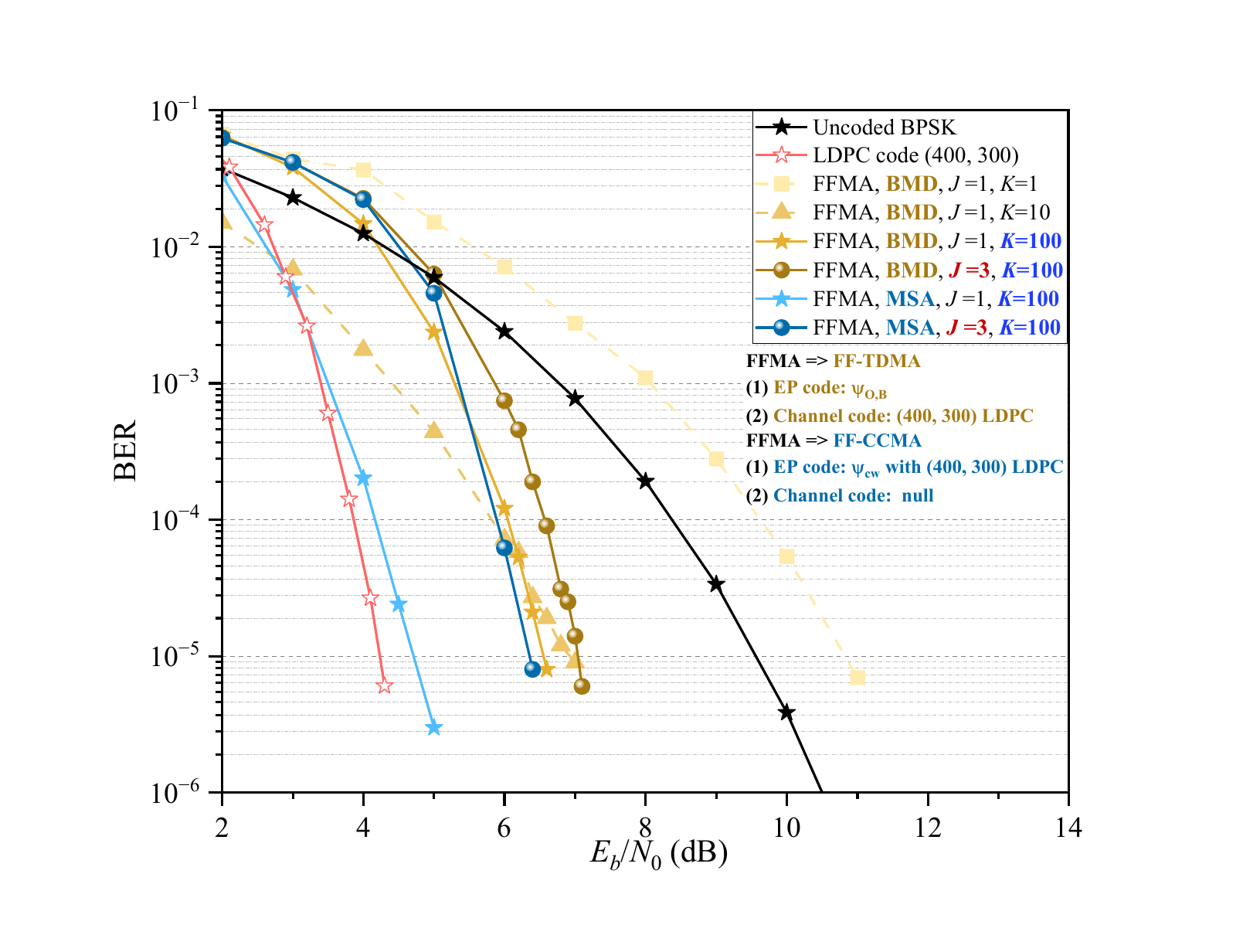}}
  \label{CCMA_sub1}
  \subfigure[FF-CCMA with channel coding.]{\includegraphics[width=0.47\textwidth]{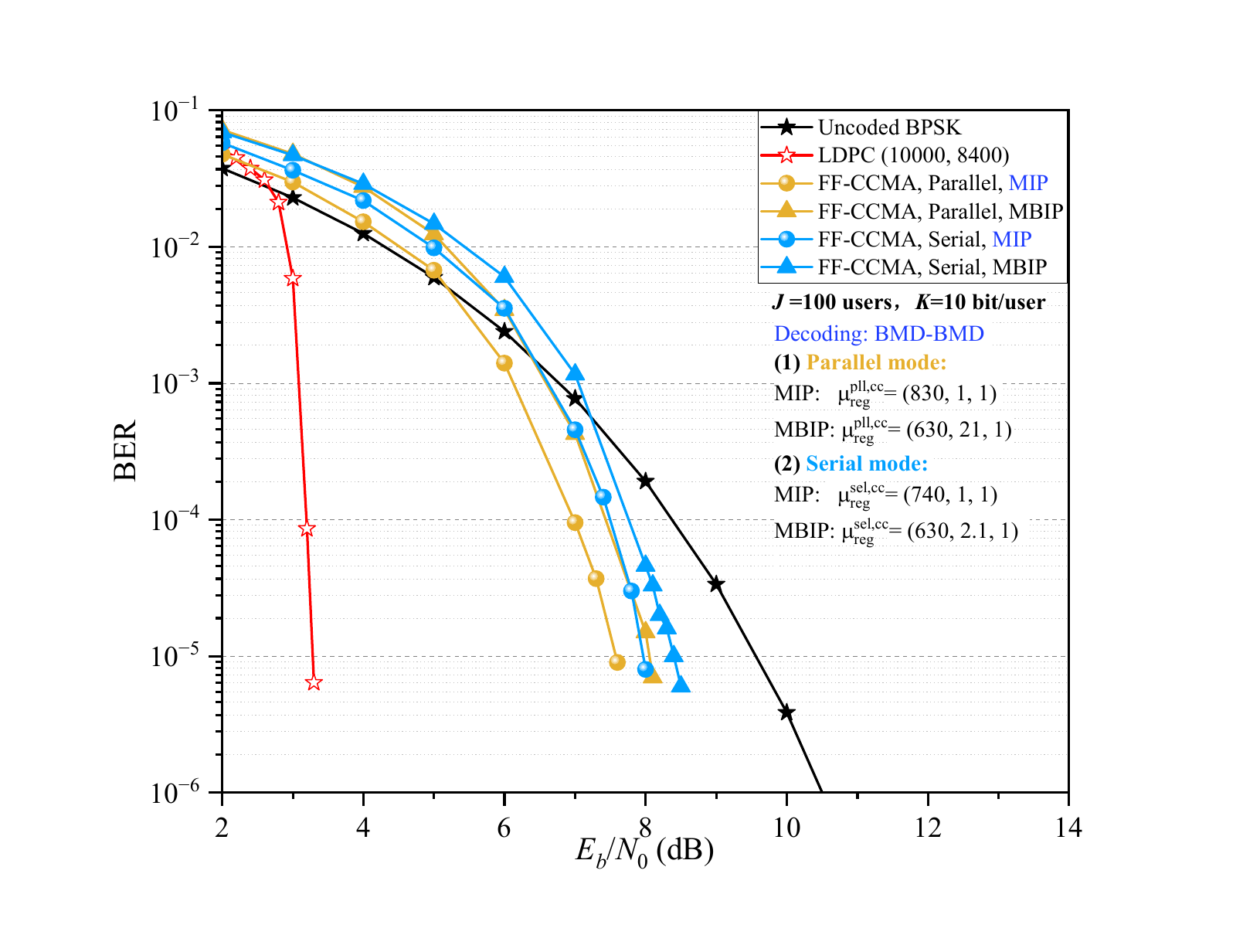}}
  \label{CCMA_sub2}
  \caption{BER performances of FF-CCMA systems in a GMAC.}
  \label{f.CCMA}
  \vspace{-1cm}
\end{figure*}

Fig. \ref{f.CCMA} (a) illustrates the error performance of the FF-CCMA system without channel coding, which is equivalent to the FF-TDMA system. In the simulation, we set the number of users as $J = 1, 3$, and the number of bits per user as $K = 1, 10, 100$. The MIP power allocation scheme is explored.
From Fig. \ref{f.CCMA} (a), as $K$ increases from $1$, $10$ to $100$, the BER performance of the BMD algorithm first improves, then deteriorates. However, the BMD algorithm remains applicable for all values of $K$. Additionally, when $J = 1$ and $K = 100$, the FF-TDMA system behaves like a $0.25$ rate PA-LDPC code. In this case, the MSA algorithm provides significantly better BER performance than the BMD algorithm, and the BER performance of the MSA algorithm is even comparable to that of the LDPC code ${\mathcal C}_{gc, b1}$, which has a rate of $0.75$.
When $J = 3$ and $K = 100$, the MSA algorithm also outperforms the BMD algorithm in terms of BER performance.


The BER performances of the FF-CCMA with different EP encoders are shown in Fig. \ref{f.CCMA} (b), where we evaluate both the serial and parallel modes of the EP encoder. In the simulation, we set $J = 100$ and $K = 10$, and use the BMD-BMD algorithm to decode the FF-CCMA system. Both the MIP and MBIP schemes are explored for power allocation.
From Fig. \ref{f.CCMA} (b), it is clear that the EP encoder in parallel mode consistently provides better BER performance than the serial mode. Therefore, when $K$ is smaller than $M$, parallel mode is preferred over serial mode. However, when $K$ is large, for instance, when $K > M$, both parallel and serial modes must be considered simultaneously. This combined approach, known as the \textit{hybrid mode}, will be further explored in our future work.

Besides, it can be observed that the MIP scheme provides better BER performance than the MBIP scheme. However, this result is not always conclusive, as the error performance is influenced by various factors, including the simulation parameters $J$, $K$, the LDPC codes ${\mathcal C}_{gc, b1}$ and ${\mathcal C}_{gc, b2}$, and even the decoding algorithms used.
In this paper, we focus on the short packet scenario where $K < M$ and the BMD-BMD (or MPA-BMD) decoding algorithm is employed. Therefore, we use the MIP scheme in the following discussion.



\begin{figure*}[t!]
  \centering
  \subfigure[$R_{sum} = 0.75$.]{\includegraphics[width=0.46\textwidth]{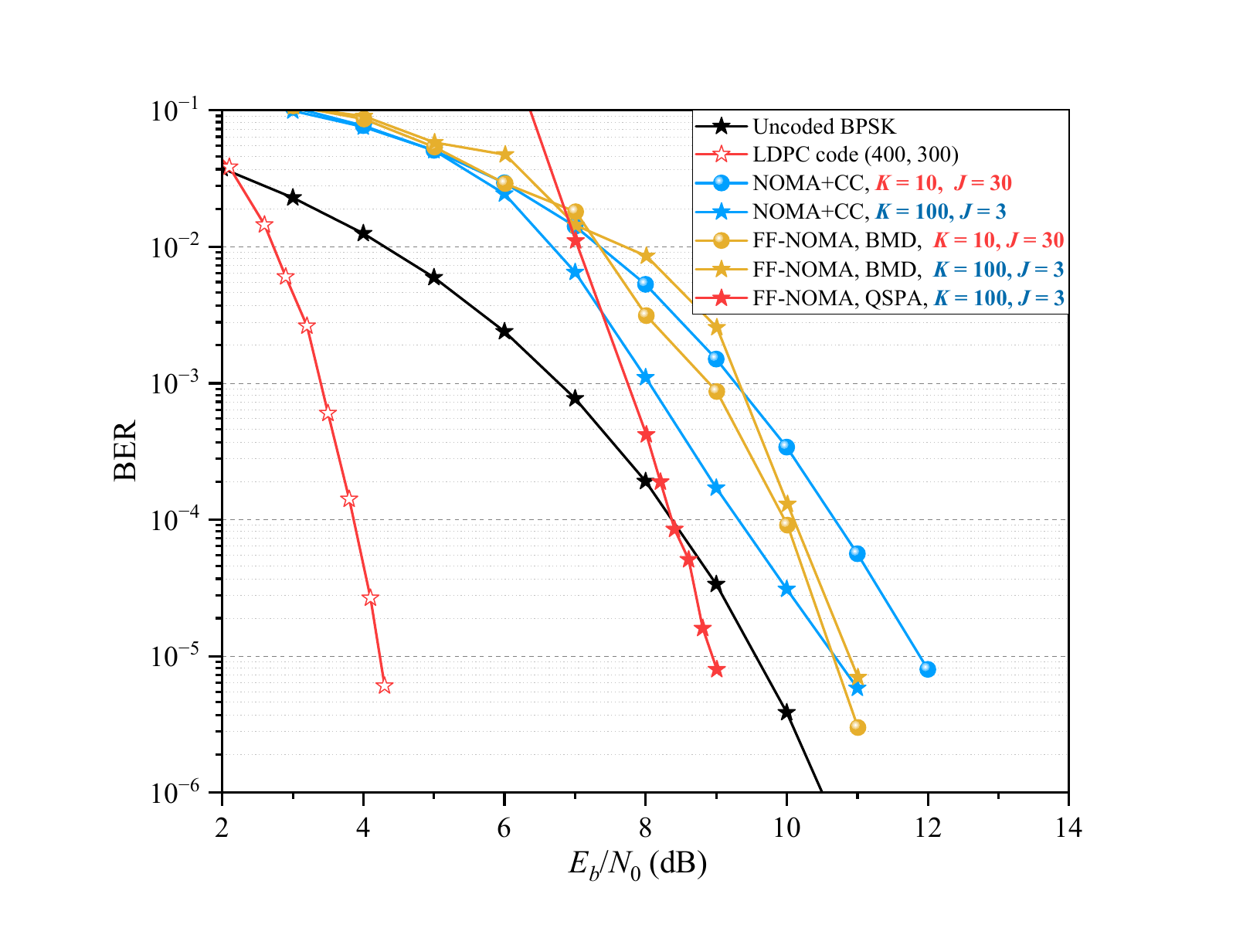}}
  \label{NOMA_sub1}
  \subfigure[$R_{sum} = 1.2$.]{\includegraphics[width=0.47\textwidth]{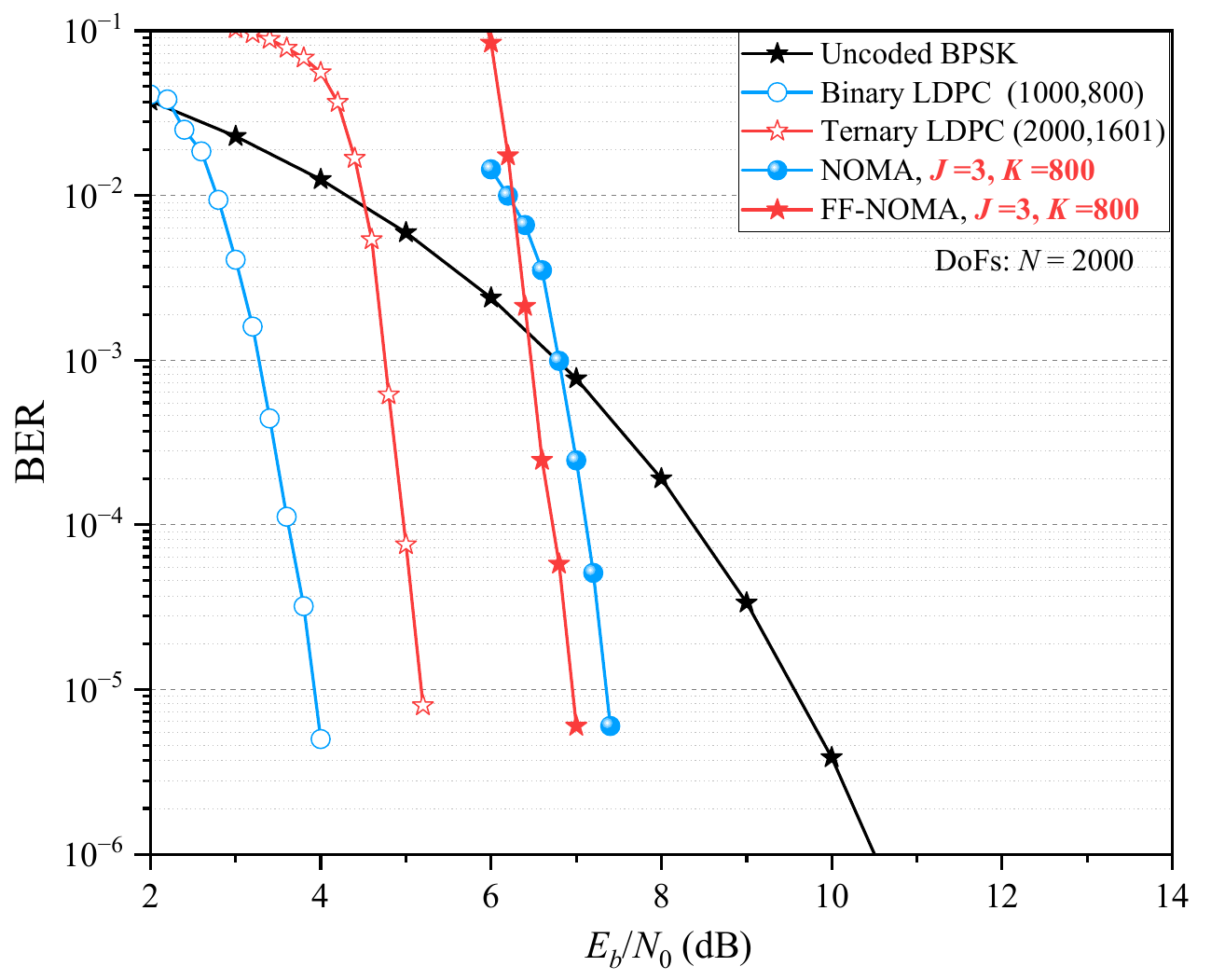}}
  \label{NOMA_sub2}
  \caption{BER performance of FF-NOMA and NOMA systems in a GMAC.}
  \label{f.NOMA}
  \vspace{-1cm}
\end{figure*}

\subsection{Error Performance of FF-NOMA Systems}

In this subsection, we compare the proposed FF-NOMA system with the non-orthogonal spreading NOMA system. To evaluate the error performance of the FF-NOMA system, we use the NO-CWEP code $\Psi_{\rm no}$ with a loading factor of $1.5$, as described in Section V, followed by a channel code. For a fair comparison, the non-orthogonal spreading NOMA system first applies a channel code, followed by a ternary uniquely decodable code (T-UDC) \cite{UD_CDMA1_2012, UD_CDMA2_2012}. The T-UDC supports $3$ users with a spreading factor (SF) of $2$, so its loading factor is also $1.5$. These parameters are the same as our proposed NO-AI-CWEP code.

In our proposed FF-NOMA system, two decoding algorithms are considered: the BMD algorithm and the $3$-ary sum-product algorithm (QSPA), implemented with $50$ iterations. The $3$-ary QSPA will be discussed in more detail in a separate paper. For power allocation, the MIP method is applied.

In Fig. \ref{f.NOMA} (a), assume the total number of degrees of freedom (DoFs) is $N = 400$. The FF-NOMA system utilizes the short binary LDPC code ${\mathcal C}_{gc, b1}$ with a rate of $0.75$, while the NOMA system employs a convolutional code (CC) with a rate of $0.75$. We consider two sets of parameters: one where the number of users is $J = 3$ and the number of bits is $K = 100$, and another where $J = 30$ and $K = 10$. With $N = 400$, the total sum rate is given by $R_{\text{sum}} = \frac{JK}{N} = 0.75$.

When $K = 100$ and $J = 3$, it is observed that the error performance of FF-NOMA with the QSPA algorithm significantly outperforms that of the BMD algorithm in terms of BER. Since the QSPA algorithm is a soft-decision algorithm, it generally provides better error performance compared to the hard-decision BMD algorithm. This behavior of the FF-NOMA decoding algorithms is consistent with that of the FF-CCMA systems. However, when $K = 10$ and $J = 30$, the QSPA algorithm becomes less effective due to the loss of high-reliability information bits. In this case, the BMD algorithm is preferred, as is the case with the FF-CCMA system.

Additionally, when $K = 10$, $J = 30$, and $\text{BER} = 10^{-5}$, our FF-NOMA system with the BMD algorithm provides a gain of more than $1$ dB, demonstrating superior BER performance compared to the NOMA with CC system. On the other hand, when $K = 100$ and $J = 3$, although the FF-NOMA with the BMD algorithm results in worse BER performance than the NOMA with CC system, the FF-NOMA with the QSPA algorithm still offers a coding gain of approximately $1.2$ dB, which significantly outperforms the NOMA with CC system in terms of BER. Therefore, the simulation results indicate that when $R_{\text{sum}} < 1$, regardless of whether the number of users is small ($J = 3$) or large ($J = 30$), our proposed FF-NOMA system consistently delivers better performance than the classical NOMA system.

Fig. \ref{f.NOMA} (b) illustrates the error performance for the case where the sum rate is $R_{\text{sum}} = 1.2$, which exceeds $1$, in contrast to Fig. \ref{f.NOMA} (a), where the sum rate is $R_{\text{sum}} = 0.75$, less than $1$. Assume the total number of DoFs is $N = 2000$. The FF-NOMA system uses a ternary LDPC code ${\mathcal C}_{gc, t}$ with parameters $(2000, 1601)$, offering an approximate loading factor of $0.8$. In contrast, the NOMA system employs a binary LDPC code ${\mathcal C}_{gc, b3}$ with parameters $(1000, 800)$ and a loading factor of $0.8$. We set $J = 3$ and $K = 800$, yielding a total sum rate given by $R_{\text{sum}} = \frac{JK}{N} = 1.2$.

From Fig. \ref{f.NOMA} (b), it can be observed that when the payload is relatively large, for example, $K = 800$, both the FF-NOMA and NOMA systems achieve a coding gain compared to the uncoded case. This is in contrast to the scenario with a smaller payload, where $K = 10$ or $100$, which demonstrates that short payload transmissions present a more challenging problem. Additionally, when $\text{BER} = 10^{-5}$, our proposed FF-NOMA system provides approximately a $0.4$ dB gain compared to the classical NOMA system, highlighting the efficiency of our FF-NOMA approach. Interestingly, although the error performance of a ternary LDPC code is worse than that of a binary LDPC code due to the 3ASK modulation, the resulting FF-NOMA system still outperforms the NOMA system in terms of BER. This phenomenon will be further investigated in the following paper.

\begin{figure*}[t!]
  \centering
  \subfigure[$J = 100/99$ and $K =10$.]{\includegraphics[width=0.47\textwidth]{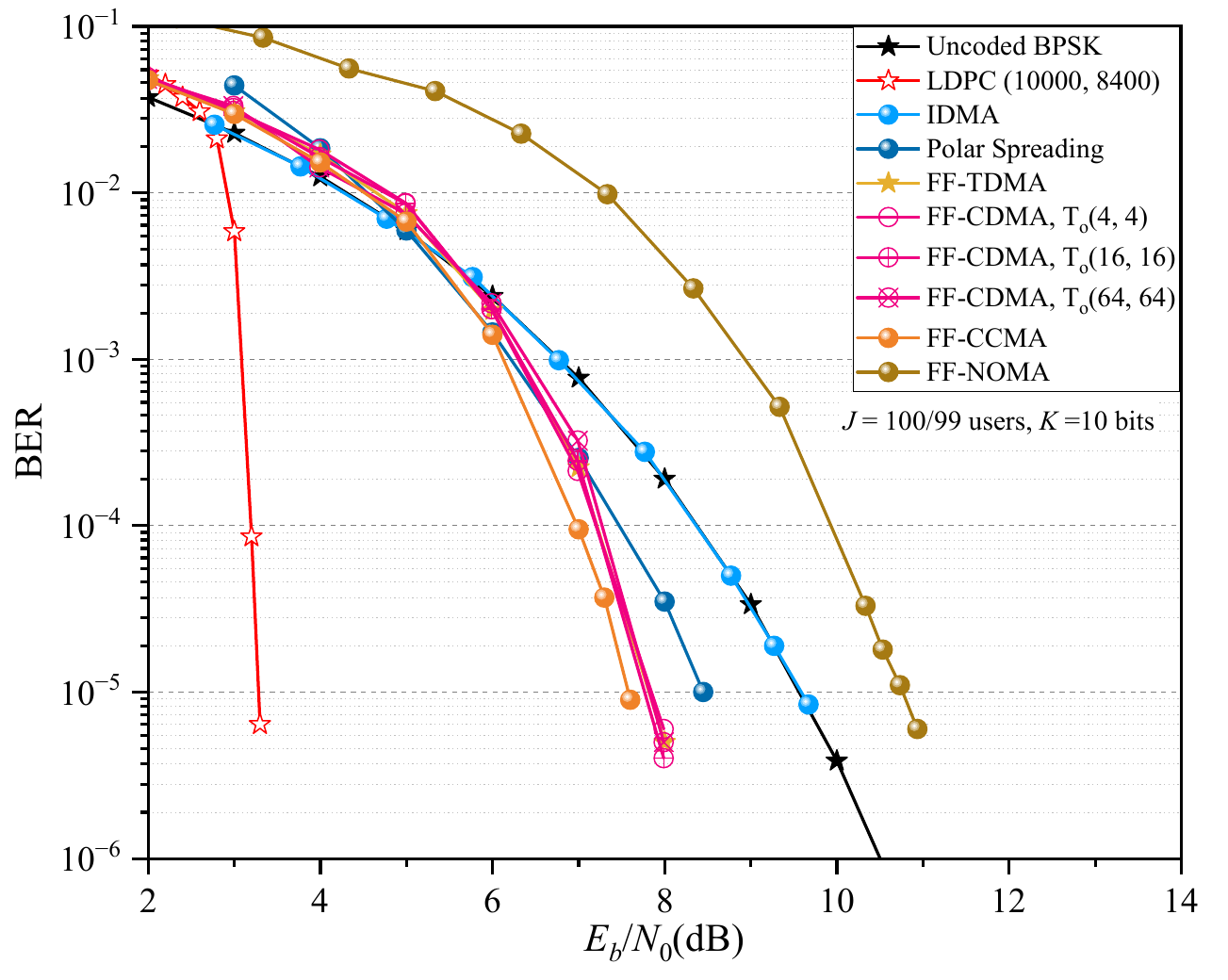}}
  \subfigure[$J = 63$ and $K = 100$.]{\includegraphics[width=0.47\textwidth]{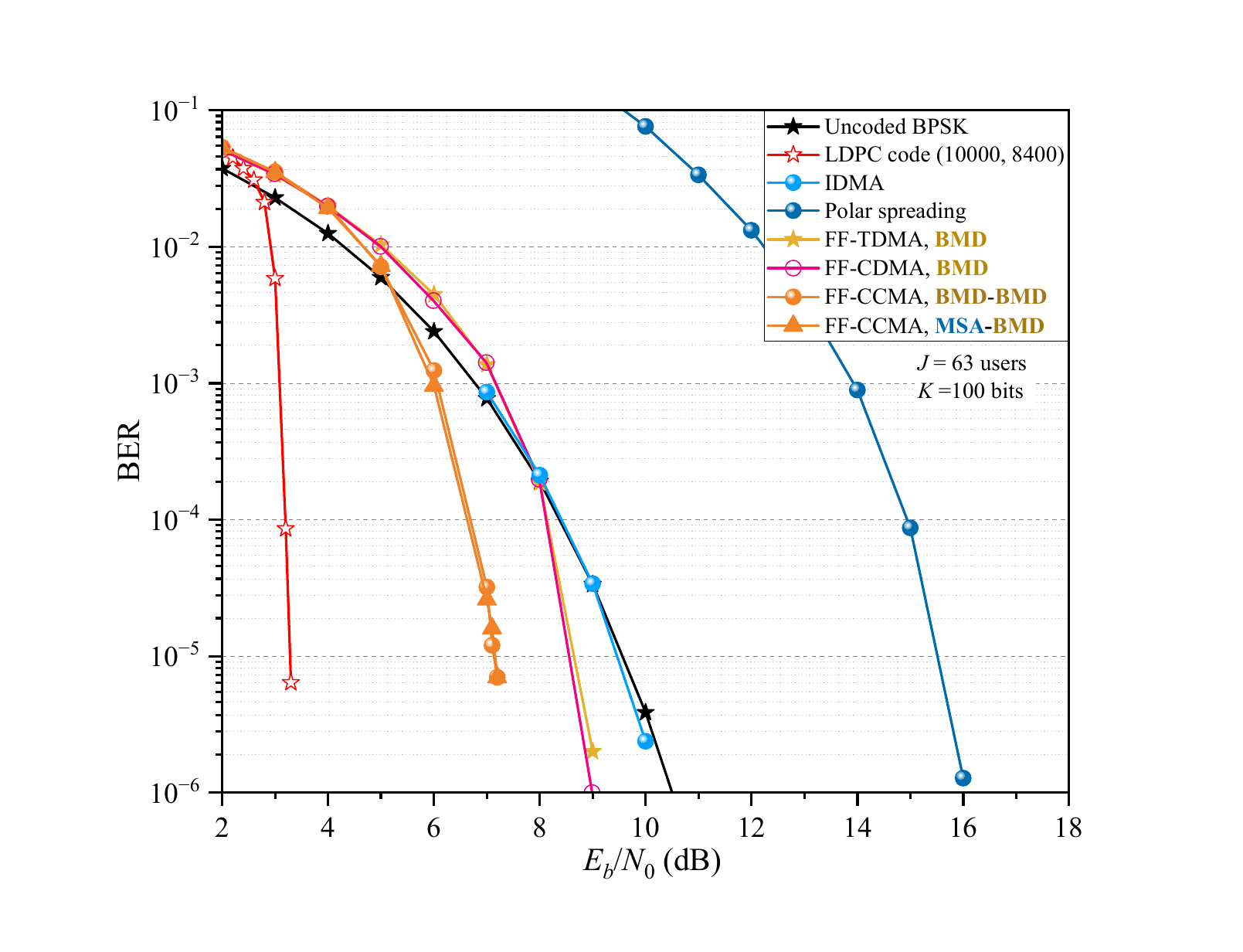}}
  \caption{BER performances among different MA Systems.}
  \label{f.Various_MA}
  \vspace{-1cm}
\end{figure*}

\subsection{Comparasions among Different MA Systems}
In this subsection, we compare the proposed FFMA systems with IDMA and polar spreading systems. The simulation conditions for IDMA and polar spreading are consistent with those presented in \cite{FFMA}. 

For the IDMA system, we utilize a concatenation code consisting of two components: a channel code (e.g., conventional code), followed by a repetition code \cite{IDMA1,IDMA2}. In general, the repetition code is employed to enhance the signal's energy and reduce multiuser interference, while the conventional code improves error performance, making the required $E_b/N_0$ closer to the Shannon limit. 
To ensure a fair comparison with the FFMA system, during simulations, the IDMA system is configured with $N = 10000$ resources, using a rate $0.84$ convolutional encoder that outputs $12$-bit codewords, and the spreading factor is increased to $840$. 
For the polar random spreading system \cite{Polar2}, we utilize a CA-polar code with a CRC length of $2$. The spreading sequence is generated using a Gaussian random sequence. The total frame length of the CA-polar spreading systems is $9984$.
The iterative decoding process of the polar spreading system combines successive interference cancellation (SIC), MMSE estimation, and single-user decoding.

For the FF-CDMA mode, we consider three AI-CWEP codes, constructed based on ${\bf T}_{\rm o}(4,4)$, ${\bf T}_{\rm o}(4,4)$, ${\bf T}_{\rm o}(16,16)$, and ${\bf T}_{\rm o}(64,64)$, respectively. Additionally, the FF-CDMA system typically employs two types of transform functions corresponding to BPSK and 3ASK modulations. Let the average power of BPSK and 3ASK be denoted as $P_{\text{avg}, b}$ and $P_{\text{avg}, t}$, respectively. To maintain a constant average power per symbol, $P_{\text{avg}}$, we must set $P_{\text{avg}, b} = P_{\text{avg}, t}$. Moreover, the power allocation follows the MIP scheme.

Fig. \ref{f.Various_MA} (a) illustrates the error performance of various MA systems under the conditions of $K = 10$ and $J = 100$. First, we observe that the proposed FF-CDMA, FF-TDMA, and FF-CCMA systems all outperform the classical IDMA and polar spreading systems in terms of error performance, highlighting the effectiveness of our proposed FFMA systems.
Next, we note that the BER performance of the three PA-FF-CDMA systems is identical, which is also the case for the PA-FF-TDMA system. This similarity arises from the MIP power allocation scheme used, which ensures that both the PA-FF-TDMA and the three PA-FF-CDMA systems receive equal power in the complex field, resulting in the same error performance across these FFMA systems.
Finally, the proposed FF-CCMA system shows slightly better error performance than both the FF-CDMA and FF-TDMA systems.

Fig. \ref{f.Various_MA} (b) investigates the error performance of various MA systems under the conditions of $K = 100$ and $J = 63$. Compared to the case with $K = 10$ shown in Fig. \ref{f.Various_MA} (a), we observe that as $K$ increases, the proposed FFMA system provides even better BER performance than the classical CFMA systems. This is due to the fact that, as $K$ increases, the CFMA systems suffer from the limitations imposed by their finite blocklength, which degrades their error performance. In contrast, the error performance of the proposed FFMA system improves.
Furthermore, when $K = 100$, both MSA-BMD and BMD-BMD decoding methods can be used for the FF-CCMA systems. From Fig. \ref{f.Various_MA} (b), it is found that both decoding algorithms yield nearly identical BER performance.

\section{Conclusion and Remarks}

In this paper, we introduce codeword-wise EP codes, and construct two specific codeword-wise EP codes, which are S-CWEP codes and AI-CWEP codes.
To make the designed codeword-wise EP codes uniquely decodable, USPM structural property constraints are presented for both the S-CWEP codes over GF($2^m$) and AI-CWEP codes over GF($3^m$).

Then, we construct UD-S-CWEP codes ${\Psi}_{\rm cw}$ based on linear block channel codes, which are used to support FF-CCMA. In addition, we introduce $\kappa$-fold orthogonal ternary orthogonal matrix ${\bf T}_{\rm o}(2^{\kappa}, 2^{\kappa})$ and ternary non-orthogonal matrix ${\bf T}_{\rm no}(M,m)$. Based on the $\kappa$-fold orthogonal ternary matrix ${\bf T}_{\rm o}(2^{\kappa}, 2^{\kappa})$ and its additive inverse matrix ${\bf T}_{\rm o, ai}(2^{\kappa}, 2^{\kappa})$, we can construct generalized UD-AI-CWEP codes $\Psi_{\rm ai,T}$ over GF($3^m$) where $m = 2^{\kappa}$, which is used to realize FF-CDMA. Based on the ternary non-orthogonal matrix ${\bf T}_{\rm no}(M,m)$ and its additive inverse matrix ${\bf T}_{\rm no, ai}(M,m)$, we construct NO-AI-CWEP codes $\Psi_{\rm no}$, which can be regarded as FF-NOMA. By using these codeword-wise EP codes, we can enlarge the applications of FFMA and support different modes of FFMA systems, such as FF-CCMA, FF-CDMA, and FF-NOMA.
Simulation results show that all these four modes can support massive users transmission with well-mannered error performances.

Since the UD-AI-CWEP code $\Psi_{\rm ai,T}$ and NO-CWEP code $\Psi_{\rm no}$ are constructed over GF($3^m$), the transform function ${\rm F}_{\rm F2C}$ exploits 3ASK. By increasing the modulation order, e.g., from BPSK to 3ASK, more information is loaded through the modulated symbol. Thereby, we see that an increase in dimensionality (e.g., from BSSK signaling to 3ASK signaling) can provide a solution for implementing error-correction orthogonal/non-orthogonal spreading codes. 

In the future, we will continue exploring different codeword-wise EP codes to support other classical CFMA techniques, such as IDMA, polarized EP code for FFMA system and others.



\vfill
\end{document}